%% file: superresposcont.tex
\documentclass[11pt]{article}
\pdfoutput=1
\usepackage{float}
\usepackage{setspace}
\usepackage{appendix}
\usepackage{xfrac}
\usepackage{amssymb,amsfonts}
\usepackage[figuresright]{rotating}
\usepackage{amsmath,amssymb}
\usepackage{dcolumn}
\usepackage{bm}
\usepackage{amscd,amsthm}
\usepackage{ifthen}
\usepackage{mathtools}
\usepackage{graphicx}
\usepackage{caption}
\usepackage{subcaption}
\usepackage{pgf,tikz}
\usepackage{pgfplots}
\usetikzlibrary{calc,shadows}
\usetikzlibrary{pgfplots.groupplots}
\usepackage{url}
\usepackage{siunitx}

\usepackage{amsthm}
\usepackage{mathtools}

\usepackage[nosort]{cite}
\usepackage{url}
\usepackage{bbm}
\usepackage{paralist}
\usepackage{caption}
\usepackage[printonlyused]{acronym}
\usepackage{xcolor}

\newlength{\figwidth}
\pgfplotsset{every non boxed x axis/.append style={x axis line style=->,>=latex},every non boxed y axis/.append style={y axis line style=->,>=latex}}

\include{math-macros}

\include{defs}

\usepackage{hyperref}

\hypersetup{
    unicode=false,          
    pdftoolbar=true,        
    pdfmenubar=true,        
    pdffitwindow=false,     
    pdfstartview={FitH},    
    pdfauthor={Veniamin I. Morgenshtern},     
    pdfcreator={Veniamin I. Morgenshtern},   
    pdfproducer={Producer}, 
    pdfnewwindow=true,      
    colorlinks=true,       
     linkcolor=blue!60!black,          
    citecolor=blue!60!black,        
    filecolor=blue!60!black,      
    urlcolor=blue!60!black,
     hypertexnames=false
}
\usepackage{autonum}
\usepackage{enumitem}

\begin{document}
\acrodef{PSF}{point-spread function}

\title{Super-Resolution of Positive Sources on an Arbitrarily Fine Grid}

\pdfinfo{
/Title  (Super-Resolution of Positive Sources on an Arbitrarily Fine Grid)
/Author (Veniamin~I.~Morgenshtern)
/Keywords ()
}

\author{
\parbox{\linewidth}{\centering
Veniamin~I.~Morgenshtern\\\vspace{0.3cm}
Chair of Multimedia Communications and Signal Processing,\\ University of Erlangen-Nuremberg,\\ 91058 Erlangen, Germany
}
}

\maketitle
\begin{abstract} {In super-resolution it is necessary to locate with high precision point 
sources from noisy observations of the spectrum of the signal at low frequencies capped by $\flo$. 
In the case when the point sources are positive and are located 
on a grid, it has been recently established that the super-resolution problem can be solved via 
linear programming in a stable manner and that the method is nearly optimal in the minimax sense.
The quality of the reconstruction critically depends on the 
Rayleigh regularity of the support of the signal; that is, on the maximum number of sources that 
can occur within an interval of side length about $1/\flo$.
This work extends the earlier result 
and shows that the conclusion continues to hold when the locations of the point sources are arbitrary, 
i.e., the grid is arbitrarily fine. The proof relies on new interpolation constructions in Fourier analysis.}
\end{abstract}


\section{Introduction}
The super-resolution problem of positive sources (see \fref{fig:conv}) consists 
of recovering a high-frequency signal
\begin{equation}
    x(w) = \sum_i x_i \dirac(w - w_i)
    \label{eq:spikescont}
\end{equation}
consisting of positive point sources (\emph{spikes}, for short) located at unknown positions $w_i\in [0, 1)$ 
and of unknown intensity $x_i>0$; $\dirac(\cdot)$ is the Dirac delta function.
The signal is observed through a convolution measurement of the form
\begin{equation}
    s(v) = \int \klo (v - w) x(w)d w + z(v),
    \label{eq:convcont}
\end{equation}
where $\klo(\cdot)$ is a low-frequency kernel that erases the high-frequency components 
of the signal and $z(\cdot)$ is noise.

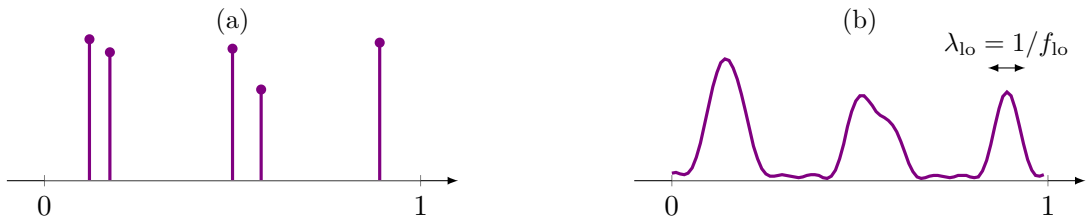
\begin{figure}[ht]
\begin{subfigure}{0.5\textwidth}
\centering
    \caption{}\vspace{-0.3cm}
\begin{tikzpicture}
\begin{axis}[
    width=6cm,scale only axis,height=2cm,
    name= plot1,
    ytick=\empty,
    xtick={0,1},
    axis x line = middle,
    name= plot2,
    every axis title/.style={at={(0.5,1)},above,yshift=-0.2cm},
    hide y axis,
    ymin=0,
    ymax=1,
    xmin=-0.1,
    xmax=1.1]
    \addplot[ycomb,range=0:1,line width=1.2pt,violet,mark options={fill=violet,scale=1},mark=*, mark size=1.3pt]
         table[x index=0,y index=1]{./figures/figspikes.txt};
  \draw[help lines] (axis cs:0,0) -- (axis cs:1,0);
\end{axis}
\end{tikzpicture}
\end{subfigure}
\begin{subfigure}{0.5\textwidth}
    \centering
    \caption{}\vspace{-0.3cm}
    \begin{tikzpicture}
    \begin{axis}[
       width=6cm,scale only axis,height=2cm,
    name= plot1,
        ytick=\empty,
    xtick={0,1},
    axis x line = middle,
    name= plot2,
    every axis title/.style={at={(0.5,1)},above,yshift=-0.2cm},
    hide y axis,
    ymin=0,
    ymax=0.2,
    xmin=-0.1,
    xmax=1.1]
    \addplot[domain=0:1, line width=1.2pt,violet]
             table[x index=0,y index=4]{./figures/figrclass.txt};
      \draw[help lines] (axis cs:0,0) -- (axis cs:1,0);
      \draw[<->,>=latex] (axis cs:0.84,0.15) -- (axis cs:0.94,0.15) node[above,midway] {\small $\llo=1/\flo$};
    \end{axis}
    \end{tikzpicture}
\end{subfigure}
\vspace{-0.3cm}
\caption{Microscope as a low-pass filter: signal (a) and convolution measurement (b).}
\label{fig:conv}
\end{figure}

This problem arises in single-molecule super-resolution microscopy~\cite{betzig06-06,dickson97on,klar00fluor}. 
In this application, $w_i$'s encode the unknown locations of fluorescent molecules, $x_i$ is 
proportional to the number of photons emitted by the $i$th molecule during the observation time. 
Crucially, the number of photons is a nonnegative number, leading to the assumption $x_i>0$, 
which makes the problem much simpler. Assume that light of wavelength $\llo$ is emitted 
by the molecules. Due to diffraction of light, the high-frequency spacial details of the 
signal are destroyed, no matter how perfect or large the microscope is. 
At the detector we record a blurred version of the signal, no 
details smaller than about $\llo$ are visible. To restate this mathematically: the function
$\klo(\cdot)$ models the \ac{PSF} of the microscope; due to diffraction of 
light the \ac{PSF} is band-limited to $\flo \define 1 /\llo$. The noise $z(v)$ represents all sources 
of noise in the system. For example, the thermal noise at the detector, the Poisson quantum 
mechanical noise due to photon quantization in low-intensity imaging, and the noise originating 
from the imperfect knowledge of the \ac{PSF} in the optical system. We refer the interested reader 
to~\cite{candes-14}, where the connection to super-resolution microscopy is worked out in details.

\subsection{Discrete model}
In the earlier work~\cite{candes-14} a discrete analog of the model in~\fref{eq:spikescont} and~\fref{eq:convcont} 
has been considered. The signal is modeled by a discrete vector $\inp = \tp{[x_0\cdots x_{\N-1}]}\in\reals^\N$, 
where $\N$ is the number of elements in the grid, corresponding to partitioning the interval $w_i\in [0, 1)$ into $\N$ equispaced segments. 
Each nonzero element in $\inp$ corresponds to one spike in~\fref{eq:spikescont}. The \ac{PSF} is modeled by matrix 
$\matQ$ that implements an ideal low-pass filter in the sense that it has a flat spectrum 
with a sharp cut-off at $\flo$. Formally,
\begin{equation}
    \label{eq:matqf}
    \matQ = \herm\matF \hat\matQ \matF,
\end{equation}
where
is the $\N\times \N$ discrete Fourier transform matrix
\begin{equation}
    [\matF]_{k,l} = \frac{1}{\sqrt{\N}}e^{-\iu 2\pi kl/\N}
\end{equation}
and $\hat \matQ=\diag(\tp{[\hat Q_{-\N/2+1}\cdots \hat Q_{\N/2}]})$ with
\begin{equation}
    \label{eq:spec}
    \hat Q_k =
    \begin{cases}
        1,&k=-\flo,\ldots, \flo,\\
        0,&\text{otherwise}.
    \end{cases}
\end{equation}
The wavelength $\llo=1/\flo$ gives the width of the convolution kernel represented by $\matQ$. 
We assume throughout the paper that $\N$ is even for simplicity.

Translated to discrete setting the model in~\fref{eq:convcont} becomes
\begin{equation}
    \label{eq:IO3altnoise1d}
    \data = \matQ\inp + \noise.
\end{equation}

\subsection{Recovery algorithm}
Our recovery method from the observations $\data$ in
\fref{eq:IO3altnoise1d} is extremely simple: solve
\begin{align}
    \label{eq:find0}\tag{$\mathrm{CVX}$}
    \hat\inp = \argmin_{\tilde\inp} \quad \onenorm{\data-\matQ\tilde\inp} \quad 
    \text{subject to}\quad \tilde\inp\ge \veczero.
\end{align}
In other words, we are looking for a set of positive spikes
such that the mismatch in received intensities is minimum. Note that
this method does not
make use of any knowledge other than the observations $\data$ and
the \ac{PSF} $\matQ$. Furthermore, (CVX) is a simple convex optimization
program, which can be recast as a linear program since both
$\inp$ and $\matQ$ are real valued.

\subsection{Rayleigh regularity}
\label{sec:reyleigh}

Consider discrete signal
$\inp \in \reals^{\N}$ as samples on the grid
$\{0,1/\N,\ldots,1-1/\N\}\subset\T$, where $\T$ is the circle in 1D, i.e., the interval 
$[0, 1)$ with $0$ and $1$ identified. We introduce a definition of Rayleigh regularity inspired 
by~\cite[Def.~1]{donoho92-09}. Let $\support(\inp) \define \{l/\N: x_l>0\}$ denote the support of the discrete signal.

As we shall see, our ability to super-resolve the signal $\inp$, will be fundamentally determined 
by how regular $\support(\inp)$ is in the following sense.
\begin{dfn}[\bf Rayleigh regularity]
\label{def:rreg}
We say that the set of points $\setV\subset\T$ is Rayleigh-regular with parameters $(d,\rr)$
and write $\setV\in \rset{d}{r}$ if it may be partitioned as
$\setV=\setV_1\union\ldots\union\setV_r$, where the $\setV_i$'s are
disjoint, and each obeys a separation constraint:
\begin{enumerate}
    \item for all $1\le i<j\le \rr$, $\setV_i\intersect\setV_j=\varnothing$;
    \item for all intervals $\setD\subset\T$ of length $\abs{\setD} = d$ (in terms of the 
    wrap-around distance\footnote{For $a, b\in \T$, the wrap-around  distance between $a$ and $b$ is  $\abs{b-a}\define b-a \mod 1$. For an interval $[a, b]$, its length is $\abs{b-a}$.} on $\T$) and all $i$,
    \begin{equation}
        \abs{\setV_i\intersect \setD}\le 1.
    \end{equation}
\end{enumerate}
\end{dfn}
In this paper we are interested in super-resolving signals with Rayleigh-regular support: $\support(\inp)\in\rset{d}{r}$.
Such signals are illustrated in \fref{fig:rclass}. 

As we will discuss, in the special case when $\support(\inp)\in \rset{\tcc\llo}{1}$ [i.e., when $r=1$] with $\tcc$ a bit larger than one (as in \fref{fig:rclass1}), the super-resolution problem is particularly easy. In this case we will say that the spikes in $\vecx$ are \emph{well-separated}.

\newcommand\cxi{0.1196}
\newcommand\cxil{0.8261}
\newcommand\cyi{0.3370}
\newcommand\cxii{0.1196}
\newcommand\cyii{0.5000}
\newcommand\cxiii{0.1196}
\newcommand\cyiii{0.7283}
\newcommand\cxiiii{0.0217}
\newcommand\cyiiii{0.1087}
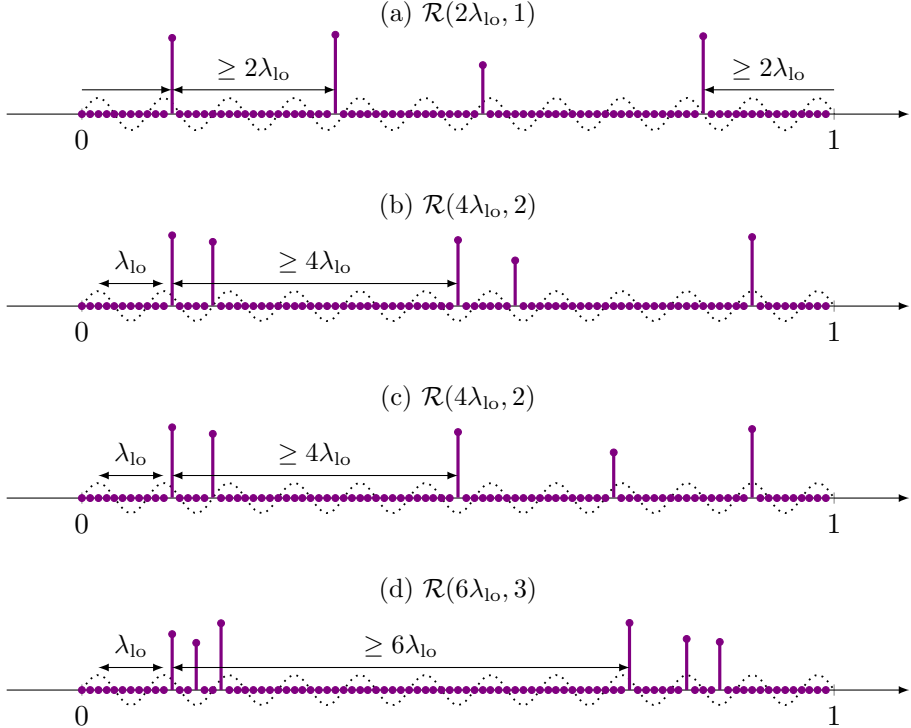
\begin{figure}[ht]
    \centering
    \vspace{-0.3cm}
    \begin{subfigure}{\textwidth}
        \caption{$\rset{2\llo}{1}$}\vspace{-0.2cm}
        \label{fig:rclass1}
        \centering
        \begin{tikzpicture}
        \begin{axis}[width=12cm,scale only axis, height=1.4cm,
            ytick=\empty, xtick={0,1},
            every axis title/.style={at={(0.5,1)},above,yshift=-0.2cm}, axis x line = middle,
            hide y axis, ymin=-0.3, ymax=1, xmin=-0.1, xmax=1.1]
            \addplot[domain=0:1,samples=201,line width=0.7pt,black, dotted]{0.2*sin(deg(2*pi*(23/2)*x))};
            \addplot[ycomb,range=0:1,line width=0.1pt,violet,mark options={fill=violet,scale=1},mark=*, 
                mark size=1.3pt,only marks] table[x index=0,y index=1]{./figures/figrclass.txt};
            \addplot[ycomb,range=0:1,line width=1.2pt,violet] 
                table[x index=0,y index=1]{./figures/figrclass.txt};
            \draw[help lines] (axis cs:0,0) -- (axis cs:1,0);
            \draw[<->,>=latex] (axis cs: \cxi,0.3) -- (axis cs: \cyi,0.3) node[above,midway] {\small $\ge 2\llo$};
            \draw[->,>=latex] (axis cs: 0,0.3) -- (axis cs: \cxi,0.3) node[above,midway] {};
            \draw[<-,>=latex] (axis cs: \cxil,0.3) -- (axis cs: 1,0.3) node[above,midway] {\small $\ge 2\llo$};
        \end{axis}
        \end{tikzpicture}
    \end{subfigure}\vspace{0.4cm}

    \begin{subfigure}{\textwidth}
        \caption{$\rset{4\llo}{2}$}\vspace{-0.2cm}
        \label{fig:rclass2}
        \centering
        \begin{tikzpicture}
        \begin{axis}[width=12cm,scale only axis, height=1.4cm,
            ytick=\empty, xtick={0,1}, axis x line = middle,
            every axis title/.style={at={(0.5,1)},above,yshift=-0.2cm}, 
            hide y axis, ymin=-0.3, ymax=1.1, xmin=-0.1, xmax=1.1]
            \addplot[domain=0:1,samples=201,line width=0.7pt,black, dotted]{0.2*sin(deg(2*pi*(23/2)*x))};
            \draw[<->,>=latex]  (axis cs:\cxiiii,0.3) -- (axis cs:  \cyiiii,0.3) node[above,midway] {\small $\llo$};
            \addplot[ycomb,range=0:1,line width=0.1pt,violet,mark options={fill=violet,scale=1},mark=*, 
                mark size=1.3pt,only marks]
                table[x index=0,y index=2]{./figures/figrclass.txt};
            \addplot[ycomb,range=0:1,line width=1.2pt,violet]
                table[x index=0,y index=2]{./figures/figrclass.txt};
            \draw[help lines] (axis cs:0,0) -- (axis cs:1,0);
            \draw[<->,>=latex] (axis cs:\cxii,0.3) -- (axis cs: \cyii,0.3) node[above,midway] {\small $\ge 4\llo$};
        \end{axis}
        \end{tikzpicture}
    \end{subfigure}\vspace{0.4cm}

    \begin{subfigure}{\textwidth}
        \caption{$\rset{4\llo}{2}$}\vspace{-0.2cm}
        \label{fig:rclass2alt}
        \centering
        \begin{tikzpicture}
        \begin{axis}[width=12cm,scale only axis, height=1.4cm,
            ytick=\empty, xtick={0,1}, axis x line = middle, 
            every axis title/.style={at={(0.5,1)},above,yshift=-0.2cm},
            hide y axis, ymin=-0.3, ymax=1.1, xmin=-0.1, xmax=1.1]
            \addplot[domain=0:1,samples=201,line width=0.7pt,black, dotted]{0.2*sin(deg(2*pi*(23/2)*x))};
            \draw[<->,>=latex]  (axis cs:\cxiiii,0.3) -- (axis cs:  \cyiiii,0.3) node[above,midway] {\small $\llo$};
            \addplot[ycomb,range=0:1,line width=0.1pt,violet,mark options={fill=violet,scale=1},mark=*, 
                mark size=1.3pt,only marks]
                table[x index=0,y index=5]{./figures/figrclass.txt};
            \addplot[ycomb,range=0:1,line width=1.2pt,violet]
                table[x index=0,y index=5]{./figures/figrclass.txt};
            \draw[help lines] (axis cs:0,0) -- (axis cs:1,0);
            \draw[<->,>=latex] (axis cs:\cxii,0.3) -- (axis cs: \cyii,0.3) node[above,midway] {\small $\ge 4\llo$};
        \end{axis}
        \end{tikzpicture}
    \end{subfigure}\vspace{0.4cm}

    \begin{subfigure}{\textwidth}
        \caption{$\rset{6\llo}{3}$}\vspace{-0.2cm}
        \label{fig:rclass3}
        \centering
        \begin{tikzpicture}
        \begin{axis}[width=12cm,scale only axis, height=1.4cm,
            ytick=\empty, xtick={0,1}, axis x line = middle,
            every axis title/.style={at={(0.5,1)},above,yshift=-0.2cm},
            hide y axis, ymin=-0.3, ymax=1.1, xmin=-0.1, xmax=1.1]
            \addplot[domain=0:1,samples=201,line width=0.7pt,black, dotted]{0.2*sin(deg(2*pi*(23/2)*x))};
            \draw[<->,>=latex]  (axis cs:\cxiiii,0.3) -- (axis cs:  \cyiiii,0.3) node[above,midway] {\small $\llo$};
            \addplot[ycomb,range=0:1,line width=0.1pt,violet,mark options={fill=violet,scale=1},mark=*, 
                mark size=1.3pt,only marks] table[x index=0,y index=3]{./figures/figrclass.txt};
            \addplot[ycomb,range=0:1,line width=1.2pt,violet] table[x index=0,y index=3]{./figures/figrclass.txt};
            \draw[<->,>=latex]  (axis cs:\cxiii,0.3) -- (axis cs: \cyiii,0.3) node[above,midway] {\small$\ge 6\llo$};
        \end{axis}
        \end{tikzpicture}
    \end{subfigure}
\caption{Examples of discrete $\N$ dimensional signals whose support belongs to the
Rayleigh classes $\rset{2\llo}{1}$, $\rset{4\llo}{2}$, $\rset{6\llo}{3}$
depicted on the grid $\{0,1/\N,\ldots,1-1/\N\}\subset \T$. Note that the 
signals in (b) and in (c) both have support in $\rset{4\llo}{2}$. In general, Rayleigh regularity 
\emph{does not} require that all spikes in 
the signal are arranged into separated clusters as is the case in (b) and in (d).
The sine wave $\sin(2\pi \flo t)$ at the highest visible frequency is
shown by the dotted line for reference. Here, $\N=92$ and $\flo=11$, so
that $\llo=1/11$. By periodicity, the endpoints are
identified.}
\label{fig:rclass}
\end{figure}

\subsection{Discrete stability estimates}
The main result of the earlier work~\cite{candes-14} is the following proposition.

\begin{prp}
\label{prp:UB}
Assume $\inp\ge \veczero$ and $\support(\inp)\in
\rset{\kappa\llo \rr}{\rr}$ with $\kappa\define 1.87$ and $\flo\ge 128\rr$. 
Assume that the observations $\data$ are given
by~\fref{eq:IO3altnoise1d}. Then the solution
$\hat\inp$ to~\fref{eq:find0} obeys
\begin{equation}
    \label{eq:mainbnd}
    \underbrace{\onenorm{\hat\inp-\inp}}_\mathrm{Error} 
    \le C_d(\rr) \cdot \left(\frac{\N}{\flo}\right)^{2\rr} \cdot \|\mathbf{z}\|_1 
    = C_d(\rr) \cdot \DSRF^{2 \rr} \cdot  \|\mathbf{z}\|_1,
\end{equation}
where $C_d(\rr)$ only depends on $\rr$ (if $\DSRF\ge 3.03/\rr$, it can be taken to be $C_d(\rr)= r^{2r}\cdot 4 \cdot 17^r$).
\end{prp}
The ratio $\DSRF\define
\N/(2\flo)$ is called the {\em discrete super-resolution factor}; this
is the ratio between the scale at which we observe the data, $1/(2\flo)$, and the scale of the finest 
details in the data, $1/\N$.

\subsection{Breakdown of discrete stability estimates}
In practice, signals do not belong to a discrete grid. In order to accurately approximate the continuous 
model in~\eqref{eq:spikescont} we might need to make the grid very fine, i.e., take $\N$ large.

The problem is that the theoretical result in~\fref{eq:mainbnd} becomes meaningless when $\flo$ 
and $\|\mathbf{z}\|_1$ remain fixed, and $\N\to\infty$. Indeed, observe that~\fref{eq:mainbnd} 
guarantees accurate signal recovery when the right-hand side of~\fref{eq:mainbnd} is much smaller 
than $\onenorm{\inp}$.  When $\N\to\infty$ with $\flo$ fixed, then $\DSRF^{2 \rr}\to\infty$ very 
quickly, so that the right-hand side of~\fref{eq:mainbnd} becomes larger than $\onenorm{\inp}$, even 
for very small noise.

This is expected. Consider the hypothetical situation illustrated in \fref{fig:contest}. The true 
signal $\inp$ consists of three spikes as depicted in \fref{fig:sigest} in purple solid. The grid is 
very fine ($\N$ is large); the PSF is wide as shown in \fref{fig:vecs} (the purple solid curve, with 
characteristic width $\llo$, represents $\data=\matQ\inp$ with $\inp$ from \fref{fig:sigest}); and 
the data is noisy. Imagine an algorithm produced an estimate $\hat\inp_\mathrm{good}$ as depicted 
in \fref{fig:goodest} by the blue dashed spikes.
\begin{figure}[h!t]
    \vspace{-0.3cm}
    \begin{subfigure}{0.5\textwidth}
        \centering
        \caption{signal $\inp$}
        \includegraphics[width=0.9\textwidth]{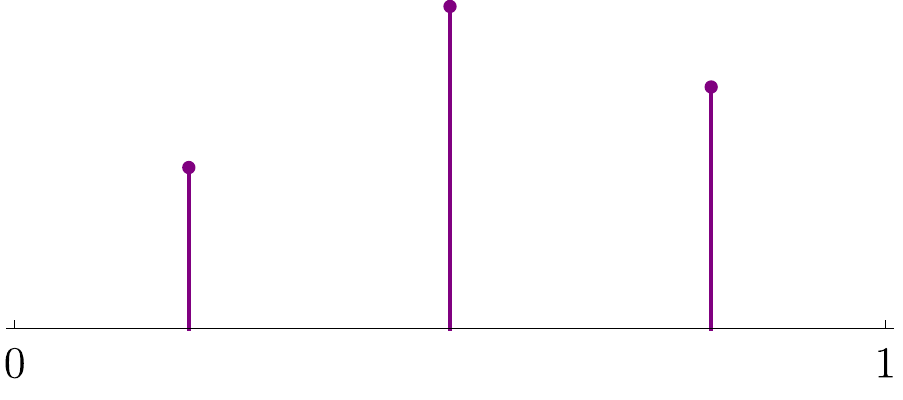}
        \label{fig:sigest}
    \end{subfigure}
    \begin{subfigure}{0.5\textwidth}
        \centering
        \caption{observations $\data$ in solid, kernel $\khi(\cdot)$ in dotted}
        \includegraphics[width=0.9\textwidth]{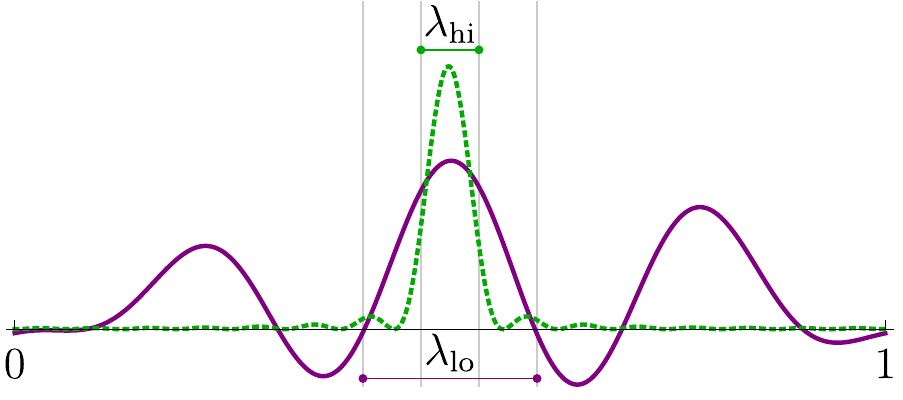}
        \label{fig:vecs}
    \end{subfigure}
    \vspace{0.8cm}

    \begin{subfigure}{0.5\textwidth}
        \centering
        \caption{$\inp$ in solid, $\hat \inp_\mathrm{good}$ in dashed, $\khi(\cdot)$ in dotted}
        \includegraphics[width=0.9\textwidth]{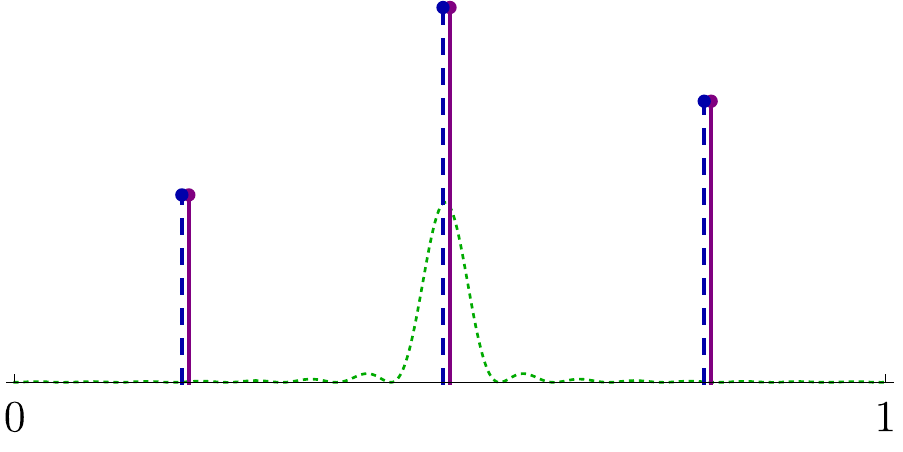}
        \label{fig:goodest}
    \end{subfigure}
    \begin{subfigure}{0.5\textwidth}
        \centering
        \caption{$\inp$ in solid, $\hat \inp_\mathrm{bad}$ in dashed, $\khi(\cdot)$ in dotted}
        \includegraphics[width=0.9\textwidth]{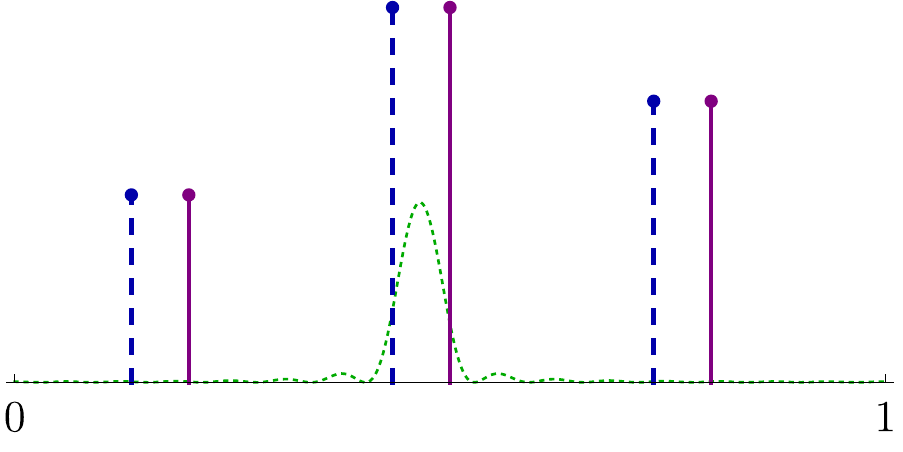}
        \label{fig:badest}
    \end{subfigure}
    \vspace{0.8cm}
    
    \begin{subfigure}{0.5\textwidth}
        \centering
        \caption{$\text{error}=\onenorm{\khi\conv (\inp-\hat \inp_\mathrm{good})}$}
        \includegraphics[width=0.9\textwidth]{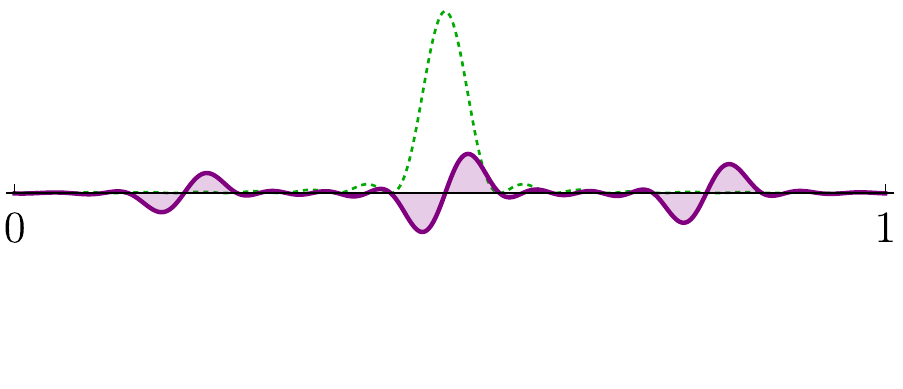}
        \label{fig:goodesterror}
    \end{subfigure}    
    \begin{subfigure}{0.5\textwidth}
        \centering
        \caption{$\text{error}=\onenorm{\khi\conv (\inp-\hat \inp_\mathrm{bad})}$}
        \includegraphics[width=0.9\textwidth]{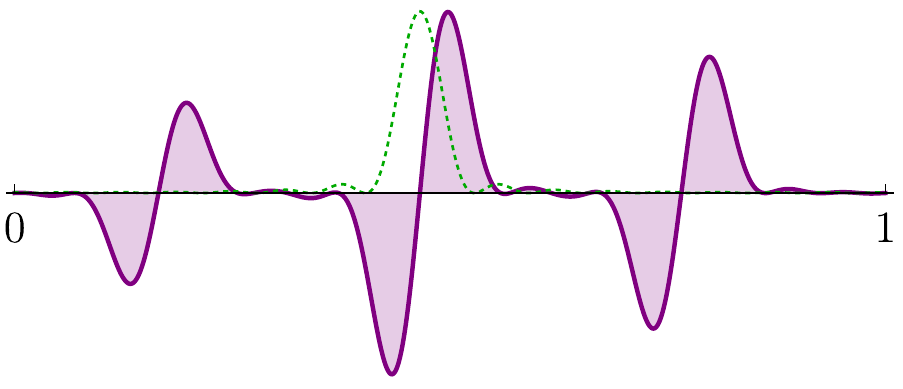}
        \label{fig:badesterror}
    \end{subfigure}
    \vspace{-0.5cm}
    \caption{Measuring the estimation error when the grid is very fine ($\N$ is large).}
\label{fig:contest}
\end{figure}
The estimate $\hat\inp_\mathrm{good}$ is excellent: the blue dashed spikes are located in the neighboring 
discrete bins to the corresponding purple solid ground truth spikes, the magnitudes are estimated perfectly. 
In the presence of noise we cannot hope for infinite resolution, so for large $\N$, we should be 
happy if we were able to obtain $\hat \inp_\mathrm{good}$ as in \fref{fig:goodest}. Yet,
\begin{equation}
    \onenorm{\hat\inp_\mathrm{good}-\inp} = 2 \onenorm{\inp},
\end{equation}
i.e., the estimation error is about as large as it can possibly be. We conclude that the 
reason why the result in~\fref{eq:mainbnd} becomes meaningless when $\flo$ and $\onenorm{\noise}$ 
remain fixed and $\N\to\infty$ is that the error metric $\onenorm{\hat\inp-\inp}$ becomes 
inadequate. We need a more \emph{forgiving} error metric that should penalize small 
localization errors on the fine grid mildly.

We will explain in \fref{sec:mainres} how to construct the more forgiving error 
metric and how to change the definition of super-resolution factor accordingly. With 
these modifications we can generalize~\fref{prp:UB} and formulate the stability estimates 
in \fref{thm:mainthm} that remain meaningful even when $\N\to\infty$ and $\flo$ and $\onenorm{\noise}$ 
are fixed. 
With the appropriate new definitions, the result in \fref{thm:mainthm} is nearly identical 
to that in~\fref{prp:UB}. Surprisingly, the proof technique necessary to obtain 
\fref{thm:mainthm} is much harder than the trick that was sufficient to prove~\fref{prp:UB}. 
The proof relies on new trigonometric interpolation constructions that constitute 
the main mathematical contribution of this paper.

\section{Main results}
\label{sec:mainres}

\subsection{Measuring the reconstruction error}
To   avoid penalizing the estimators that produce spikes very close to the original spikes on the 
fine grid, a natural approach is to convolve the difference $\hat\inp-\inp$ with a 
nonnegative kernel $\khi(\cdot)$ of width $\lhi$ (represented by the dotted green line in \fref{fig:vecs}) before computing the $\ell_1$ norm:
\begin{equation}
    \text{error}=\onenorm{\khi\conv (\hat\inp- \inp)},
\end{equation}
where
\begin{equation}
    \left[\khi\conv (\hat\inp-\inp)\right]_n \define \sum_{m=0}^{\N-1} \khi\lefto(\frac{n-m}{\N}\right) h_m
\end{equation}
and $\vech = \tp{[h_0, h_1, \ldots, h_{\N-1}]} \define \hat\inp-\inp$ is the difference vector.
The new error metric is illustrated in Figures~\ref{fig:goodest}--\ref{fig:badesterror}. 
When the estimated spikes are closer 
than $\lhi$ to the original spikes, as is the case for $\hat\inp=\hat\inp_\mathrm{good}$ 
in~\fref{fig:goodest}, the error, represented by the area of the shaded region in~\fref{fig:goodesterror},
is very small.
Conversely, when the estimated spikes are further than $\lhi$ from the original spikes, as is the case 
for $\hat\inp=\hat\inp_\mathrm{bad}$ in~\fref{fig:badest},  we have,
$\text{error}=\onenorm{\khi\conv (\hat\inp-\inp)}\approx 2\onenorm{\inp}$, so that the
error is large, as illustrated in~\fref{fig:badesterror}.

The width, $\lhi$, of the kernel $\khi(\cdot)$ is a parameter of the theory. This 
parameter will be chosen to be
\begin{inparaenum}[(i)]
    \item larger (or equal to) the finest scale of the data, $\lhi\ge 1/\N$, and, simultaneously,
    \item smaller than the native resolution of the observations, $\lhi< \llo$.
\end{inparaenum}
Having chosen $\lhi$, we define the \emph{super-resolution factor} as:
\begin{equation}
    \SRF \define \frac{\llo}{\lhi}.
\end{equation}
The $\SRF$ will play the same role in our theory as the $\DSRF$ played in~\fref{prp:UB}.
In \fref{fig:vecs}, the $\SRF$ is the ratio between $\llo$, the width of the kernel $\matQ$,
and $\lhi$, the width of the kernel $\khi(\cdot)$.

To be concrete, a reasonable situation might be: $\llo = 1/10$, $1/\N = 1/1000$ so 
that $\DSRF = 100$. This makes the right-hand side of~\fref{eq:mainbnd} huge so that 
the stability estimate is useless. Now, choose $\lhi=1/100$ so that $\SRF=10$, which 
is much smaller than $\DSRF$. The main result of this paper, \fref{thm:mainthm} below, 
shows that we can upper-bound the error
$\onenorm{\khi\conv (\hat\inp-\inp)}$ in terms of $\SRF^{2r}$, which is much smaller 
than $\DSRF^{2r}$, keeping the bound tight for realistic values of the noise.

For $\khi(\cdot)$, in this paper we use the Fejér kernel:
\begin{equation}
    \khi(t) 
    \define \frac{1}{\N}\frac{1}{\fhi+1}  
        \left(\frac{\sin(\pi (\fhi+1) t)}{\sin(\pi t)}\right)^2,\quad 
    \fhi = 1/\lhi.
    \label{eq:khi}
\end{equation}
The normalization is such that
\begin{equation}
    \sum_{n=0}^{\N-1} \khi\lefto(\frac{n}{\N}\right) = 1, \label{eq:cabshi}
\end{equation}
which ensures that the ``energy'' in the error is preserved in the sense that 
$\text{error}=\onenorm{\khi\conv (\hat\inp-\inp)}\approx 2\onenorm{\inp}$ whenever 
the estimated spikes in $\hat\inp$ are far away from the true spikes in $\inp$. The concrete form of the kernel
$\khi(\cdot)$ is not important. Our results hold for any other periodic nonnegative high-resolution kernel
as long as it satisfies conditions~\fref{eq:cabshid} and \fref{eq:cabshidd} below.

When $\N\to\infty$, the error metric defined here becomes the one used in~\cite{candes21-2} 
in the analysis of the continuous super-resolution problem. Compared to~\cite{candes21-2}, the key novelty of this
paper is that the results in~\cite{candes21-2} apply only when the spikes in the signal are well-separated [$\support(\inp)\in \rset{1.87\llo}{1}$] as in \fref{fig:rclass1}, a stringent assumption. 
In this paper we don't assume that the spikes are well-separated 
 and our results also hold for 
signals with $\support(\inp)\in\rset{1.87\llo\rr}{\rr}$ and $\rr>1$ as in Figures~\ref{fig:rclass2}, \ref{fig:rclass2alt}, \ref{fig:rclass3}. 
The price we pay is that our results are only valid 
for nonnegative signals, whereas the results in~\cite{candes21-2} are valid for complex-valued signals.

\subsection{Stability estimate on an arbitrarily fine grid}
In this paper we prove the following theorem.
\begin{thm}
\label{thm:mainthm}
Assume $\inp\ge \veczero$ and $\support(\inp)\in
\rset{\kappa\llo \rr}{\rr}$ with $\kappa\define 1.87$ and $\flo\ge 128\rr$. 
Assume $\lhi<\llo$, $\lhi\ge 1/\N$, and $\SRF>12$.
Assume, in addition, that the elements of $\support(\inp)$ are separated by at least $2\lhi$: 
if $t,t'\in\support(\inp)$ with $t\ne t'$, then $\abs{t-t'}\ge 2\lhi$, where $\abs{\cdot}$ 
is the wrap-around distance on $\T$.
Assume that the observations $\data$ are given
by~\fref{eq:IO3altnoise1d}. Then the solution $\hat\inp$ to \fref{eq:find0} 
obeys
\begin{equation}
    \label{eq:stability}
    \onenorm{\khi\conv (\hat\inp-\inp)} \le \CC{\rr} \SRF^{2\rr} \onenorm{\noise},
\end{equation}
where $\CC{\rr}\define\rr^{2\rr+4} \cc^{\rr+1}$ and the positive numerical constant $\cc$ is 
defined in~\fref{eq:ccdef} below.
\end{thm}
The theorem is proven in the next section and in the appendices. Before we embark on the proof, 
we discuss the significance and the accuracy of the result.

\subsubsection{Significance of the result}
\fref{thm:mainthm} gives essentially the same stability estimate for an arbitrarily fine 
grid as~\fref{prp:UB} does for a discrete grid. With the new definition for error metric, $\lhi$ in \fref{thm:mainthm}
plays the same role as the grid segment size, $1/\N$, played in~\fref{prp:UB}. In turn, the grid segment
size, $1/\N$, in \fref{thm:mainthm} may be arbitrarily small without affecting the stability 
estimate at all. The only thing that changes when $\N$ grows is that it becomes numerically 
harder to solve~\fref{eq:find0}.

\subsubsection{Tightness}
The result is information-theoretically tight in the following sense. It is possible to prove 
a converse theorem (see~\cite[Sec.~2.3]{candes-14}) that says that the best possible algorithm in the 
worst case (the minimax setting)
cannot achieve stability estimate in~\fref{eq:stability} with super-resolution 
factor dependence better than $\SRF^{2\rr-1}$. In other words,
the exponent of $\SRF$ in~\fref{eq:stability} is near-optimal.

We have made no attempt to optimize $\CC{\rr}$. Finding the tightest possible $\CC{\rr}$ is an 
important open problem, which seems to be hard to address with the mathematical techniques developed 
in this paper.

\subsubsection{Mathematical novelty} 
The reader might expect that since \fref{thm:mainthm} is so similar to~\fref{prp:UB}, the 
proof of \fref{thm:mainthm} is a minor modification of the work done in~\cite{candes-14}. 
Perhaps surprisingly, this is not the case. 

The proof technique in~\cite{candes-14} relied on a simple and elegant trigonometric 
interpolation construction reviewed in \fref{sec:subqz}. In this paper, in addition, we had 
to develop a flexible set of techniques that allowed us to build trigonometric 
polynomials with specific interpolation properties. These techniques---that constitute 
the main mathematical contribution of this paper---are presented in Sections~\ref{sec:q1} and~\ref{sec:q2} 
and in~\fref{app:dualprod1d}. 
We believe that the new techniques are interesting in their own right and may be useful in other projects. 

\subsubsection{Separation by $2\lhi$}
\fref{thm:mainthm} requires the assumption that no two spikes in $\inp$ are 
closer than $2\lhi$. It is important to contrast this assumption with the 
separation assumption in~\cite{candes21-2, candes13}. The results in~\cite{candes13} hold only when 
no two spikes in $\inp$ are closer than $1.87\llo$ (the spikes are well-separated). 
Our separation requirement is much weaker than the one needed in~\cite{candes21-2, candes13}: we require 
the separation at the scale of $\lhi$ whereas the results in~\cite{candes21-2, candes13} need separation 
on the scale of~$\llo$. Since the whole point of super-resolution is to reconstruct the 
original signal with accuracy about $\lhi\ll\llo$, our assumption is mild, whereas the assumption 
in~\cite{candes21-2, candes13} is restrictive.

Further, it follows from the proof of \fref{thm:mainthm} that the  $2\lhi$ separation 
requirement may be relaxed to, for example, $\lhi/2$, or, more generally, to 
$\lhi/\beta$ for any $\beta>1$. The result in \fref{thm:mainthm} will not change, 
except that the constant $\CC{\rr}$ will now depend on $\beta$. Specifically, the result will read:
\begin{equation}
    \onenorm{\khi\conv (\hat\inp-\inp)} \le \rr^{2\rr+4} \cc^{\rr+1} \beta^{2\rr} \SRF^{2\rr} \onenorm{\noise}.
\end{equation}
To keep the proof of \fref{thm:mainthm} as clean as possible, we decided to stick with 
the $2\lhi$ separation assumption in the theorem.

Finally, it is not clear if the separation assumption of the form $\lhi/\beta$ is 
fundamentally necessary. Certainly, it is necessary for the proof technique developed in 
this paper. It is an open problem to either find a proof of \fref{thm:mainthm} that does 
not rely on this assumption, or to prove a converse result showing that this assumption 
is unavoidable. Note that there is no \emph{explicit} separation assumption in~\fref{prp:UB}; 
however, since the spikes are on the grid, the separation assumption at the 
scale of $1/\N$ is made \emph{implicitly}.

\subsubsection{Density constant}
We next discuss the following question: can the constant $\kappa=1.87$ in \fref{thm:mainthm} 
be made smaller without changing the result? The answer is ``probably yes''. Specifically, 
our proof builds upon Lemmas~\ref{lem:dualcarlos} and~\ref{lem:dualcarlos2} below. 
The lemmas generalize~\cite[Lm.~2.4, Lm.~2.5, Sec.~2.5]{candes21-2} and their proof exploits a 
construction developed in~\cite{candes13}. The specific value for $\kappa=1.87$ comes from 
the construction borrowed from~\cite{candes13}. An improved construction has recently been 
reported in~\cite{fernandez-granda16a} leading to a smaller value $\kappa=1.26$. To keep 
this paper as simple as possible, we decided not to accommodate this improvement. To do so, 
one would need to change Lemmas~\ref{lem:dualcarlos} and~\ref{lem:dualcarlos2}
below and the proof in~\fref{app:dp}; all 
other derivations in this paper will remain unchanged.
The constant $\CC{\rr}$ in \fref{thm:mainthm} would need to be updated accordingly. 

We expect that there is a trade-off: the larger $\kappa$ is, the smaller the constant $\CC{\rr}$ 
can be made. However, our estimates do not provide the smallest 
possible constant. Hence, we cannot analyze the trade-off.

Finally, as explained in~\cite[Sec.~2.3.1]{candes-14}, $\kappa>1$ is a fundamental limit, so our 
result is within the factor $1.87$ from the optimum.

\subsubsection{Gridless super-resolution}
It has been shown in~\cite{bhaskar11-09, candes13, candes21-2} that under the assumption 
that spikes are separated by at least $1.87\llo$ (well-separated spikes), one can solve the gridless 
super-resolution problem in which the spikes have completely arbitrary locations 
on~$\T$ (no need for the $1/\N$ discretization). It turns out that in the gridless setup one needs 
to solve an infinite-dimensional, but convex, total-variation-minimization problem 
(see~\cite[eq.~(1.4)]{candes13}). Surprisingly, if one works in the dual domain and 
uses the idea of lifting, the equivalent problem becomes finite-dimensional and, 
therefore, may be solved on the computer. The solution to the original problem may 
then be reconstructed by duality. This approach is explained in \cite[Sec.~4]{candes13}.

The approach, by now standard, may be carried over to the problem considered in 
this paper, where we work with a nonnegative signal $\inp$ and the spikes need not be well-separated.
The same trigonometric polynomials that certify 
optimality of~\fref{eq:find0} and lead to \fref{thm:mainthm} may also be used to 
prove stability of the corresponding gridless algorithm.

The reason why we chose to focus on the arbitrarily fine grid and not 
to discuss the gridless problem in details is the following \emph{practical} consideration. 
In applications, for example in super-resolution microscopy, there 
is no real difference between the gridless problem and the problem with 
a very fine grid. The real sources have some finite 
nonzero size, perhaps small. Therefore, in practice, one has a choice between solving 
\fref{eq:find0} on a sufficiently fine grid or solving the infinite-dimensional total-variation-minimization
problem via lifting. To solve \fref{eq:find0} with $\N$ variables 
efficiently, one would use a first-order solver whose complexity 
is dominated by repeated multiplications by $\matQ$, $\tp\matQ$. Using~\fref{eq:matqf} 
one would implement $\matQ$ via the fast Fourier transform so that each matrix multiplication 
takes $\mathcal{O}(\N\log \N + \flo)$ multiplications. The gridless approach via lifting requires 
one to solve a semidefinite convex optimization problem  (see~\cite[eq.~(4.3)]{candes13}) 
with $\mathcal{O}(\flo^2)$ variables. The complexity of the gridless approach does not 
depend on $\N$ at all, a very nice property. However, the necessity to deal with a 
semidefinite problem with $\mathcal{O}(\flo^2)$ variables make it much more costly 
than solving~\fref{eq:find0} on a sufficiently fine grid in the applications we have encountered.

\subsubsection{General \acp{PSF}}
The sharp rectangular frequency cut-off of $\matQ$ in~\fref{eq:spec} corresponds to 
the \ac{PSF} $\klo(\cdot)=\sinc(\cdot)$ in~\fref{eq:convcont}. The $\sinc$ function 
takes negative values (as shown in \fref{fig:vecs} in purple solid), whereas all \acp{PSF} 
in microscopy take nonnegative values (as shown in \fref{fig:conv}). The 
simplest PSF that takes nonnegative values is the Fejèr kernel. The spectrum 
of $\matQ$ that corresponds to the Fejèr kernel has a triangular decay of 
$\hat q_k$ in~\fref{eq:spec} as in~\cite[eq.~13]{candes-14} and in~\fref{eq:fejsum}. 
The results for the rectangular 
spectrum can be translated into the results for the triangular spectrum (in fact for the spectrum 
of any reasonable shape) using the idea of spectrum equalization. We refer the reader 
to~\cite{candes-14} for a detailed explanation on how this can be done. In this paper 
we focus on the basic case in~\fref{eq:spec} only.

\subsubsection{2D model}
All results in this paper are for the 1D model. The discrete results have been 
generalized to the 2D model in~\cite{candes-14}. We believe that the  
results in this paper may be generalized to the 2D model in a similar way. We leave this 
generalization for future work.

\section{Literature review and innovations}
\subsection{Prior art}
\paragraph*{Prony's method.} 
Prony's method~\cite{prony95essai} is an algebraic approach for
solving the gridless super-resolution problem from \emph{noiseless} data when the
number of spikes is known a priori.  The observations $\data$ are used to
form a trigonometric polynomial, whose roots coincide with the spike
locations. The trigonometric polynomial is then factored, thus revealing those
locations, and the amplitudes estimated by solving a system of linear
equations. In the noiseless case, Prony's method recovers $\inp$
perfectly provided that $\zeronorm{\inp}\le\flo$. No further
Rayleigh regularity assumption on the signal support is needed.  With
noise, however, the performance of Prony's method degrades
sharply. The difficulty comes from the fact that the roots of a
trigonometric polynomial constructed by an algebraic method are
unstable and can shift dramatically even with small changes
in the data. Therefore, a crucial problem is to solve the super-resolution problem
in the \emph{presence of noise}.

\paragraph*{Fundamental limits.} 
In the pioneering work~\cite{donoho92-09}, Donoho studied limits of
performance for the super-resolution problem and recognized the importance of
Rayleigh regularity as the fundamental property that determines how easy it is 
to super-resolve the signal. He analyzed an \emph{intractable} exhaustive search algorithm 
and demonstrated that assuming $\support(\inp)\in\rset{2\llo \rr}{\rr}$, 
the estimator, $\hat\inp$,  produced by this algorithm satisfies:
\begin{equation}
    \label{eq:donest}
    \twonorm{\hat\inp-\inp}\le \tCC{\rr} \SRF^{2\rr+1}\twonorm{\noise}.
\end{equation}
The algorithm proposed by Donoho may only be applied to vectors $\inp$ with very 
few dimensions. Therefore, the fundamental problem posed by Donoho is to find 
an \emph{efficient} algorithm that is stable in the sense of~\fref{eq:donest}. 
Donoho has also proven a converse to~\fref{eq:donest}: the $\SRF$ 
dependence in~\fref{eq:donest} cannot be better than $\SRF^{2\rr-1}$ 
even for the best possible algorithm in the worst-case scenario (the minimax setting). 
The results of Donoho have been recently (partially) improved 
in~\cite{demanet14the-r, batenkov20a} where for the same intractable algorithm the 
following stability estimate was derived:
\begin{equation}
\label{eq:demanet}
\twonorm{\hat\inp-\inp}\le \tCC{\rr, \zeronorm{\inp}} \SRF^{2\rr-1}\twonorm{\noise}.
\end{equation}
The result is sharp in the sense that the $\SRF$ dependence matches Donoho's converse. 
The weakness is that $\tCC{\rr, \zeronorm{\inp}}$ depends on the total number of 
spikes in the signal, which may be very large. Note also that the stability estimates 
in~\fref{eq:donest}, \fref{eq:demanet} are expressed in terms of $\ell_2$ norms, 
whereas our stability estimates in~\fref{eq:stability} are expressed in terms of $\ell_1$ norms.

Other works \cite{stoica89stati,stoica89maxim,batenkov13accuracy}
study the stability of the super-resolution problem in the presence of
noise, but likewise do not provide a tractable algorithm to perform
recovery. Work in~\cite{shahram04-05,shahram05resolv,helstrom64-10}
analyzes the detection and separation of two closely-spaced spikes,
but does not generalize to the case when there are more than two
spikes in the signal.

\paragraph*{Super-resolution for well-separated spikes.}
Progress towards resolving the question posed in \cite{donoho92-09} in
the general situation where $\inp\in\complexset^\N$---in this
paper we consider the case $\inp\ge\veczero$ only---has been
made in~\cite{candes13,candes21-2,fernandez-granda16a}. The sharpest from this series of 
results~\cite{fernandez-granda16a} implies the following.
Assume $\support(\inp)\in\rset{1.26\llo}{1}$, then 
the solution to $\ell_1$-minimization problem
\begin{equation}
    \tag{L1}
    \label{eq:l1min}
    \hat\inp = \argmin_{\tilde\inp} 
    \onenorm{\tilde\inp} \quad  
    \text{subject to} \quad 
    \vecnorm{\data-\matQ\tilde\inp}_1\le \delta
\end{equation}
with $\delta$ chosen so that $\onenorm{\noise}\le \delta$ satisfies
\begin{equation}
    \label{eq:carlosbnd}
    \onenorm{\inp-\hat\inp}\le \tcc \cdot \SRF^{2},
\end{equation}
where  $\tcc$ is a positive numerical constant.  The requirement
$\support(\inp)\in\rset{1.26\llo}{1}$ (well-separated spikes in our terminology)
is restrictive because it means that the
signal $\inp$ cannot contain spikes that are at a distance less than
$1.26\llo$. This is a limitation
for many applications including single-molecule microscopy, as it is
usually understood that the goal of super-resolution is to distinguish
spikes that are (significantly) closer than the Rayleigh diffraction
limit, i.e.,~at a fraction of $\llo$ apart. Unfortunately, if there
are spikes at a distance smaller than $\llo$, $\ell_1$ minimization
does not, in general, return the correct solution even if there is no
noise. The central question therefore is: which algorithms and under which assumptions are
able to super-resolve signals \emph{robustly} when the distance between some of the spikes 
may be \emph{substantially smaller} than $\llo$?

On a similar line of research, see~\cite{tang13-1} and \cite{tang13compr}
for related results on the denoising of line spectra and on the
recovery of sparse signals from a random subset of their low-pass
Fourier coefficients. The accuracy of support detection for well-separated spikes 
is analyzed in~\cite{fernandez-granda13suppo,azais15spike}.

\paragraph*{Noise-aware algebraic methods.} 
Many noise-aware versions of Prony's method are used frequently in
engineering applications, for example in radar (see~\cite[Ch.~6]{stoica05}). 
The most popular methods are MUSIC and its numerous
variations~\cite{barabell83impro, bienvenu79influ, schmidt86-03,pisarenko73theretr,
tufts82estim,cadzow88signa}, matrix-pencil~\cite{hua90matri}, and
ESPRIT~\cite{paulraj86a-sub,roy89-07}. For more details on algebraic
methods we refer the reader to the excellent book~\cite[Ch.~4]{stoica05}. 
It is important to point out that unlike convex optimization based
methods like~\fref{eq:l1min}, algebraic methods do
\emph{not} need the spikes to be well-separated ($\inp$ may contain spikes closer 
than $\llo$) even when the signal is complex-valued, at least in the noiseless case.

The stability of noise-aware algebraic methods
is an active area of research. Asymptotic results (at high SNR)
on the stability of MUSIC in the presence of Gaussian noise are
derived in~\cite{clergeot89perfo,stoica91stati}. 
More recently, some
steps towards analyzing MUSIC and matrix-pencil in a non-asymptotic regime have
been taken in~\cite{liao14music} and in~\cite{moitra15a}, respectively.

Especially important is the question of stability of algebraic methods when the spikes are not well-separated.
Substantial progress in understanding this for MUSIC and ESPRIT algorithms
has been made by Li and Liao in the last two years~\cite{li19b, li17a, li20a}. See also~\cite{kunis19a} 
for a simplified exposition of ideas in~\cite{li17a} and some extensions.
The authors considered a \emph{separated cluster model} for spike locations; the model is similar to
Rayleigh regularity in spirit, but is more restrictive. For example, 
the signals depicted in Figures~\ref{fig:rclass2} and~\ref{fig:rclass2alt} are both Rayleigh-regular with $r=2$. 
However, only the signal in \fref{fig:rclass2}, 
but not the signal in \fref{fig:rclass2alt}, has separated spike clusters. 
For MUSIC in~\cite{li19b, li17a} and for ESPRIT in~\cite{li20a},
assuming Gaussian noise and making a further (restrictive) assumption $\flo\gtrsim \zeronorm{\inp}^2$,
the authors derived bounds on signal-to-noise ratio
in terms of $\SRF^{2\rr-2}$ and a factor that depend on $\flo$ so that
the correct signal support recovery is guaranteed. There is still a large gap between
these stability estimates and the minimax converse results. For example, for ESPRIT, the gap is a 
factor proportional to $\flo$, which may be very large for high-dimensional signals~\cite{li20a}. 
Hence, the problem of finding a super-resolution method for complex-valued signals that performs 
well empirically and has sharp 
theoretical stability estimates in the case when the spikes are not well-separated
is still open.

\paragraph*{Super-resolution of nonnegative signals.} 
The case of nonnegative signal, $\inp\ge \veczero$, was
analyzed in~\cite{donoho90-06}, see also~\cite{fuchs05} for a shorter
exposition of the same idea. It is proven in~\cite{donoho90-06} that as long as 
$\zeronorm{\inp}\le \flo$, one can recover $\inp$ by solving a simple convex feasibility 
problem in the noiseless setting. In the presence of noise, \cite{donoho90-06} does not provide 
sharp estimates: it does not reveal the correct $\SRF$ dependence in the stability estimate.

More recently, the authors of~\cite{schiebinger17a} generalized~\cite{donoho90-06} to the case of more general 
point spread functions and sampling patterns in the noiseless case. The corresponding noisy case has been studied 
in~\cite{eftekhari19a}. Being very general, the results of~\cite{eftekhari19a} do not appear to be sharp enough 
to reveal the fundamental dependence between the stability of the algorithm, the regularity of the signal,
and the super-resolution factor. 

Most relevant to this work is the earlier paper~\cite{candes-14} where~\fref{prp:UB} has been 
proven. The key question remained: what happens if the grid becomes arbitrarily fine or when there 
is no grid at all (the gridless setting). Some progress towards answering 
this question has since been made in~\cite{denoyelle17a} where stability estimates for the detection 
of signal support have been expressed in terms of $\SRF^{2\zeronorm{\inp}-1}$. Note that 
$\zeronorm{\inp}$ may be arbitrarily large for high-dimensional signals, and so the bounds 
in~\cite{denoyelle17a} become highly suboptimal for the practically relevant case 
in which the spikes are distributed in a regular way in the signal.

\subsection{Innovations}
The innovations in this paper may be summarized as follows:
\begin{itemize}
    \item Generalization of the results of~\cite{candes-14} to the case when the grid is arbitrarily fine. 
    \item Seamless connection between the super-resolution results for the discrete grid and the 
    results for the gridless (continuous) setting. This has theoretical as well as practical implications.
    \item Mathematically the paper builds on the ideas from~\cite{candes21-2} and \cite{candes-14} and 
    develops these methods further. The interpolation constructions
    in \fref{lem:dualprod1} and \fref{lem:dualprod1d} are new. These constructions may be of independent 
    interest and may be useful for other problems.
\end{itemize}

\section{Notation}
Sets are denoted by calligraphic letters $\setA, \setB$, and so on.
Boldface letters $\matA,\matB,\ldots$ and $\veca,\vecb,\ldots$ denote
matrices and vectors, respectively.  The element
in the $i$th row and $j$th column of a matrix $\matA$ is
$\matAc_{i j}$ or $[\matA]_{i,j}$, and the $i$th element of a
vector~$\vectr$ is $\vectrc_i$ or $[\vectr]_i$.  For a vector
$\vectr$, $\diag(\vectr)$ stands for the diagonal matrix that has the
entries of $\vectr$ on its main diagonal. 
The vector of all zeros is denoted $\veczero$.
The superscript~$\tp{}$ stands for transposition.
For a finite set $\setI$, we write $\abs{\setI}$ for the cardinality.
For $x\in\reals$, $\lceil x\rceil\define
\min\{m\in\integers\mid m\geq x\}$.  We use $\natseg{l}{k}$ to
designate the set of natural numbers $\left\{l, l+1,\ldots,k\right\}$.
For a vector
$\vectr\in\complexset^n$, $\onenorm{\vectr}=\sum_{j=0}^{n-1}
\abs{\vectrc_j}$ denotes the $\ell_1$ norm; $\twonorm{\vectr}=\bigl(\sum_{j=0}^{n-1}
\vectrc_j^2\bigr)^{1/2}$ denotes the $\ell_2$ norm; $\infnorm{\vectr}=\max_{j}\abs{\vectrc_j}$ denotes 
the $\ell_\infty$ norm; and $\zeronorm{\vectr}$ denotes the number of nonzero elements in $\vectr$.
For a function 
$\fun(\cdot):\reals\to\reals$, $\infnorm{\fun(\cdot)}=\max_{t\in\reals}\abs{\fun(t)}$. 
The indicator function is denoted as $\I{\cdot}$,  it is equal to one if the condition in 
the brackets is satisfied and zero otherwise.  We use $c$ with 
various subindexes and superindexes to denote \emph{positive numerical} constants; to track things simpler, we
use the convention that the numerical constants with the subscript $u$, like $c_{u1}$, satisfy $c_{u1}>1$, and
the numerical constants with subscript $l$, like $c_{l1}$, satisfy $0<c_{l1}<1$. 
Throughout the paper we use the convention:
$\flo$ denotes the frequency cut-off of the measured data [see \fref{eq:spec}], $\llo=1/\flo$ is the corresponding wavelength; $\fc$ denotes an abstract frequency cut-off (this value changes in different places in the paper) and $\lc=1/\fc$ is the corresponding wavelength. To simplify writing, we follow the conventions:
$\prod_{i=1}^r a_i=1$ and $\{a_1,\ldots, a_r\}=\varnothing$ when $r=0$.

\section{Structure of the proof}
Previous results in the field \cite{candes13, candes21-2, candes-14} suggest that \fref{thm:mainthm} may be proven by constructing an appropriate dual certificate. Since the 
measurement operator is a low-pass kernel, the dual certificate for this problem is a real-valued 
trigonometric polynomial frequency-limited to $\flo$ with additional properties. 
In fact, similar to~\cite{candes21-2}, we will need three trigonometric polynomials instead of one, each with its 
own properties; they will be called $q_0(\cdot)$, $q_1(\cdot)$, and $q_2(\cdot)$. 
These dual trigonometric polynomials are constructed in Lemmas~\ref{lem:dualprod}, \ref{lem:dualprod1}, 
and \ref{lem:dualprod1d} in \fref{sec:dualcerts}; $q_0(\cdot)$ is borrowed from~\cite{candes-14}, 
$q_1(\cdot)$ and $q_2(\cdot)$ are new---they are the main mathematical contribution of this paper. 
In \fref{sec:stability} we use
$q_0(\cdot)$, $q_1(\cdot)$, and $q_2(\cdot)$ to derive the stability estimates and prove \fref{thm:mainthm}.

We invite the reader unfamiliar with the concept of dual certificates in convex optimization 
to study the short proof of \cite[Lm.~1]{candes-14} before reading this paper further. The 
derivations in \fref{sec:stability} generalize \cite[Lm.~1]{candes-14} to the arbitrarily fine 
grid setting, but they are much more involved.

Some calculations in this paper are complicated, but we tried to present the key new ideas in a simple way.
At the first pass through the paper we 
suggest that the reader studies Sections~\ref{sec:build}--\ref{sec:subqz}; then
focuses on the formulations of Lemmas \ref{lem:dualprod1} and \ref{lem:dualprod1d} and the new constructions in \fref{sec:dualprod1m} and in \fref{sec:const2}; skips the details in Sections~\ref{sec:propphizk}--\ref{sec:bndfarV} and in Sections~\ref{sec:existencephipk}--\ref{sec:proofdualprod14}; and finally studies the stability estimates in \fref{sec:stability}. After this, return to the technical details in Sections~\ref{sec:propphizk}--\ref{sec:bndfarV} and in Sections~\ref{sec:existencephipk}--\ref{sec:proofdualprod14}.

\section{Dual certificates}
\label{sec:dualcerts}

Throughout the paper we will use the following definitions.
Define the error vector
\begin{equation}
    \vech = \tp{[h_0, \ldots, h_{\N-1}]}\define\hat\inp - \inp
\end{equation}
and the set of points where the error vector takes on negative values
\begin{equation}
    \label{eq:setT}
    \setT = \{t_1, \ldots, t_\sparsity\}\define\{m/\N: h_m<0\}.
\end{equation}
The points are ordered according to $t_1<\ldots<t_\sparsity$.
Recall, $\hat\inp\ge\veczero$ and $\inp\ge\veczero$. Therefore, $h_m$ can only take on 
negative values on $\support(\inp)$, which implies $\setT\subset\support(\inp)$. 
Since $\support(\inp)\in \rset{\kappa \llo\rr}{\rr}$ and 
since the elements of $\support(\inp)$ are separated 
by at least $2\lhi$, it follows $\setT\in\rset{\kappa\llo\rr}{\rr}$ and the elements 
of $\setT$ are also separated by at least $2\lhi$. As we will see below, the dual 
trigonometric polynomials $q_0(\cdot)$, $q_1(\cdot)$, and $q_2(\cdot)$ need to  satisfy specific interpolation 
constraints on $\setT$.

\label{sec:neigh}
Throughout the paper we will use the following neighborhood notations.
\begin{dfn}
\label{dfn:neigh}
For $\tau\in\T$, $\delta>0$,
\begin{equation}
    \near{\delta}{\tau}\define \{t\in\T: \abs{t-\tau}\le \delta \},
\end{equation}
where $\abs{\cdot}$ denotes the wrap-around distance on $\T$. Above, $\near{\cdot}{\cdot}$ stands for ``near'' (i.e., the points near $\tau$).

For a set $\setV\subset\T$ and $\delta>0$,
\begin{align}
    \near{\delta}{\setV}&\define\union_{\tau\in\setV} \near{\delta}{\tau},\\
    \far{\delta}{\setV}&\define\T\setdiff \near{\delta}{\setV}.
\end{align}
Above, $\far{\cdot}{\cdot}$ stands for ``far'' (i.e., the points far from $\setV$).
\end{dfn}

\subsection{Building blocks}
\label{sec:build}
The following two lemmas serve as common building blocks for the construction of trigonometric 
polynomials $q_0(\cdot)$, $q_1(\cdot)$, and $q_2(\cdot)$.

\fref{lem:dualcarlos} allows us to construct a trigonometric polynomial frequency-limited to $\fc$ that 
\emph{interpolates zeros} at well-separated points as illustrated in \fref{fig:toolsmult2n1}.
\begin{figure}[ht]
    \begin{subfigure}{0.5\textwidth}
        \centering
        \caption{}\vspace{-0.3cm}
        \label{fig:toolsmult2n1}
        \includegraphics[width=0.9\textwidth]{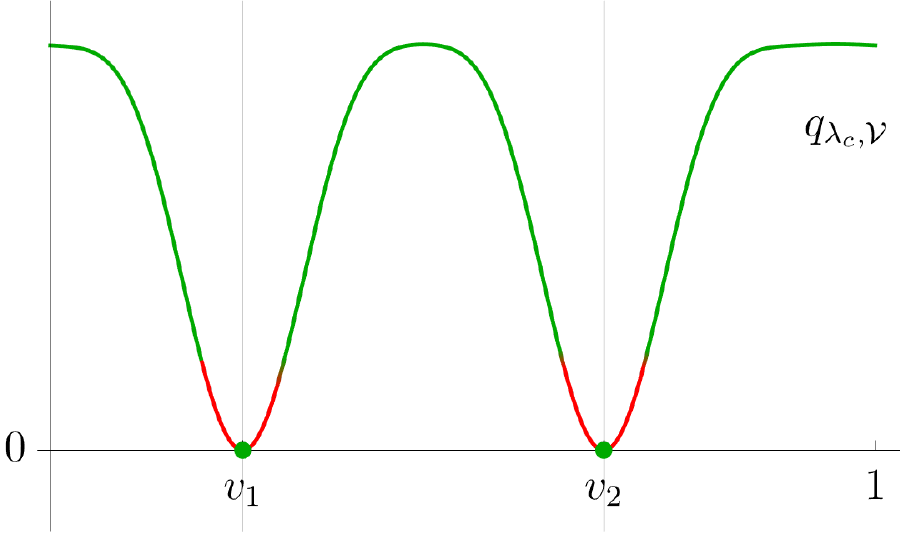}
    \end{subfigure}
    \begin{subfigure}{0.5\textwidth}
        \centering
        \caption{}\vspace{-0.3cm}
        \label{fig:toolsmult2n2}
        \includegraphics[width=0.9\textwidth]{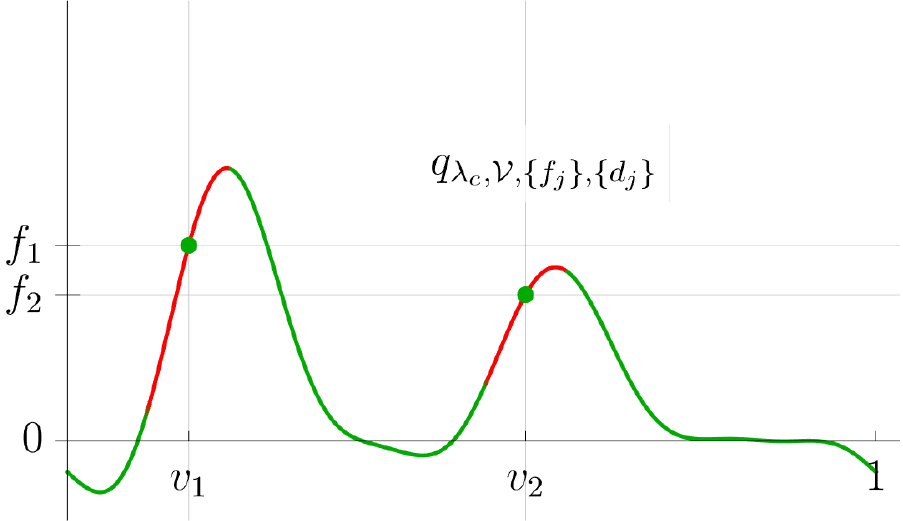}
    \end{subfigure}
    \caption{(a) Illustration of \fref{lem:dualcarlos}. Trigonometric polynomial 
    frequency-limited to $\fc=6$
    interpolates zeros at well-separated points $\{v_1, v_2\}\in \rset{2.5\lc}{1}$. 
    Specifically, $q_{\lc, \setV}(v_j)=q_{\lc, \setV}'(v_j)=0$ 
    and the curvature in the neighborhoods of $v_1$ and $v_2$ is controlled (indicated in red) 
    according to~\fref{eq:qnearbndeq}. (b) Illustration of \fref{lem:dualcarlos2}. 
    Trigonometric polynomial frequency-limited to $\fc=6$
    interpolates values $f_1$ and $f_2$ at well-separated 
    points $\{v_1, v_2\}\in \rset{2.5\lc}{1}$.
    Specifically, $q_{\lc, \setV, \{f_j\}, \{d_j\}}(v_j)=f_j$ and the derivatives at
    $v_1$ and $v_2$ are constrained (indicated in red) according 
    to~$q'_{\lc, \setV,\{f_j\}, \{d_j\}}(v_j)=d_j$.}
    \label{fig:toolsmult2n}
\end{figure}

\begin{lem}
\label{lem:dualcarlos}
Let $\lc\in (0,1/128)$, set $\fc\define 1/\lc$. Consider a collection of points 
$v_1< v_2< \ldots< v_{V}$, define $\setV\define\{v_1, v_2, \ldots, v_{V}\}$ and 
assume~$\setV\in\rset{\kappa\lc}{1}$. Then, there exists a real-valued trigonometric 
polynomial $q(\cdot) = q_{\lc,\setV}(\cdot)$ that satisfied the following properties.
\begin{enumerate}
    \item \label{prop:lf} Frequency limitation to $\fc$: 
    $q(t)=\sum_{k=-\fc}^{\fc} \hat q_k e^{-\iu 2\pi kt}$ for some $\hat q_k\in\complexset$.
    \item \label{prop:qzero} Zero values and zero derivatives on $\setV$: 
        for all $v\in\setV$, $q(v)=q'(v)=0$.
    \item \label{prop:qzeroone} Uniform confinement between zero and one: 
        for all $\tau\in\reals$, $0\le q(\tau)\le 1$.
    \item \label{prop:qnearbnd}Quadratic behavior near $\setV$: for all $v\in\setV$ 
        and for all $\tau\in\near{\D \lc}{v}$
    \begin{equation}
    \frac{\cql (v-\tau)^2}{\lc^2}\le q(\tau) \le \frac{\cqu (v-\tau)^2}{\lc^2}.
    \label{eq:qnearbndeq}
    \end{equation}
    \item \label{prop:qinfbnd} Boundedness away from zero far from $\setV$: for all 
        $\tau\in\far{\D \lc}{\setV}$, $q(\tau) \ge \cqlf>0$.
    \item \label{prop:qpbnd} Uniform confinement of the derivative: $\infnorm{q'(\cdot)}\le 2\pi/\lc$.
    \item \label{prop:qppbnd} Uniform confinement of the second derivative: $\infnorm{q''(\cdot)}\le 4\pi^2/\lc^2$.
\end{enumerate}
Above,  all the constants are positive numerical constants. Specifically,
\begin{alignat}{2}
    &\kappa\define 1.87, &&\D\define 0.17,\\
    &\cql\define 0.029, &&\cqu\define 2\pi^2,\\
    &\cqlf\define \D^2 \cql = \num{8.3e-4}. \label{eq:cqlfcnear}
\end{alignat}
\end{lem}

\begin{proof}
This lemma is a direct consequence of the technique developed in~\cite{candes13}.
Let $q_{CFG}(\cdot)$ denote the trigonometric polynomial constructed as in~\cite[eq.~(2.4)]{candes13} 
in order to interpolate $-1$ on $\setV$. Then, according to~\cite[Lm.~2.4, Lm.~2.5, Sec.~2.5]{candes13}, 
$q(\cdot) = 0.5(q_{CFG}(\cdot)+1)$ satisfies Properties~\ref{prop:lf}, \ref{prop:qzero}, \ref{prop:qzeroone}, 
\ref{prop:qinfbnd} of the lemma, and the lower bound in~\fref{eq:qnearbndeq}.
Since, by \fref{prop:qzeroone}, $\infnorm{q(\cdot)}\le 1$, Properties~\ref{prop:qpbnd} 
and~\ref{prop:qppbnd} follow by applying \fref{eq:bernstein} [Bernstein theorem]. 
Finally, the upper bound in~\fref{eq:qnearbndeq} follows from \fref{prop:qzero} 
and \fref{prop:qppbnd} by~\fref{eq:mvt2} [Mean Value theorem].
\end{proof}

\fref{lem:dualcarlos2} allows us to construct a trigonometric 
polynomial frequency-limited to~$\fc$ that \emph{interpolates arbitrary values} and has constrained derivatives
at well-separated points as illustrated in \fref{fig:toolsmult2n2}.
\begin{lem}
\label{lem:dualcarlos2}
Let $\lc\in (0,1/128)$, set $\fc\define 1/\lc$. 
Consider a collection of points $v_1< v_2< \ldots< v_{V}$, define 
$\setV\define\{v_1, v_2, \ldots, v_{V}\}$ and assume~$\setV\in\rset{\kappa\lc}{1}$.
Consider two sets of real numbers $\{\fvalc_1, \fvalc_2,\ldots, \fvalc_\npoints\}$ 
and $\{\dvalc_1, \dvalc_2, \ldots, \dvalc_\npoints\}$ that satisfy
\begin{equation}
    \abs{\fvalc_j}\le 1\quad \text{and}\quad \abs{\dvalc_j}\le \frac{1}{\lc}\label{eq:condp2}
\end{equation}
for all $j=1,\ldots, \npoints$.
Then, there exists a real-valued trigonometric polynomial 
$q(\cdot)=q_{\lc,\setV, \{\fvalc_j\}, \{\dvalc_j\}}(\cdot)$ that satisfies
the following properties.
\begin{enumerate}
    \item \label{prop:lf2} Frequency limitation 
    to $\fc$: $q(t)=\sum_{k=-\fc}^{\fc} \hat q_k e^{-\iu 2\pi kt}$ for some $\hat q_k\in\complexset$.
    \item \label{prop:rhoint} Constrained values and derivatives on $\setV$: for all $j=1,\ldots, \npoints$, 
    \begin{equation}
        q(v_j)=\fvalc_j\quad \text{and}\quad q'(v_j)=\dvalc_j.
    \end{equation}
    \item \label{prop:dqinfbnd} Uniform confinement: $\infnorm{q(\cdot)}\le \cqdv$.
    \item \label{prop:qpbnd1} Uniform confinement of the derivative: $\infnorm{q'(\cdot)}\le \cqdvd/\lc$.
    \item \label{prop:qppbnd1} Uniform confinement of the second derivative: 
        $\infnorm{q''(\cdot)}\le \cqdvdd/\lc^2$.
\end{enumerate}
Above, $\cqdv$, $\cqdvd$, and $\cqdvdd$ are  positive numerical constants that are defined in 
the proof of the lemma in~\fref{app:dp}.
\end{lem}

The proof of the lemma generalizes the results in~\cite[Lm.~2.4, Lm.~2.5, Sec.~2.5]{candes13} 
slightly in several technical aspects; it is given in \fref{app:dp} for completeness.

\subsection{Dual certificate $q_0(\cdot)$}
\label{sec:subqz}
We are now ready to construct the trigonometric polynomial $q_0(\cdot)$.
This trigonometric polynomial, illustrated in \fref{fig:toolsmult2}, is frequency-limited 
to $\flo$, interpolates zeros on a Rayleigh-regular set, is confined between 
zero and one, and quickly grows around its zeros.

\begin{figure}[ht]
    \centering
    \includegraphics[width=0.6\textwidth]{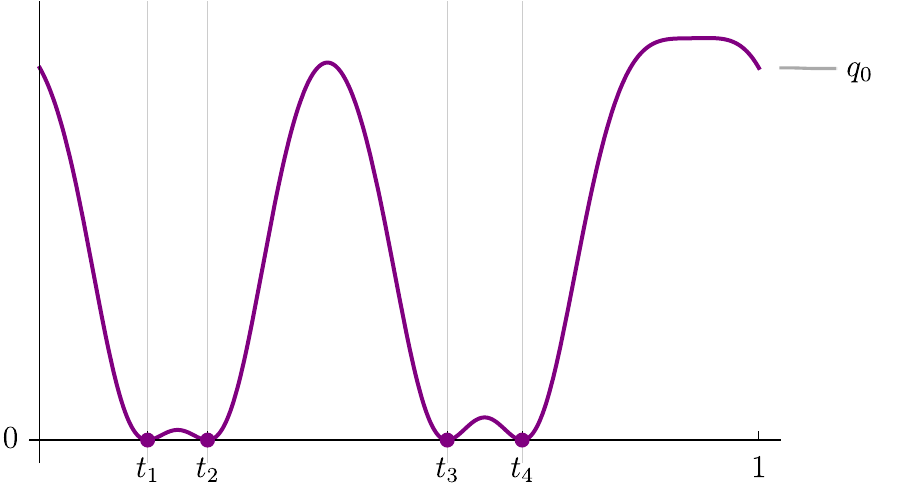}
    \caption{Illustration of \fref{lem:dualprod}. 
    Trigonometric polynomial frequency-limited to $\flo=12$ interpolates zeros on 
    Rayleigh-regular set $\setT=\{t_1, t_2, t_3, t_4\}\in\rset{5\llo}{2}$ and 
    bounces away from zeros ``quickly'': the curvature in the neighborhoods of each 
    point $t_i$ is ``high'' in the sense of~\fref{eq:qzl}. In the figure, 
    $\abs{t_3-t_1}\ge 5\llo=5/12$, $\abs{t_4-t_2}\ge 5\llo=5/12$,
    $\abs{t_2-t_1}\sim 2\lhi$, $\abs{t_3-t_4}\sim 2\lhi$.}
    \label{fig:toolsmult2}
\end{figure}

The key difference between the trigonometric polynomial $q_0(\cdot)$ and the building 
block $q_{\lc,\setV}(\cdot)$ constructed in \fref{lem:dualcarlos} is that the points where $q_0(\cdot)$ 
must take zero values may belong to a Rayleigh-regular set from a 
class $\rset{d}{r}$ with $r>1$. Zeros of $q_0(\cdot)$ may be close, whereas 
zeros of $q_{\lc,\setV}(\cdot)$ are well-separated (compare \fref{fig:toolsmult2n1}
to \fref{fig:toolsmult2}). This is the reason why the technique of~\cite{candes13} 
and \cite{candes21-2} that was used to prove \fref{lem:dualcarlos} cannot be applied directly to construct~$q_0(\cdot)$.

\begin{lem}
\label{lem:dualprod}
There exists a real-valued trigonometric polynomial $q_0(\cdot)$ that satisfies the following properties.
\begin{enumerate}
    \item \label{prop:qlowf} Frequency limitation to $\flo$: 
    $q_0(t)=\sum_{k=-\flo}^{\flo} \hat q_{0,k} e^{-\iu 2\pi kt}$ for some $\hat q_{0,k}\in\complexset$.
    \item \label{prop:qzeroprod} Zero values and zero derivatives on $\setT$: 
        for all $t\in \setT$, $q_0(t)=q_0'(t)=0$. 
    \item \label{prop:zeroonebnd} Uniform confinement between zero and one: for all $\tau\in\reals$, $0\le q_0(\tau)\le 1$. 
    \item \mbox{Controlled  behavior near $\setT$: Take $\tau\in\near{\rr\D\llo}{\setT}$. 
        Let
        $\{\vt_1,\ldots,\vt_{\hrr}\}\define \near{\rr\D\llo}{\tau}\intersect\setT$.}
        [Note: since $\setT\in \rset{\rr\kappa\llo}{\rr}$ and $\D<\kappa$, it follows that $1\le\hrr\le\rr$.] Set
        $\vtn\define\argmin_{v\in\{\vt_1,\ldots,\vt_{\hrr}\}} \abs{v-\tau}$.
        Then, the following estimates hold.
        \begin{enumerate}
        \item Lower bound:
        \begin{align}
        q_0(\tau)
        &\ge \cz^\rr \frac{ \prod_{l=1}^{\hrr}(\vt_l-\tau)^2}{(\rr\llo)^{2\hrr}}
            \label{eq:qzlp}\\
        &\ge \cz^\rr  \frac{ (\vtn-\tau)^2 \lhi^{2(\rr-1)}}{(\rr\llo)^{2\rr}}. 
            \label{eq:qzl}
        \end{align}
        \item Upper bound:
        \begin{align}
            q_0(\tau) 
            &\le 
            \cqu^{\hrr} \frac{ \prod_{l=1}^{\hrr} (\vt_l-\tau)^{2}}{(\rr\llo)^{2\hrr}}.
            \label{eq:qzlunivbnd}
        \end{align}
    \end{enumerate}
    \item \label{prop:q0lbfar} Boundedness away from zero far from $\setT$: for 
    all $\tau \in\far{\rr\D\llo}{\setT}$,
    \begin{equation}
        q_0(\tau) \ge \cqlf^{\rr}>0.
        \label{eq:q0lbxx}
    \end{equation}
    \item \label{prop:qzloinfhi} Fast growth immediately away from $\setT$: 
    for all $\tau \in\far{\lhi}{\setT}$,
    \begin{equation}
        q_0(\tau)\ge \cql^\rr \frac{\lhi^{2\rr}}{(\rr\llo)^{2\rr}}.
    \end{equation}
    \end{enumerate}
    Above, $\cz$ is a positive numerical constant, 
    defined in the proof below.
\end{lem}

The trick to prove this lemma is the main contribution of the earlier paper~\cite{candes-14}. 
The key observation is the following. It is possible to construct the nonnegative 
trigonometric polynomial~$q_0(\cdot)$ frequency-limited to $\flo$ that is zero on all the points of the 
set $\setT\in\rset{\rr\kappa\llo}{\rr}$ as a product of $\rr$ trigonometric polynomials. 
Each of these trigonometric polynomials is zero on a set that belongs to $\rset{\kappa\llo\rr}{1}$ and is 
constructed via \fref{lem:dualcarlos}. We reproduce the proof below because it motivates the new 
construction in \fref{sec:q1}.
\begin{proof} 
\label{sec:lemdualprodproof}
Set 
\begin{equation}
    \label{eq:setTk}
    \setT_k\define\Bigl\{t_{j\rr+k}: j\in\natseg{0}{\floor{(\sparsity-1)/\rr}}\Bigr\},\quad k=1,\ldots,\rr.
\end{equation}
Observe  that $\setT = \setT_1\cup\ldots \cup\setT_r$ and $\setT_k\in\rset{\kappa\llo\rr}{1}$. Set
\begin{equation}
    q_0(t)\define q_{\rr\llo,\setT_1}(t)\times\cdots\times q_{\rr\llo,\setT_\rr}(t),
    \label{eq:q0asprod}
\end{equation}
where $q_{\rr\llo,\setT_k}(\cdot),\ k=1,\ldots,\rr,$ are the trigonometric polynomials 
constructed\footnote{Strictly speaking this requires that the frequency limitation of $q_{\rr\llo,\setT_k}(\cdot)$, $\flo/\rr$, is an integer. In the rest of the paper, for simplicity, we will make this additional assumption. If this assumption is not satisfied, we can simply substitute $\flo$ with $\floor{\flo/\rr}\rr$ and repeat all the arguments in the paper, leading only to a small increase in  the density constant $1.87$ in \fref{thm:mainthm}.} 
via \fref{lem:dualcarlos} with $\lc=\rr\llo$ and $\setV=\setT_k\in \rset{\kappa\llo\rr}{1}$. The idea of this construction for $r=2$ is illustrated in \fref{fig:proof}.
\renewcommand\cxi{0.1196}
\renewcommand\cxii{0.1739} \renewcommand\cxiii{0.5000}
\newcommand\cxiv{0.5761} \newcommand\cxv{0.8913}
\begin{figure}[ht]
    \centering
    \includegraphics[width=0.6\textwidth]{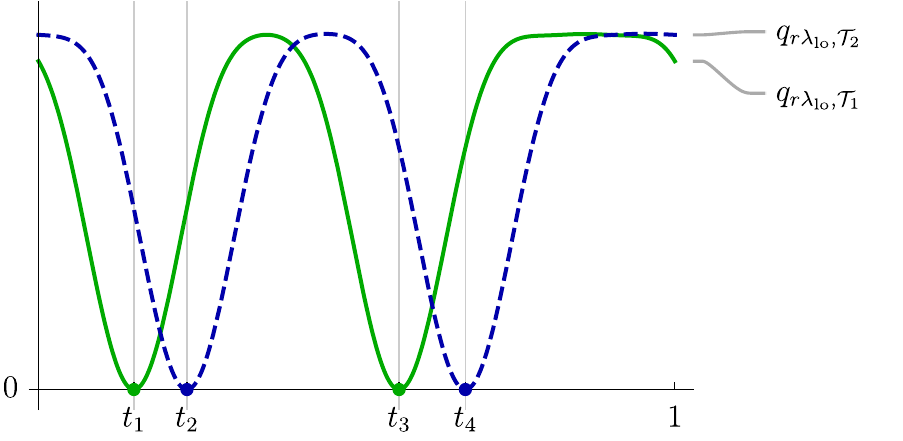}
    \caption{Illustration of the proof of \fref{lem:dualprod}. The 
    set $\setT=\{t_1, t_2, t_3, t_4\}$ is Rayleigh-regular: $\setT\in\rset{5\llo}{2}$,  
    with $\rr=2$ and $\llo=1/12$. The idea is to split this set as $\setT=\setT_1\union\setT_2$
    with $\setT_1=\{t_1,t_3\}$ and $\setT_2=\{t_2,t_4\}$ and observe $\setT_i\in \rset{5\llo}{1}$. 
    The trigonometric polynomials are frequency-limited to $\flo/2=6$ and 
    satisfy the interpolation constraints $q_{\rr\llo,\setT_1}(t)=q_{\rr\llo,\setT_1}'(t)=0$ 
    for all $t\in \setT_1$ and $q_{\rr\llo,\setT_2}(t)=q_{\rr\llo,\setT_2}'(t)=0$ for 
    all $t\in \setT_2$. Then, $q_0(\cdot)=(q_{\rr\llo,\setT_1}\times q_{\rr\llo,\setT_2})(\cdot)$ 
    satisfies $q_0(t)=q_0'(t)=0$ for all $t\in\setT$ and is frequency-limited to $2\times \flo/2=12$. 
    The trigonometric polynomial $q_0(\cdot)$ is displayed in \fref{fig:toolsmult2}. In the figure, $\abs{t_2-t_1}\sim 2\lhi$, $\abs{t_3-t_4}\sim 2\lhi$.}
    \label{fig:proof}
\end{figure}

It remains to verify that Properties \ref{prop:qlowf}--\ref{prop:qzloinfhi} are satisfied. 
Broadly, this follows from~\fref{eq:q0asprod} and \fref{lem:dualcarlos}; the details are given below.

\fref{prop:qlowf} is satisfied because each of trigonometric 
polynomials $q_{\rr\llo,\setT_k}(\cdot),\ k=1,\ldots,\rr$ is frequency-limited 
to $\flo/\rr$.
Hence, the product in~\fref{eq:q0asprod} is frequency-limited to $\rr(\flo/\rr)=\flo$.

Properties \ref{prop:qzeroprod} and \ref{prop:zeroonebnd} follow from~\fref{eq:q0asprod} 
and from \fref{lem:dualcarlos}, Properties~\ref{prop:qzero} and~\ref{prop:qzeroone}, respectively.

To prove~\fref{eq:qzlp} we lower-bound the terms in~\fref{eq:q0asprod} separately as 
follows. Assume that $k\in\{1,\ldots,\rr\}$ is such
that $\near{\rr\D\llo}{\tau}\intersect\setT_k\ne\varnothing$, i.e., 
there exist $l\in\{1,\ldots,\hat\rr\}$ that satisfies $\vt_l\in\setT_k$. In this 
case, we use the left-hand side of~\fref{eq:qnearbndeq} to write
\begin{equation}
    \label{eq:qIlb1}
    q_{\rr\llo,\setT_k}(\tau)\ge \cql\frac{(\vt_l-\tau)^2}{(\rr\llo)^2}.
\end{equation}
Note that there are exactly $\hrr$ such terms in~\fref{eq:q0asprod}. Assume that $k\in\{1,\ldots,\rr\}$ 
is such that $\near{\rr\D\llo}{\tau}\intersect\setT_k=\varnothing$. In this 
case, use \fref{lem:dualcarlos}, \fref{prop:qinfbnd}, to write 
\begin{equation}
    \label{eq:qIlb2}
    q_{\rr\llo,\setT_k}(\tau)\ge\cqlf.
\end{equation}
Note that there are exactly $\rr-\hrr$ such terms in~\fref{eq:q0asprod}. The desired 
bound~\fref{eq:qzlp} is obtained by plugging~\fref{eq:qIlb1} and~\fref{eq:qIlb2} 
into~\fref{eq:q0asprod} and setting $\cz\define\min(\cql, \cqlf)$.

Bound~\fref{eq:qzl} follows because the elements of $\setT$ are separated 
by at least $2\lhi$ and because $\lhi/\llo<1$.

To prove~\fref{eq:qzlunivbnd} we upper-bound the terms in~\fref{eq:q0asprod} separately as 
follows. Assume that $k\in\{1,\ldots,\rr\}$ is such that 
$\near{\rr\D\llo}{\tau}\intersect\setT_k\ne\varnothing$, i.e., 
there exist $l\in\{1,\ldots,\hat\rr\}$ that satisfies $\vt_l\in\setT_k$. In this case, we use the
right-hand side of~\fref{eq:qnearbndeq} to write
\begin{equation}
    \label{eq:qubrwc}
    q_{\rr\llo,\setT_k}(\tau)\le \cqu\frac{(\vt_l-\tau)^2}{(\rr\llo)^2}.
\end{equation}
Assume that $k\in\{1,\ldots,\rr\}$ is such that $\near{\rr\D\llo}{\tau}\intersect\setT_k=\varnothing$.
In this case, we use
\fref{lem:dualcarlos}, \fref{prop:qzeroone}, to write 
\begin{equation}
    \label{eq:qIu1}
    q_{\rr\llo,\setT_k}(\tau)\le 1.
\end{equation}
The desired bound~\fref{eq:qzlunivbnd} is obtained by plugging~\fref{eq:qubrwc} 
and~\fref{eq:qIu1} into~\fref{eq:q0asprod}.

\fref{prop:q0lbfar} follows by~\fref{eq:q0asprod} and \fref{lem:dualcarlos}, \fref{prop:qinfbnd}.

Finally, \fref{prop:qzloinfhi} follows  from~\fref{eq:q0asprod}, \fref{eq:qnearbndeq}, 
\fref{lem:dualcarlos}, \fref{prop:qinfbnd}, and~\fref{eq:cqlfcnear}.
\end{proof}

\subsection{Dual certificate $q_1(\cdot)$}
\label{sec:q1}
We are now ready to construct the trigonometric polynomial $q_1(\cdot)$. This construction 
and its analysis is the main mathematical contribution of this paper.
Trigonometric polynomial~$q_1(\cdot)$, illustrated in \fref{fig:toolsneed}, is 
frequency-limited to $\flo$ and, on the 
points $t_j\in\setT$, $q_1(\cdot)$ interpolates the set of signs
\begin{equation}
    \label{eq:dsdef}
    s_{j}\define\sign\lefto(\sum_{m/\N\in\nearhi{t_j}} h_m\right),\ j=1,\ldots,\sparsity,
\end{equation}
at a (low) level $\rho/2$, $\rho=(\lhi/\llo)^{2\rr}\ll 1$. The behavior of $q_1(\cdot)$ is controlled by $q_0(\cdot)$ as explained in \fref{lem:dualprod1} below.
\begin{figure}[ht]
    \centering
    \vspace{-0.2cm}
    \includegraphics[width=0.6\textwidth]{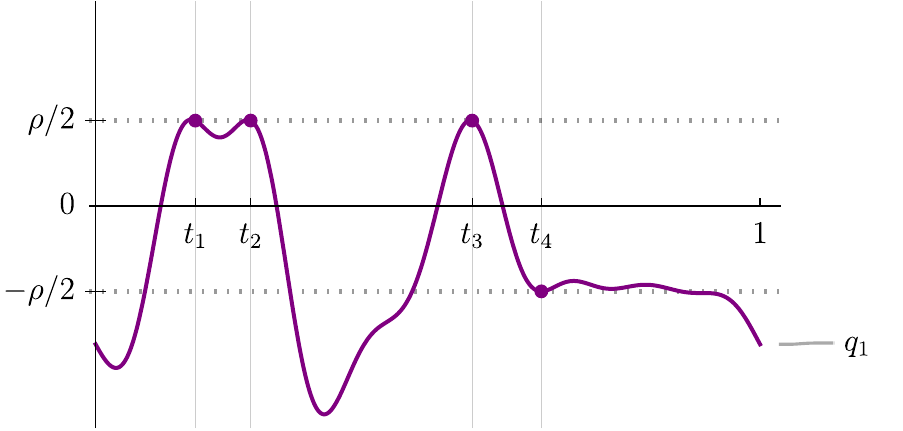}
    \vspace{-0.2cm}
    \caption{Illustration of \fref{lem:dualprod1}. 
    Trigonometric polynomial frequency-limited to $\flo=12$ interpolates the 
    sign pattern $\{s_1, s_2, s_3, s_4\} = \{+1, +1, +1, -1\}$ at a (low) level $\rho/2$. 
    Specifically, $q_1(t_i)=s_i \rho/2$ and $q'_1(t_i)=0$. The set $\setT=\{t_1, t_2, t_3, t_4\}$ 
    is Rayleigh-regular: $\setT\in\rset{5\llo}{2}$ with $\llo=1/\flo$ and $\abs{t_2-t_1}\sim 2\lhi$, $\abs{t_3-t_4}\sim 2\lhi$.}
    \label{fig:toolsneed}
\end{figure}

\begin{lem}
\label{lem:dualprod1}
Set $\rho\define\lhi^{2\rr}/\llo^{2\rr}$.
Then, there exists a real-valued trigonometric polynomial $q_1(\cdot)$
that satisfies the following properties.
\begin{enumerate}
    \item \label{prop:q1prop1} Frequency limitation to $\flo$: 
        $q_1(t)=\sum_{k=-\flo}^{\flo} \hat q_{1,k} e^{-\iu 2\pi kt}$ 
        for some $\hat q_{1,k}\in\complexset$.
    \item \label{prop:q1prop3} Constrained sign pattern (at level $\rho$) on $\setT$ and controlled behavior 
        near $\setT$: for all $j=1,\ldots,\sparsity$ and all $\tau\in\near{\lhi}{t_j}$,
        \begin{equation}
            \label{eq:bndnear}
            \abs{q_1(\tau)-\frac{\rho s_j}{2}}
            \le  \rr^{2\rr+4} \cuff^{\rr+1}  q_0(\tau),
        \end{equation}
        where $s_j$ are defined\footnote{The lemma is valid for an arbitrary sign pattern, we formulate it for the sign pattern defined in~\fref{eq:dsdef} for concreteness.} in~\fref{eq:dsdef}.
    \item \label{prop:q1prop2}
    Uniform confinement: $\infnorm{q_1(\cdot)}\le \rr^{2\rr+1}\cugm^{\rr}$.
    \item \label{prop:q1prop4} Boundedness far from $\setT$: for all $\tau\in\far{\lhi}{\setT}$,
    \begin{equation}
    \label{eq:dualprod14}
    \abs{q_1(\tau)}\le \rr^{2\rr+2} \cuffff^\rr q_0(\tau),
    \end{equation}
\end{enumerate}
The positive numerical constants $\cuff$, $\cugm$, and $\cuffff$ are defined in the proof below.
\end{lem}

\paragraph{Discussion.}
Let's compare $q_1(\cdot)$ illustrated in \fref{fig:toolsneed} to
$q_{\lc,\setV, \{\fvalc_j\}, \{\dvalc_j\}}(\cdot)$ 
constructed in \fref{lem:dualcarlos2} and illustrated in \fref{fig:toolsmult2n2}.
In $q_{\lc,\setV, \{\fvalc_j\}, \{\dvalc_j\}}(\cdot)$, the behavior at a well-separated set of 
points is independently controlled: 
the trigonometric polynomial can take arbitrary values (between $-1$ and $1$). 
Reminder: we say that the points are well-separated if the distances between the points are no smaller 
than $\sim \tcc/\fc$, where $\fc$ is the frequency limitation of the trigonometric polynomial 
under consideration and $\tcc$ is a bit larger than 1. In the case of $q_1(\cdot)$, the points where the behavior is 
controlled are \emph{not} well-separated as  illustrated on \fref{fig:toolsneed}: $\abs{t_2-t_1}\sim 2\lhi\ll 1/\flo$, $\abs{t_3-t_4}\sim 2\lhi\ll 1/\flo$. Therefore, 
by Bernstein theorem (see~\fref{thm:bernstein}), the behavior of $q_1(\cdot)$ at nearby points 
cannot be controlled~\emph{independently}. To be concrete: suppose we require 
that $q_1(t_1)=-1$ and $q_1(t_2)=+1$. Since the points $t_1$ and $t_2$ are separated by 
about $2\lhi\ll \llo$ (not well-separated), Bernstein theorem says that these two 
requirements cannot be satisfied simultaneously. Indeed, since $\infnorm{q_1(\cdot)}\le \tCC{\rr}=\rr^{2\rr+1}\cugm$, 
by~\fref{eq:bernstein}, $\infnorm{q_1'(\cdot)}\le 2\pi\tCC{\rr}\flo$. If the two requirement would be satisfied 
simultaneously, the derivative of $q_1(\cdot)$ between the points $t_1$ and $t_2$ would be about $(q_1(t_2)-q_1(t_1))/(2\lhi)=2/(2\lhi)=\fhi\gg 2\pi\tCC{\rr}\flo$ (we are assuming that $\SRF$ is large). However, if the we require that $q_1(t_1)=-\rho$ 
and $q_1(t_2)=+\rho$ and $\rho$ is small enough, the two requirements may be satisfied 
simultaneously. This is the reason 
why $\rho$ is set to $\lhi^{2\rr}/\llo^{2\rr}\ll 1$ in the formulation of \fref{lem:dualprod1}.

Let's compare $q_1(\cdot)$ to the trigonometric polynomial $q_0(\cdot)$ constructed in \fref{lem:dualprod} 
and illustrated in \fref{fig:toolsmult2}. In both trigonometric polynomials the behavior is 
controlled on a Rayleigh-regular set, whose points are not well-separated in general. The difference 
is that~$q_0(\cdot)$ takes \emph{the same} value (zero) on all the points of the Rayleigh-regular set. 
This allows us to use the multiplication trick illustrated in \fref{fig:proof} to prove \fref{lem:dualprod}. 
In the case of $q_1(\cdot)$ this does not work because we need to interpolate an \emph{arbitrary} 
sign pattern on the Rayleigh-regular set. A method to resolve this problem, presented next, 
is the main mathematical contribution of this paper.

\begin{proof}
Lemma~\ref{lem:dualprod1} is proven in Sections~\ref{sec:dualprod1m}--\ref{sec:bndfarV} below.

\subsubsection{Construction}
\label{sec:dualprod1m}

We first describe how the trigonometric polynomial $q_1(\cdot)$ is constructed. In Sections~\ref{sec:propphizk}--\ref{sec:bndfarV} we prove that the 
construction is valid and that it satisfies the required Properties \ref{prop:q1prop1}--\ref{prop:q1prop4}. 

Recall, $\setT=\{t_1,\ldots, t_\sparsity\}$ is defined in~\fref{eq:setT} and, as before, define $\setT_k$, $k=1,\ldots,\rr$, as in~\fref{eq:setTk}; remember that $\setT = \setT_1\cup\ldots \cup\setT_r$ and $\setT_k\in\rset{\kappa\llo\rr}{1}$.
Set $\eta_j=\rho (s_j+1)/2$ for $j=1,\ldots,\sparsity$.

We will construct the trigonometric polynomial $q_1(\cdot)$ as a (shifted) \emph{sum} of $\rr$ 
trigonometric polynomials $\{\phi_k(\cdot)\}_{k=1}^\rr$ (see \fref{fig:sumneed}):
\begin{equation}
    \label{eq:q1defm}
    q_1(t)=\sum_{k=1}^\rr\phi_k(t)-\rho/2.
\end{equation}
Each of the trigonometric  polynomials $\{\phi_k(\cdot)\}_{k=1}^\rr$ is frequency-limited to $\flo$,
\begin{equation}
    \label{eq:phifreqconstr}
    \phi_k(t)
        =\sum_{l=-\flo}^{\flo} \hat \phi_{k,l} e^{-\iu 2\pi l t} \quad 
            \text{for some}\quad \hat \phi_{k, l}\in\complexset
\end{equation}
and is constructed separately to satisfy the following interpolation constraints on $\setT$:
\begin{align}
    \phi_k(t_l)&=
    \begin{cases}
    \eta_l, &\text{ if } t_l\in\setT_k,\\
    0, &\text{ if } t_l\in \setT_k^c\define \setT\setminus\setT_k,
    \end{cases}\label{eq:req1m}\\
    \phi'_k(t)&=0 \text{ for all } t\in\setT.\label{eq:req3m}
\end{align}

Constraints~\fref{eq:req1m}, \fref{eq:req3m}, and definition~\fref{eq:q1defm} guarantee that 
for all $l=1,\ldots,\sparsity$
\begin{align}
    q_1(t_l)&=\rho s_l/2,\label{eq:q1int1}\\
    q'_1(t_l)&=0.\label{eq:q1int2}
\end{align}

To develop intuition, observe that~\fref{eq:phifreqconstr} and~\fref{eq:q1defm} guarantee 
that \fref{prop:q1prop1} is satisfied. Further, observe that the interpolation constraints~\fref{eq:q1int1} 
and~\fref{eq:q1int2} are  needed for~\fref{eq:bndnear} to hold because $q_0(t)=q'_0(t)=0$ for all $t\in\setT$.

For $r=2$ the construction is illustrated in \fref{fig:sumneed}. The trigonometric 
polynomials $\phi_1(\cdot)$ and $\phi_2(\cdot)$ are displayed in \fref{fig:sumneed1}; they satisfy the interpolation 
constraints~\fref{eq:req1m} and~\fref{eq:req3m} as  indicated by the points highlighted in 
bold. When we compute $(\phi_1+\phi_2)(\cdot)$ we obtain the trigonometric polynomial displayed in \fref{fig:sumneed2},
which, when shifted down by $\rho/2$, is equal to the desired $q_1(\cdot)$ displayed
in \fref{fig:toolsneed}.

\begin{figure}[h!t]
    \centering
    \vspace{-0.5cm}
    \begin{subfigure}{0.5\textwidth}%
        \caption{}\vspace{-0.6cm}
        \label{fig:sumneed1}
        \includegraphics[width=\textwidth]{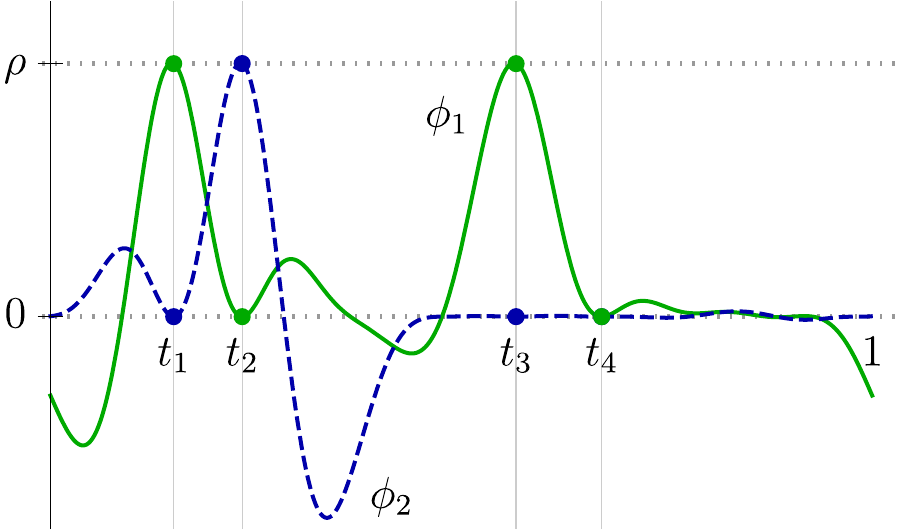}
    \end{subfigure}
    \vspace{0.5cm}
    
    \begin{subfigure}{0.5\textwidth}
        \caption{}\vspace{-0.6cm}
        \label{fig:sumneed2}
        \includegraphics[width=\textwidth]{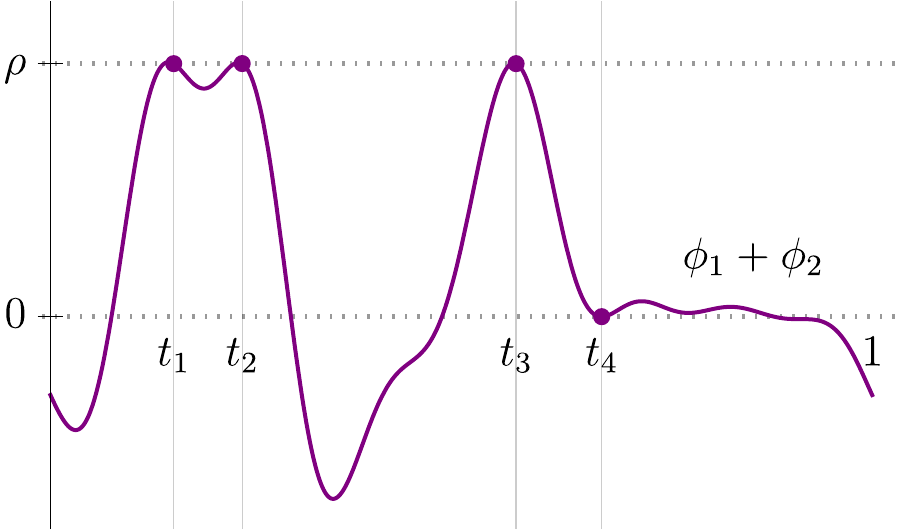}
    \end{subfigure}
    \vspace{-0.3cm}
    \caption{Construction of the trigonometric polynomial $q_1(\cdot)$ (displayed 
    in \fref{fig:toolsneed}) with target sign pattern $\{s_1, s_2, s_3, s_4\} = \{+1, +1, +1, -1\}$. 
    (a) trigonometric polynomials $\phi_1(\cdot)$ and $\phi_2(\cdot)$
    satisfy interpolation 
    constraints~\fref{eq:req1m} and~\fref{eq:req3m} as  indicated by the points highlighted in bold. 
    Specifically, $\phi_1(t_1)=\phi_1(t_3)=\rho$, 
    $\phi_1(t_2)=\phi_1(t_4)=0$, $\phi_2(t_1)=\phi_2(t_3)=\phi_2(t_4)=0$ 
    and $\phi_2(t_3)=\rho$; and further $\phi_i'(t_j)=0$. (b) the sum of $\phi_1(\cdot)$ and $\phi_2(\cdot)$ that, 
    after shifting down by $\rho/2$, is equal to $q_1(\cdot)$. In this figure, $\phi_1(\cdot)$ and $\phi_2(\cdot)$
    are frequency-limited to $\flo=12$;
    $\setT\in\rset{5\llo}{2}$ with $\llo=1/\flo$ is represented 
    as $\setT=\setT_1\union\setT_2$ with $\setT_1=\{t_1, t_3\}\in\rset{5\llo}{1}$, 
    $\setT_2=\{t_2, t_4\}\in\rset{5\llo}{1}$; $\abs{t_2-t_1}\sim 2\lhi$, $\abs{t_3-t_4}\sim 2\lhi$.}
    \label{fig:sumneed}
\end{figure}

The difficulty  remains: how to construct trigonometric polynomials $\phi_k(\cdot)$?
Set 
\begin{equation}
    \setT_k^0\define\Bigl\{t_{j\rr+k}: j\in\natseg{0}{\floor{(\sparsity-1)/\rr}} \text{ and } \eta_{j\rr+k}=0\Bigr\}\ \ \text{and}\ \ \setT_k^+\define\setT_k\setdiff \setT_k^0\ \ \text{for}\ \ k=1,\ldots,\rr.
\end{equation}
The idea now is to construct $\phi_k(\cdot)$ as a 
\emph{product} of two trigonometric polynomials (see \fref{fig:prodd}):
\begin{equation}
    \label{eq:phikdfn}
    \phi_k(t)\define \phi_{0,k}(t)\times \phi_{+,k}(t).
\end{equation}
The first term in the product is defined as
\begin{equation}
    \phi_{0,k}(t)\define \prod_{\substack{1\le l\le r,\ l\ne k}} q_{\rr\llo,\setT_l}(t),
    \label{eq:phizk}
\end{equation}
where $q_{\rr\llo,\setT_l}(\cdot),\ l=1,\ldots,\rr$, are the trigonometric polynomials constructed 
via \fref{lem:dualcarlos} with $\lc=\rr\llo$ and $\setV=\setT_l\in \rset{\kappa\llo\rr}{1}$. 
Observe similarity to the trigonometric polynomial in~\fref{eq:q0asprod}; 
the difference is that the $k$th term is missing from the product. 

\begin{figure}[h!t]\vspace{-0.5cm}
    \begin{subfigure}{0.5\textwidth}
        \centering
        \caption{}\vspace{-0.5cm}
        \label{fig:prodda}
        \includegraphics[width=0.9\textwidth]{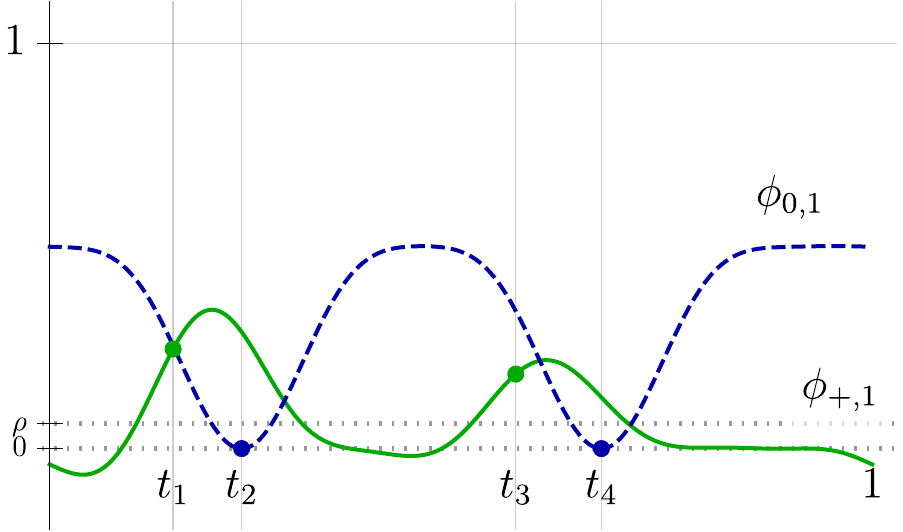}
    \end{subfigure}
    \begin{subfigure}{0.5\textwidth}
        \centering
        \caption{}\vspace{-0.5cm}
        \label{fig:proddb}
        \includegraphics[width=0.9\textwidth]{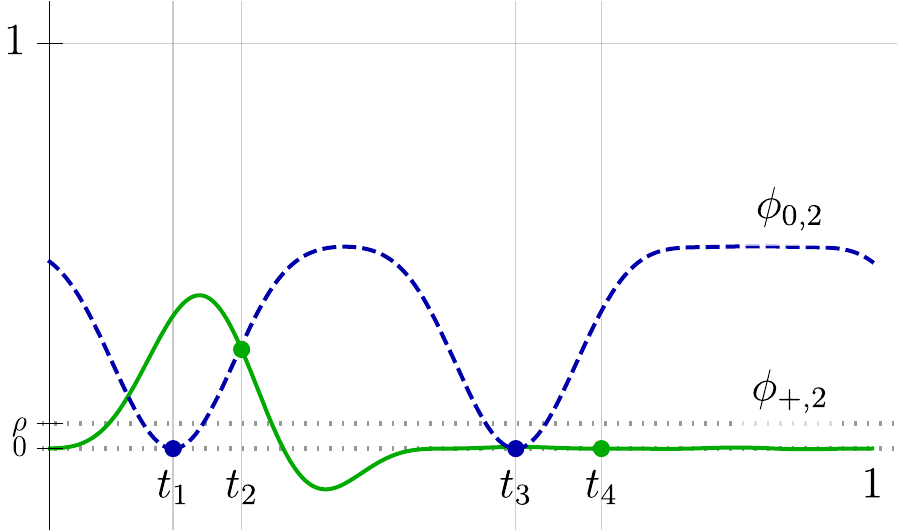}
    \end{subfigure}\vspace{0.5cm}
    
    \begin{subfigure}{0.5\textwidth}
        \centering
        \caption{}\vspace{-0.6cm}
        \label{fig:proddc}
        \includegraphics[width=0.9\textwidth]{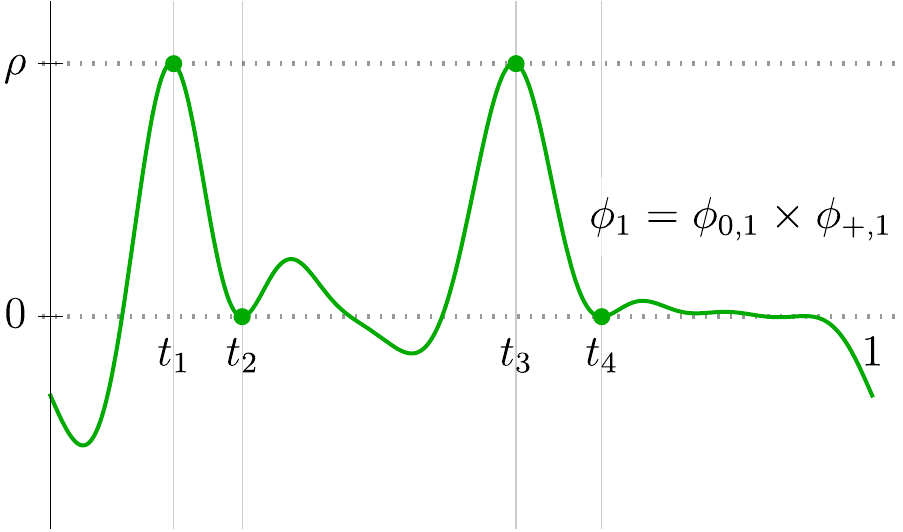}
    \end{subfigure}
    \begin{subfigure}{0.5\textwidth}
        \centering
        \caption{}\vspace{-0.6cm}
        \label{fig:proddd}
        \includegraphics[width=0.9\textwidth]{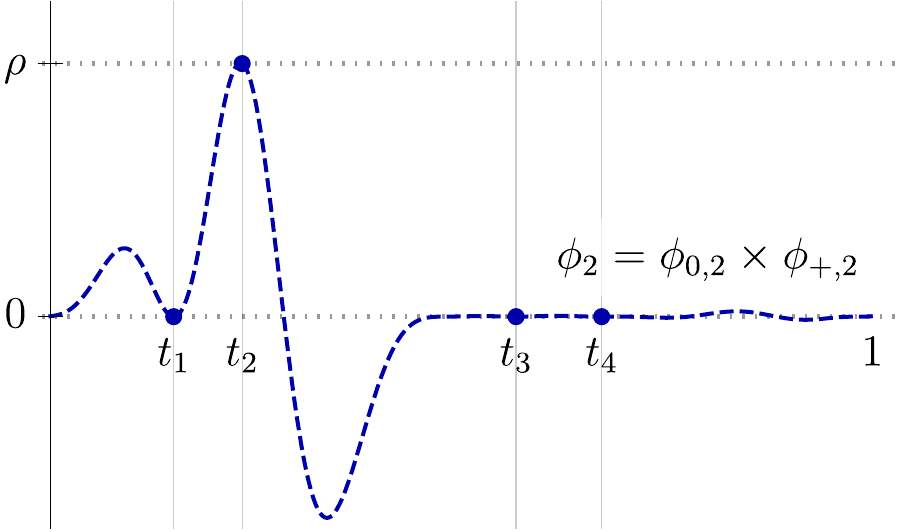}
    \end{subfigure}
    \caption{
    Left column: constructing $\phi_1(\cdot)$ as a product of $\phi_{0,1}(\cdot)$ 
    and $\phi_{+,1}(\cdot)$. The trigonometric polynomial $\phi_{0,1}(\cdot)$ is constrained 
    to take zero values (bold blue points) on $\setT_2=\{t_2, t_4\}$ and it is strictly positive 
    everywhere else. The function values of $\phi_{0,1}(\cdot)$ on 
    $\setT_1=\{t_1, t_3\}$ are unconstrained. The trigonometric polynomial $\phi_{+,1}(\cdot)$, 
    in turn, is only constrained on $\setT_1$ (bold green points). In this case 
    $\setT_1=\setT_1^0\union \setT_1^+$ with $\setT_1^+=\{t_1, t_3\}$ and $\setT_1^0=\varnothing$. 
    The function values and derivatives of $\phi_{+,1}(\cdot)$ are constrained on $\setT_1$ 
    to ``compensate'' for the function values and derivatives of $\phi_{0,1}(\cdot)$ on 
    $\setT_1$ in the sense of~\fref{eq:phi11m} and \fref{eq:phi11m1}. The compensation is such 
    that once the two polynomials are multiplied we obtain $\phi_1(\cdot)$ with the local maxima 
    at level $\rho$ on $\setT_1$ as shown in (c) (bold green points). The local 
    minima of $\phi_1(\cdot)$ on $\setT_2$ are produced ``automatically'', because 
    $\phi_{0,1}(\cdot)$ has zeros on $\setT_2$. Note that the function values and the 
    derivatives of $\phi_{+,1}(\cdot)$ can be controlled at $t_1$ and $t_3$ independently, 
    because these two points are well-separated and this would have been impossible if these 
    points where not well-separated.
    Right column: constructing $\phi_2(\cdot)$ as a product of $\phi_{0,2}(\cdot)$ 
    and $\phi_{+,2}(\cdot)$. The construction is similar, with the roles 
    of $\setT_1$ and $\setT_2$ reversed. The difference is that in this case 
    $\setT_2 = \setT_2^0\union \setT_2^+$ with $\setT_2^+=\{t_2\}$ and 
    $\setT_2^0=\{t_4\}$. Since $\setT_2^0$ is nonempty, we set 
    $\phi_{+,2}(t_4)=\phi'_{+,2}(t_4)=0$.
    Finally: observe that the scale in (a) and (b) is different from the scale in (c) and (d); 
    the level $\rho$ is marked for reference in (a) and (b) by a dotted line 
    just above the zero line. The fact that $\rho=1/\SRF^{2r} \ll 1$ is responsible for 
    the noise amplification. The setup is the same as in Figures \ref{fig:toolsneed} and \ref{fig:sumneed}:
    $\flo=12$;
    $\setT=\setT_1\union\setT_2\in\rset{5\llo}{2}$;
    $\abs{t_2-t_1}\sim 2\lhi$, $\abs{t_3-t_4}\sim 2\lhi$; and the 
    target sign patters is $\{s_1, s_2, s_3, s_4\} = \{+1, +1, +1, -1\}$.}
\label{fig:prodd}
\end{figure}

The second term in the product,
\begin{equation}
    \label{eq:phipdef}
    \phi_{+,k}(t)\define \rr^{2\rr}\czdduf^{\rr}q_{\rr\llo,\setT_k,\{f_j\},\{d_j\}}(t)
\end{equation}
is a (rescaled) trigonometric polynomial $q_{\rr\llo,\setT_k,\{f_j\},\{d_j\}}(\cdot)$ constructed 
via \fref{lem:dualcarlos2}  with $\lc=\rr\llo$, and $\setV=\setT_k\in \rset{\kappa\llo\rr}{1}$
and $\czdduf$ is a positive numerical constant defined in~\fref{eq:czdduf} below.
Further, the function-values and derivatives of $q_{\rr\llo,\setT_k,\{f_j\},\{d_j\}}(\cdot)$ are constrained 
on $\setT_k=\setT_k^0\union\setT_k^+$ so that $\phi_{+,k}(\cdot)$ satisfies the following:
\begin{align}
    \label{eq:phi11m}
    \phi_{+,k}(t)&=\begin{cases}
    0, &t\in\setT_k^0,\\
    \rho\frac{1}{\phi_{0,k}(t)}, &t\in \setT_k^+,
    \end{cases}\\
    \label{eq:phi11m1}
    \phi'_{+,k}(t)&=\begin{cases}
    0, &t\in\setT_k^0,\\
    -\rho\frac{\phi'_{0,k}(t)}{\phi^2_{0,k}(t)}, &t\in \setT_k^+.
    \end{cases}
\end{align}
We will prove in \fref{sec:estimates} below, that this specification is valid, in the sense that the 
corresponding function values and derivatives of~$q_{\rr\llo,\setT_k,\{f_j\},\{d_j\}}(\cdot)$ on $\setT_k$ satisfy 
requirements~\fref{eq:condp2} of \fref{lem:dualcarlos2}.

It follows from~\fref{eq:phikdfn}, \fref{eq:phizk}, \fref{eq:phi11m}, \fref{lem:dualcarlos}, 
Properties \ref{prop:qzero}, \ref{prop:qnearbnd}, and \ref{prop:qinfbnd}
that the interpolation constraint~\fref{eq:req1m} is satisfied:%
\begin{alignat}{2}
    \phi_k(t)
    &=\underbrace{\phi_{0,k}(t)}_0 \phi_{+,k}(t)=0 
        &&\text{ for all } t\in\setT_k^c,\\
    \phi_k(t)
    &=\phi_{0,k}(t) \underbrace{\phi_{+,k}(t)}_0=0 
        &&\text{ for all } t\in\setT_k^0,\\
    \phi_k(t)&=\phi_{0,k}(t) \phi_{+,k}(t)=\rho  &&\text{ for all } t\in\setT_k^+.
\end{alignat}
Next, by~\fref{eq:phikdfn},
\begin{equation}
    \phi'_k(t)=\phi'_{0,k}(t)\phi_{+,k}(t)+\phi_{0,k}(t)\phi'_{+,k}(t).
\end{equation}
Therefore, by~\fref{eq:phi11m}, \fref{eq:phi11m1}, \fref{lem:dualcarlos}, 
Properties \ref{prop:qzero}, \ref{prop:qnearbnd}, and \ref{prop:qinfbnd},
the interpolation constraint~\fref{eq:req3m} is satisfied:
\begin{alignat}{2}
    \phi'_k(t)
        &=\underbrace{\phi'_{0,k}(t)}_0 \phi_{+,k}(t)
            +\underbrace{\phi_{0,k}(t)}_0\phi'_{+,k}(t)=0,\  
            &&\text{for all } t\in\setT_k^c,\\
    \phi'_k(t)
        &=\phi'_{0,k}(t)\underbrace{\phi_{+,k}(t)}_0 +\, \phi_{0,k}(t) 
            \underbrace{\phi'_{+,k}(t)}_0=0,\
            &&\text{for all } t\in\setT_k^0,\\
    \phi'_k(t)
        &=\phi'_{0,k}(t)\phi_{+,k}(t)+\phi_{0,k}(t)\phi'_{+,k}(t)\nonumber\\
        &=\phi'_{0,k}(t)\rho\frac{1}{\phi_{0,k}(t)}-\phi_{0,k}(t)\rho\frac{\phi'_{0,k}(t)}{\phi^2_{0,k}(t)}=0,\ 
            &&\text{for all }t\in\setT_k^+.
\end{alignat}

Finally,~\fref{eq:phifreqconstr} follows from~\fref{eq:phikdfn} because $\phi_{0,k}(\cdot)$ in~\fref{eq:phizk} 
is frequency-limited to $(\rr-1)/(\llo\rr)$ [\fref{lem:dualcarlos}, \fref{prop:lf}] and 
$\phi_{+,k}(\cdot)$ in~\fref{eq:phipdef} is 
frequency-limited to $1/(\llo\rr)$ [\fref{lem:dualcarlos2}, \fref{prop:lf2}] so that $\phi_k(\cdot)$ is 
frequency-limited to 
    $(\rr-1)/(\llo\rr)+1/(\llo\rr)=1/\llo=\flo.$
Therefore, by~\fref{eq:q1defm}, $q_1(\cdot)$ is also frequency-limited to $\flo$, 
which proves~\fref{prop:q1prop1}.

For $\rr=2$ the construction is illustrated in \fref{fig:prodd}. In \fref{fig:prodda} trigonometric polynomials
$\phi_{0,1}(\cdot)$ and $\phi_{+,1}(\cdot)$ are displayed; $\phi_{0,1}(t)=0$ for $t\in\setT_2$ as indicated by 
the bold blue points;
$\phi_{+,1}(\cdot)$ satisfies the interpolation constraints~\fref{eq:phi11m} and~\fref{eq:phi11m1} on $\setT_1$ as indicated
by the bold green points. When we compute $\phi_1(\cdot)=(\phi_{0,1}\times \phi_{+,1})(\cdot)$ we obtain the trigonometric
polynomial in \fref{fig:proddc}. The same process is displayed in Figures \ref{fig:proddb} and \ref{fig:proddd} for $\phi_{0,2}(\cdot)$ and $\phi_{+,2}(\cdot)$. The trigonometric polynomials $\phi_1(\cdot)$ and $\phi_2(\cdot)$ in Figures \ref{fig:proddc} and \ref{fig:proddd} are the same ones
as in \fref{fig:sumneed1}.

\subsubsection{Properties of $\phi_{0,k}(\cdot)$}
\label{sec:propphizk}
We will now record useful properties of $\phi_{0,k}(\cdot)$ that are needed in the proof below. For $\rr=1$, according 
to~\fref{eq:phizk}, $\phi_{0,k}(t)=1$ for all $t$. For $\rr>1$, the following properties hold.
\begin{enumerate}
    \item Controlled behavior near $\setT_k^c$: Take $\tau\in\near{\rr\D\llo}{\setT_k^c}$. Let
    \begin{equation}
        \{\vt_1,\ldots,\vt_{\hrr}\}\define \near{\rr\D\llo}{\tau}\intersect\setT_k^c.
    \end{equation}
    Note that since $\setT_k^c\in \rset{\kappa\llo\rr}{\rr-1}$ and $\D<\kappa$, it follows 
    that $1\le\hrr\le\rr-1$. 
    Then, the following estimates hold.
    \begin{enumerate}
        \item Lower bound:
        \begin{align}
            \phi_{0,k}(\tau)
            &\ge \cz^{\rr-1} 
                 \frac{ \prod_{l=1}^{\hrr}(\vt_l-\tau)^2}{(\rr\llo)^{2\hrr}}.\label{eq:qzlppl}
        \end{align}
        \item Upper bound:
        \begin{align}
            \label{eq:phikub}
            \phi_{0,k}(\tau) &\le 
            \cqu^{\hrr} \frac{ \prod_{l=1}^{\hrr}(\vt_l-\tau)^{2}}{(\rr\llo)^{2\hrr}}.
            \label{eq:qzlunivbndpl}
        \end{align}
        \item Upper bound on modulus of the first derivative:
        \begin{equation}
            \abs{\phi_{0,k}'(\tau)}
            \le \sum_{m=1}^{\hrr} \czdu^{\hrr} 
            \frac{\prod_{\substack{1\le l\le \hrr\\l\ne m}} (\vt_{l}-\tau)^{2}\abs{\vt_m-\tau}}
                 {(\rr\llo)^{2\hrr}} 
            +(\rr-1-\hrr) \czdu^{\hrr+1}  \frac{ \prod_{l=1}^{\hrr} (\vt_{l}-\tau)^{2}}{(\rr\llo)^{2\hrr+1}},
        \label{eq:qzdbndpl}
        \end{equation}
        where $\czdu$ is a positive numerical constant defined in the proof below.
        \item Upper bound on modulus of the second derivative:
        \begin{align}
            \abs{\phi_{0,k}''(\tau)}
            &\le  \sum_{1\le m\le\hrr} \sum_{\substack{1\le m'\le\hrr\\m\ne m'}} \I{\hrr\ge 2}
                  \czdu^{\hrr}
                  \frac{ \prod_{\substack{1\le l\le \hrr\\ l\ne m,\ l\ne m'}} (\vt_l-\tau)^{2}}{(\rr\llo)^{2(\hrr-2)}}
                  \frac{\abs{\vt_m-\tau}}{(\rr\llo)^2} \frac{\abs{\vt_{m'}-\tau}}{(\rr\llo)^2}\nonumber\\
            &\dummyrel{\le} +2(\rr-1-\hrr)\sum_{1\le m\le\hrr} \czdu^{\hrr+1} 
                  \frac{\prod_{\substack{1\le l\le \hrr\\ l\ne m}} (\vt_l-\tau)^{2}}{(\rr\llo)^{2(\hrr-1)}}
                  \frac{\abs{\vt_m-\tau}}{(\rr\llo)^2}\frac{1}{\rr\llo}\nonumber\\
            &\dummyrel{\le} +(\rr-1-\hrr)^2  \czdu^{\hrr+2} 
                  \frac{ \prod_{l=1}^{\hrr} (\vt_l-\tau)^{2}}{(\rr\llo)^{2\hrr}}
                  \frac{1}{(\rr\llo)^2}\nonumber\\
            &\dummyrel{\le} + \sum_{1\le m\le\hrr} \czdu^{\hrr}
                  \frac{\prod_{\substack{1\le l\le \hrr\\ l\ne m}} (\vt_l-\tau)^{2}}{(\rr\llo)^{2(\hrr-1)}}
                  \frac{1}{(\rr\llo)^2} 
                  +(\rr-1-\hrr) \czdu^{\hrr+1} 
                  \frac{ \prod_{l=1}^{\hrr} (\vt_l-\tau)^{2}}{(\rr\llo)^{2\hrr}}\frac{1}{(\rr\llo)^2}.
        \label{eq:qzddbndpl}
        \end{align}
    \end{enumerate}
    \item \label{prop:q0lbfarpl} Boundedness away from zero far from $\setT_k^c$: for 
    all $\tau \in\far{\rr\D\llo}{\setT_k^c}$, 
    \begin{equation}
        \label{eq:phi0klb}
        \phi_{0,k}(\tau) \ge \cqlf^{\rr-1}>0.
    \end{equation}
    \item \label{prop:qodbndpl} Uniform confinement of the derivative:
    \begin{equation}
        \label{eq:phi0kdub}
        \infnorm{\phi_{0,k}'(\cdot)}\le 2\pi/\llo.
    \end{equation}
    \item \label{prop:qoddbndpl} Uniform confinement of the second derivative:
    \begin{equation}
        \label{eq:phi0kddub}
        \infnorm{\phi_{0,k}''(\cdot)}\le \czdduu/\llo^2.
    \end{equation}
    \item \label{prop:qzloinfhipl} Fast growth immediately away from $\setT_k^c$: 
        for all $\tau \in\far{\lhi}{\setT_k^c}$,
    \begin{equation}
        \label{eq:phizkbndfar}
        \phi_{0,k}(\tau)\ge \cql^{\rr-1} \frac{\lhi^{2(\rr-1)}}{(\rr\llo)^{2(\rr-1)}}.
    \end{equation}
\end{enumerate}

Next, we give the proofs of the properties.

\paragraph{Proof of properties 1a--1b.}
These properties are derived in the same way as Properties 4a and 4b in \fref{lem:dualprod}.

\paragraph{Proof of property 1c.}
To prove~\fref{eq:qzdbndpl}, observe
\begin{align}
    \abs{\phi_{0,k}'(\tau)}
    = \abs{\left[\prod_{\substack{1\le m\le \rr\\ m\ne k}} q_{\rr\llo,\setT_m}(\tau)\right]'}
    \le \sum_{\substack{1\le m\le \rr\\ m\ne k}} \abs{\prod_{\substack{1\le j\le \rr \\ j\ne k, j\ne m}} q_{\rr\llo,\setT_j}(\tau)} 
        \abs{q_{\rr\llo,\setT_m}'(\tau)}.
    \label{eq:qqpmixed}
\end{align}
Above, we applied the chain rule for derivative to~\fref{eq:phizk} and used the triangle inequality.

To upper-bound the sum in~\fref{eq:qqpmixed}, we upper-bound the quantities 
$\abs{q_{\rr\llo,\setT_j}(\tau)}$ and $\abs{q_{\rr\llo,\setT_{m}}'(\tau)}$ 
separately. To upper-bound $\abs{q_{\rr\llo,\setT_j}(\tau)}$ we use the same bounds 
as in~\fref{eq:qubrwc} and~\fref{eq:qIu1}. To upper-bound $\abs{q_{\rr\llo,\setT_{m}}'(\tau)}$ 
we use a similar strategy as follows. Assume that $m$ is such that 
$\near{\rr\D\llo}{\tau}\intersect\setT_m\ne\varnothing$, i.e., 
there exist $l\in\{1,\ldots,\hrr\}$ that satisfies $\vt_l\in\setT_m$.
In this case, according to \fref{lem:dualcarlos}, \fref{prop:qzero}, 
$q'_{\rr\llo,\setT_m}(\vt_l)=0$ and according to \fref{lem:dualcarlos}, 
\fref{prop:qppbnd}, $\abs{q''_{\rr\llo,\setT_m}(t)}\le 4\pi^2/(\rr\llo)^2$ 
for all $t$. This, by~\fref{eq:mvt1} [Mean Value theorem], gives the following bound:
\begin{equation}
    \label{eq:qubrwcd}
    \abs{q'_{\rr\llo,\setT_m}(\tau)}\le \cqdvdd\frac{\abs{\vt_l-\tau}}{(\rr\llo)^2}.
\end{equation}
Assume that $m$ is such that $\near{\rr\D\llo}{\tau}\intersect\setT_m=\varnothing$.
In this case, we use \fref{lem:dualcarlos}, \fref{prop:qpbnd}, to write
\begin{equation}
    \label{eq:qubrwcd2}
    \abs{q'_{\rr\llo,\setT_m}(\tau)}\le 2\pi\frac{1}{\rr\llo}.
\end{equation}
Plugging the estimates for $\abs{q_{\rr\llo,\setT_j}(\tau)}$ [\fref{eq:qubrwc} 
and~\fref{eq:qIu1}], \fref{eq:qubrwcd}, and~\fref{eq:qubrwcd2} into~\fref{eq:qqpmixed}, setting 
$\czdu \define \max(2\pi, 4\pi^2, \cqu)$ we obtain~\fref{eq:qzdbndpl}.

\paragraph{Proof of property 1d.}
To prove~\fref{eq:qzddbndpl}, observe
\begin{align}
    \abs{\phi_{0,k}''(\tau)}
    &= \abs{\left[\prod_{\substack{1\le m\le \rr\\ m\ne k}} q_{\rr\llo,\setT_m}(\tau)\right]''}\\
    &\le \!\!\sum_{\substack{1\le m\le \rr\\ m\ne k}}\, \sum_{\substack{1\le m'\le \rr\\m'\ne k, m'\ne m}} \!\abs{\prod_{\substack{1\le j\le \rr\\j\ne k, j\ne m, j\ne m'}} 
        q_{\rr\llo,\setT_j}(\tau)}
        \abs{q_{\rr\llo,\setT_{m}}'(\tau)}
        \abs{q_{\rr\llo,\setT_{m'}}'(\tau)}\nonumber\\
    &\dummyrel{\le}
        +\sum_{\substack{1\le m\le \rr\\m\ne k}} \abs{\prod_{\substack{1\le j\le \rr\\j\ne k, j\ne m}} q_{\rr\llo,\setT_j}(\tau)} 
                       \abs{q_{\rr\llo,\setT_{m}}''(\tau)}.
    \label{eq:qqpmixedd}
\end{align}
To upper-bound the sum in~\fref{eq:qqpmixedd}, we upper-bound the 
quantities $\abs{q_{\rr\llo,\setT_j}(\tau)}$, $\abs{q_{\rr\llo,\setT_{m}}'(\tau)}$, 
and $\abs{q_{\rr\llo,\setT_{m}}''(\tau)}$ separately.
To upper-bound $\abs{q_{\rr\llo,\setT_j}(\tau)}$ and $\abs{q'_{\rr\llo,\setT_m}(\tau)}$ 
we use estimates~\fref{eq:qubrwc}, \fref{eq:qIu1} and~\fref{eq:qubrwcd}, \fref{eq:qubrwcd2}, 
respectively. To upper-bound $\abs{q''_{\rr\llo,\setT_m}(\tau)}$ we use 
\fref{lem:dualcarlos}, \fref{prop:qppbnd}, to write
\begin{equation}
    \abs{q''_{\rr\llo,\setT_m}(\tau)}\le 4\pi^2\frac{1}{(\rr\llo)^2}.
    \label{eq:qdd}
\end{equation}
Plugging these estimates into~\fref{eq:qqpmixedd},
we obtain~\fref{eq:qzddbndpl}.

\paragraph{Proof of properties 2--5.}
\label{sec:prooflast}
\fref{prop:q0lbfarpl} follows by~\fref{eq:phizk} and \fref{lem:dualcarlos}, \fref{prop:qinfbnd}.
\fref{prop:qodbndpl} follow from~\fref{eq:qqpmixed} and from \fref{lem:dualcarlos}, \fref{prop:qpbnd}:
\begin{equation}
    \abs{\phi_{0,k}'(\tau)}\le (\rr-1)\frac{2\pi}{\rr\llo}<\frac{2\pi}{\llo}.
\end{equation}
\fref{prop:qoddbndpl} follow from~\fref{eq:qqpmixedd} and from \fref{lem:dualcarlos}, 
Properties \ref{prop:qpbnd} and \ref{prop:qppbnd}:
\begin{equation}
    \abs{\phi_{0,k}''(\tau)}
    \le (\rr-1)(\rr-2)\frac{4\pi^2}{(\rr\llo)^2}
        +(\rr-1)\frac{4\pi^2}{(\rr\llo)^2}<\frac{8\pi^2}{\llo^2},
\end{equation}
where we defined $\czdduu\define 8\pi^2$.
Finally, \fref{prop:qzloinfhipl} follows from~\fref{eq:phizk}, \fref{eq:qnearbndeq}, 
\fref{lem:dualcarlos}, \fref{prop:qinfbnd}, and~\fref{eq:cqlfcnear}.

\subsubsection{Existence of $\phi_{+,k}(\cdot)$}
\label{sec:estimates}
In this subsection, we check that trigonometric polynomial $\phi_{+,k}(\cdot)$ that satisfies~\fref{eq:phi11m} 
and~\fref{eq:phi11m1} can indeed be defined according to~\fref{eq:phipdef} 
with $q_{\rr\llo,\setT_k,\{f_j\},\{d_j\}}(\cdot)$ constructed via \fref{lem:dualcarlos2}
with $\lc=\rr\llo$ and $\setV=\setT_k\in \rset{\kappa\llo\rr}{1}$.
To this end, we need to show 
that the constraints on the function values $\{f_j\}$ and on the derivatives $\{d_j\}$
that are implied by the constraints~\fref{eq:phi11m} and~\fref{eq:phi11m1} satisfy 
requirements~\fref{eq:condp2} of \fref{lem:dualcarlos2}.

First consider the case $r=1$. As already discussed, in this case $\phi_{0,k}(t)=1$ for all
$t$, and, therefore, $\phi'_{0,k}(t)=0$ for all $t$. Plugging these values into~\fref{eq:phi11m} 
and~\fref{eq:phi11m1} we see from~\fref{eq:phipdef} that the requirements~\fref{eq:condp2} of \fref{lem:dualcarlos2}
are satisfied.

Next, consider the case $r>1$.

For $t\in\setT_k^0$, by~\fref{eq:phipdef}, \fref{eq:phi11m}, \fref{eq:phi11m1}, 
$q_{\rr\llo,\setT_k}(t)=q'_{\rr\llo,\setT_k}(t)=0$ so that 
requirements~\fref{eq:condp2} of \fref{lem:dualcarlos2}, are satisfied. 

To check that requirements~\fref{eq:condp2}  are also satisfied for $t\in\setT_k^+$, we need 
to find upper bounds on $\abs{\phi_{+,k}(\cdot)}$ and $\abs{\phi_{+,k}'(\cdot)}$.

Take $t\in\setT_k^+$ and observe:
\begin{align}
    \abs{\phi_{+,k}(t)}
    \stackrel{(a)}{=}\abs{\rho\frac{1}{\phi_{0,k}(t)}}
    \stackrel{(b)}{\le} \frac{\lhi^{2\rr}}{\llo^{2\rr}} 
        \frac{1}{\frac{\lhi^{2(\rr-1)}}{(\rr\llo)^{2(\rr-1)}}} 
        \frac{1}{\cql^{\rr-1}}
    &\stackrel{(c)}{\le} \rr^{2(\rr-1)}
         \frac{1}{\cql^{\rr}} \frac{\lhi^{2}}{\llo^{2}} \label{eq:ph11bdnm0}\\
    &\le \rr^{2\rr}\frac{1}{\cql^{\rr}}.\label{eq:ph11bdnm}
\end{align}
Above, (a) follows by~\fref{eq:phi11m}; (b) follows by~\fref{eq:phizkbndfar} which is valid 
because $t\in\setT_k^+$ implies $t\in\far{\lhi}{\setT_k^c}$; (c) follows because $\cql<1$.

Next, take $t\in\setT_k^+$  and observe, according to~\fref{eq:phi11m1},
\begin{equation}
    \label{eq:phipderiv}
    \abs{\phi_{+,k}'(t)}=\abs{\rho\frac{\phi'_{0,k}(t)}{\phi^2_{0,k}(t)}}.
\end{equation}
Consider two cases.

Case 1: $t\in\far{\rr\D\llo}{\setT_k^c}$. Then, by~\fref{eq:phi0klb}, 
$\phi_{0,k}(t) \ge \cqlf^{\rr-1}$, and, by~\fref{eq:phi0kdub},
$\abs{\phi_{0,k}'(t)}\le 2\pi/\llo$. Plugging these estimates 
into~\fref{eq:phipderiv} we obtain
\begin{equation}
    \abs{\phi_{+,k}'(t)}
    \le \abs{\rho\frac{2\pi}{\cqlf^{\rr-1}}\frac{1}{\llo}}
    \stackrel{(a)}{\le} \frac{2\pi}{\cqlf^{2\rr-2}}\frac{\lhi}{\llo^2}
    \stackrel{(b)}{<} \rr^{2\rr-1} \left(\frac{2\pi}{\cqlf^{2}}\right)^\rr  \frac{\lhi}{\llo^2}.
    \label{eq:phipderiv1}
\end{equation}
Above, in (a) we used $\rho=(\lhi/\llo)^{2\rr}\le \lhi/\llo$; (b) is a crude inequality
where we used $\cqlf<1$.

Case 2: $t\in\near{\rr\D\llo}{\setT_k^c}$. In this case set 
$\{v_1,\ldots,v_{\hrr}\}\define\setT_k^c\intersect\near{\rr\D\llo}{t}$ 
and note $1\le \hrr\le \rr-1$.
Hence,  by~\fref{eq:qzlppl}:
\begin{equation}
    \label{eq:qqlb}
    \phi_{0,k}(t)\ge \cz^{\rr-1}\frac{ \prod_{j=1}^{\hrr}(v_j-t)^2}{(\rr\llo)^{2\hrr}}.
\end{equation}
By~\fref{eq:qzdbndpl}:
\begin{equation}
    \label{eq:qqdub}
    \abs{\phi_{0,k}'(t)}\le \sum_{m=1}^{\hrr} \czdu^{\hrr} 
    \frac{\prod_{\substack{1\le l\le \hrr\\l\ne m}} (v_{l}-t)^{2}\abs{v_m-t}}{(\rr\llo)^{2\hrr}} 
    +(\rr-1-\hrr) \czdu^{\hrr+1} 
    \frac{ \prod_{l=1}^{\hrr} (v_{l}-t)^{2}}{(\rr\llo)^{2\hrr+1}}.
\end{equation}
Plugging~\fref{eq:qqlb} and~\fref{eq:qqdub} into~\fref{eq:phipderiv}:
\begin{align}
    \frac{\abs{\phi'_{0,k}(t)}}{\phi_{0,k}^2(t)}
    &\le 
    \sum_{m=1}^{\hrr} 
        \frac{\czdu^{\hrr}}{\cz^{2(\rr-1)}}
        \frac{(\rr\llo)^{2\hrr}}{\prod_{\substack{1\le j\le\hrr\\j\ne r}} (v_{j}-t)^{2}\abs{v_m-t}^3}
    +(\rr-1-\hrr) \frac{\czdu^{\hrr+1}}{\cz^{2(\rr-1)}}
                  \frac{(\rr\llo)^{2\hrr-1}}{\prod_{j=1}^{\hrr} (v_{j}-t)^{2}}\\
    &\stackrel{(a)}{\le} 
    \rr^{2\hrr+1}\frac{\czdu^{\hrr+1}}{\cz^{2(\rr-1)}} 
        \left(\frac{\llo^{2\hrr}}{\lhi^{2\hrr+1}}+\frac{\llo^{2\hrr-1}}
                   {\lhi^{2\hrr}}\right)
    \stackrel{(b)}{\le} 
        \rr^{2\rr-1}\frac{\czdu^{\rr}}{\cz^{2(\rr-1)}}  
        \left(\frac{\llo^{2\rr-2}}{\lhi^{2\rr-1}}
             +\frac{\llo^{2\rr-3}}{\lhi^{2\rr-2}}\right)\nonumber\\
    &\stackrel{(c)}{\le} 
        2\rr^{2\rr-1}\left(\frac{\czdu}{\cz^2}\right)^\rr  
        \frac{\llo^{2\rr-2}}{\lhi^{2\rr-1}}
    \stackrel{(d)}{\le} \rr^{2\rr-1} \czdduuu^\rr   \frac{\llo^{2\rr-2}}{\lhi^{2\rr-1}}.
    \label{eq:fracderivbnd}
\end{align}
Above, in (a) we used that $\abs{v_{j}-t}\ge 2\lhi$ for all $j=1,\ldots,\hrr$, $\rr-1-\hrr\le\rr$, and $\czdu>1$; 
in (b) we used that $\hrr\le \rr-1$, $\llo/\lhi>1$, and $\czdu>1$; 
in (c) we used $\llo/\lhi>1$ and $\cz<1$; in (d) we defined $\czdduuu \define 2 \czdu/\cz^2$.
Plugging the estimate~\fref{eq:fracderivbnd} into~\fref{eq:phipderiv},
\begin{align}
    \abs{\phi_{+,k}'(t)}
    &\le \rr^{2\rr-1} \czdduuu^\rr 
         \frac{\lhi^{2\rr}}{\llo^{2\rr}} 
         \frac{\llo^{2\rr-2}}{\lhi^{2\rr-1}}
    =\rr^{2\rr-1} \czdduuu^\rr\frac{\lhi}{\llo^2}.\label{eq:bndphikdnear11}
\end{align}
Combining~\fref{eq:phipderiv1} and~\fref{eq:bndphikdnear11} we find that for all $t\in\setT_k^+$, 
\begin{align}
    \abs{\phi_{+,k}'(t)} 
    &\le \rr^{2\rr-1}\czdduuuu^\rr\frac{\lhi}{\llo^2} \label{eq:bndphikdnear1f}\\
    &\le \rr^{2\rr-1}\czdduuuu^\rr\frac{1}{\llo}, \label{eq:bndphikdnear1}
\end{align}
where we defined $\czdduuuu\define\max(\czdduuu, \cqdvd/\cqlf^{2})$.

It follows from~\fref{eq:ph11bdnm} and~\fref{eq:bndphikdnear1} that the function 
values and derivatives of $q_{\rr\llo,\setT_k}(t)=\phi_{+,k}(t)/(\rr^{2\rr}\czdduf^{\rr})$ with 
\begin{equation}
    \czdduf\define\max(\czdduuuu, 1/\cql)
    \label{eq:czdduf}
\end{equation} 
satisfy requirements~\fref{eq:condp2} of \fref{lem:dualcarlos2} on $\setT_k^+$. We 
conclude that $\phi_{+,k}(\cdot)$ can indeed be defined according to~\fref{eq:phipdef}. 
According to Properties~\ref{prop:dqinfbnd}, \ref{prop:qpbnd1}, and \ref{prop:qppbnd1}
of \fref{lem:dualcarlos2}, and~\fref{eq:phipdef}, $\phi_{+,k}(\cdot)$ satisfies the following properties:
\begin{align}
    &\infnorm{\phi_{+,k}(\cdot)}\le \rr^{2\rr}\czdduf^{\rr}\cqdv,\label{eq:phi11ddm00}\\
    &\infnorm{\phi_{+,k}'(\cdot)}\le \rr^{2\rr-1}\czdduf^{\rr} \cqdvd\frac{1}{\llo},\label{eq:phi11ddm0}\\
    &\infnorm{\phi_{+,k}''(\cdot)}\le \rr^{2\rr-2}\czdduf^{\rr} \cqdvdd\frac{1}{\llo^2}.\label{eq:phi11ddm}
\end{align}

\subsubsection{Proof of property 2}
Take $j\in\{1,\ldots, \sparsity\}$ and consider $t_j\in\setT$. 
There exists a unique $l\in\{1,\ldots,\rr\}$ such that $t_j\in\setT_l$. We will show that 
for all $\tau\in\near{\lhi}{t_j}$
\begin{equation}
    \abs{\phi_l(\tau)-\eta_j}\le \rr^{2\rr+3} \cuuuuu^{\rr+1}  q_0(\tau) \label{eq:toprove2}
\end{equation}
and
\begin{equation}
    \abs{\phi_k(\tau)}\le \rr^{2\rr+3} \cuf^{\rr+1}  q_0(\tau),\ 
        \text{for}\ k\in\{1,\ldots, \rr\},\ k\ne l, \label{eq:toprove1}
\end{equation}
where the positive numerical constants $\cuuuuu$ and $\cuf$ are defined below.

From this we will conclude that
\begin{align}
    \abs{q_1(\tau)-\frac{\rho s_j}{2}}
    =\abs{\sum_{k=1}^\rr \phi_k(\tau) -\frac{\rho}{2}-\frac{\rho s_j}{2}}
    &=\abs{\sum_{k=1}^\rr \phi_k(\tau) -\eta_j}\\
    &\le\sum_{k\ne l}^\rr \abs{\phi_k(\tau)}+ \abs{\phi_l(\tau) -\eta_j}
    \le\rr^{2\rr+4} \cuff^{\rr+1}  q_0(\tau)
\end{align}
with $\cuff \define 2\max(\cuuuuu, \cuf)$, as desired.

To prove~\fref{eq:toprove2} and~\fref{eq:toprove1},  recall, by~\fref{eq:req1m} and~\fref{eq:req3m}:
\begin{alignat}{2}
    &\abs{\phi_l(t_j)-\eta_j} =  0 = q_0(t_j),\label{eq:init1}\\
    &\phi_k(t_j) =  0 = q_0(t_j), &&\text{for}\ k\in\{1,\ldots, \rr\},\ k\ne l,\label{eq:init2}\\
    &\phi'_k(t_j)=0=q'_0(t_j), &&\text{for}\ k\in\{1,\ldots, \rr\}.\label{eq:init3}
\end{alignat}
Hence, in order to prove the bounds in~\fref{eq:toprove1} and~\fref{eq:toprove2}, we will derive upper bounds on 
the second derivatives $\abs{\phi''_k(\tau)}$, $k\in\{1,\ldots,\rr\}$, valid for 
all $\tau\in\near{\lhi}{t_j}$
and use the Mean Value theorem (see \fref{thm:meanvalue}).

Taking the second derivative of~\fref{eq:phikdfn} and applying the triangle inequality we find:
\begin{equation}
    \abs{\phi''_{k}(\tau)}
    \le \underbrace{\abs{\phi''_{0,k}(\tau)}\abs{\phi_{+,k}(\tau)}}_{E_{1}(\tau)}
        +2\underbrace{\abs{\phi'_{0,k}(\tau)}\abs{\phi'_{+,k}(\tau)}}_{E_{2}(\tau)}
        +\underbrace{\abs{\phi_{0,k}(\tau)}\abs{\phi''_{+,k}(\tau)}}_{E_3(\tau)}.
    \label{eq:phiddsumtermsl1}
\end{equation}
In the derivation below we upper-bound the terms separately.

We will need the following notations.
Set $\{\vt_1,\ldots,\vt_{\hrr}\}\define \near{\rr\D\llo}{\tau}\intersect\setT_k^c$ and
set $\{v_1,\ldots,v_{\tilde\rr}\}\define \near{\rr\D\llo-\lhi}{t_j}\intersect\setT_k^c$. 
Note that the set $\{v_1,\ldots,v_{\tilde\rr}\}$ does not depend on $\tau$ and 
also $\{v_1,\ldots,v_{\tilde\rr}\}\subset \{\vt_1,\ldots,\vt_{\hrr}\}$ so 
that $\tilde\rr\le \hrr$. 

The remainder of the proof of Property 2 is organized as follows. First, consider the 
case $t_j\in \setT_k$ and prove~\fref{eq:toprove2}, next consider the 
case $t_j\in\setT_k^c$ and prove~\fref{eq:toprove1}.

\paragraph{Proof of~\fref{eq:toprove2}: case~$t_j\in\setT_k$.}
\subparagraph{Bounding $E_1(\tau)$.}
By \fref{eq:mvt2} [Mean Value theorem] and the triangle inequality we can write
\begin{equation}
    \abs{\phi_{+,k}(\tau)}
    \le \abs{\phi_{+,k}(t_j)}+\abs{\phi_{+,k}'(t_j)}\abs{\tau- t_j}
        +\frac{1}{2}\abs{\phi_{+,k}''(\tau_m)}(\tau- t_j)^2
\end{equation}
with $\tau_m\in (t_j, \tau)$. Next, we use~\fref{eq:phi11m} and~\fref{eq:ph11bdnm0} to 
upper-bound $\abs{\phi_{+,k}(t_j)}$ by the right-hand side of~\fref{eq:ph11bdnm0}; 
use~\fref{eq:phi11m1} and~\fref{eq:bndphikdnear1f} to upper 
bound $\abs{\phi_{+,k}'(t_j)}$ by the right-hand side of~\fref{eq:bndphikdnear1f}; 
use~\fref{eq:phi11ddm} to upper-bound $\abs{\phi_{+,k}''(\tau_m)}$. With 
these estimates we can further upper-bound $\abs{\phi_{+,k}(\tau)}$ as follows:
\begin{align}
    \abs{\phi_{+,k}(\tau)}
    &\le\rr^{2\rr-1}\czdduff^\rr \left(\frac{\lhi^2}{\llo^2}
        +\frac{\lhi}{\llo^2}\abs{\tau-t_j}
        +\frac{1}{\llo^2}(\tau-t_j)^2\right)
    \le \rr^{2\rr-1}\czddufff^\rr \frac{\lhi^2}{\llo^2}.
    \label{eq:phipbnd}
\end{align}
Above, we defined  $\czdduff \define \max(1/\cql,\czdduuuu, \czdduf\cqdvdd)$, $\czddufff\define 3 \czdduff$, 
and used $\abs{\tau-t_j}\le\lhi$.

Assume $\tilde \rr\ge 1$ (the case  $\tilde \rr=0$ will be treated separately below) 
so that $\hrr\ge 1$ and $\tau\in\near{\rr\D\llo}{\setT_k^c}$, which 
implies that we can use~\fref{eq:qzddbndpl} to upper-bound $\abs{\phi''_{0,k}(t)}$:
\begin{align}
    \abs{\phi_{0,k}''(\tau)}
    &\le \sum_{1\le m\le\hrr} 
        \sum_{\substack{1\le m'\le\hrr\\m\ne m'}} 
        \I{\hrr\ge 2}\czdu^{\hrr}
        \frac{ \prod_{\substack{1\le l\le \hrr\\ l\ne m,\ l\ne m'}} (\vt_l-\tau)^{2}}{(\rr\llo)^{2(\hrr-2)}}
        \frac{\abs{\vt_m-\tau}}{(\rr\llo)^2} \frac{\abs{\vt_{m'}-\tau}}{(\rr\llo)^2}\nonumber\\
    &\dummyrel{\le} +2(\rr-1-\hrr)\sum_{1\le m\le\hrr} 
        \czdu^{\hrr+1} 
        \frac{ \prod_{\substack{1\le l\le \hrr\\ l\ne m}} (\vt_l-\tau)^{2}}{(\rr\llo)^{2(\hrr-1)}}
        \frac{\abs{\vt_m-\tau}}{(\rr\llo)^2}\frac{1}{\rr\llo}\nonumber\\
    &\dummyrel{\le} +(\rr-1-\hrr)^2 
        \czdu^{\hrr+2}
        \frac{ \prod_{l=1}^{\hrr} (\vt_l-\tau)^{2}}{(\rr\llo)^{2\hrr}}
        \frac{1}{(\rr\llo)^2}\nonumber\\
    &\dummyrel{\le} 
        + \sum_{1\le m\le\hrr} 
            \czdu^{\hrr}
            \frac{ \prod_{\substack{1\le l\le \hrr\\ l\ne m}} (\vt_l-\tau)^{2}}{(\rr\llo)^{2(\hrr-1)}}
            \frac{1}{(\rr\llo)^2} 
        +(\rr-1-\hrr) 
            \czdu^{\hrr+1}
            \frac{ \prod_{l=1}^{\hrr} (\vt_l-\tau)^{2}}{(\rr\llo)^{2\hrr}}\frac{1}{(\rr\llo)^2}.
\label{eq:qddub}
\end{align}

Multiplying~\fref{eq:phipbnd} and~\fref{eq:qddub} and simplifying we obtain the 
following upper bound on $E_1$:
\begin{align}
    E_1(\tau)=\abs{\phi''_{0,k}(\tau)}\abs{\phi_{+,k}(\tau)}
    &\stackrel{(a)}{\le} \rr^{2\rr+1} \czdduffff^{\rr+1}  \frac{\prod_{1\le l\le \hrr} (\vt_l-\tau)^2}{(\rr\llo)^{2\hrr}}\frac{1}{\llo^2}\nonumber\\
    &\stackrel{(b)}{\le} \rr^{2\rr+1} \czdduffff^{\rr+1}  \frac{\prod_{1\le l\le \tilde\rr} (v_l-\tau)^2}{(\rr\llo)^{2\tilde\rr}}\frac{1}{\llo^2}.
    \label{eq:E1}
\end{align}
Above, in (a) we  used (multiple times) the bound $\lhi\le \abs{\vt_l-\tau}$,
which is true for all $l\in\{1,\ldots,\hrr\}$ 
(follows because the elements of $\setT$ are separated 
by at least $2\lhi$), used $\lhi/\llo<1$, and defined $\czdduffff \define \max(6 \czdu \czddufff, \czdduf\cqdv\czdduu)$;
in (b) we used the fact that $\abs{\vt_l-\tau}/(\llo \rr)\le \D<1$ for 
all $l\in\{1,\ldots,\hrr\}$. 

For the case $\tilde r=0$, the 
upper-bound~\fref{eq:E1} also holds by~\fref{eq:phi0kddub}
and~\fref{eq:phi11ddm00}. 

\subparagraph{Bounding $E_2(\tau)$.} 
By~\fref{eq:mvt1} [Mean Value theorem] we can write
\begin{equation}
    \abs{\phi'_{+,k}(\tau)}
    \le \abs{\phi'_{+,k}(t_j)}+\abs{\phi_{+,k}''(\tau_m)}\abs{\tau- t_j}
\end{equation}
with $\tau_m\in (t_j, \tau)$. Next, we use~\fref{eq:phi11m1} and~\fref{eq:bndphikdnear1f} 
to upper-bound $\abs{\phi_{+,k}'(t_j)}$ by the right-hand side of~\fref{eq:bndphikdnear1f}; 
use~\fref{eq:phi11ddm} to upper-bound $\abs{\phi_{+,k}''(\tau_m)}$. With these estimates we can further upper-bound 
$\abs{\phi'_{+,k}(\tau)}$ as follows:
\begin{align}
    \abs{\phi'_{+,k}(\tau)}
    &\le\rr^{2\rr-1} \czdduffd^\rr 
        \left(\frac{\lhi}{\llo^2}+\frac{1}{\llo^2}\abs{\tau-t_j}\right)
    \le  \rr^{2\rr-1} \czddufffdd^\rr\frac{\lhi}{\llo^2}.
    \label{eq:phipkdub}
\end{align}
Above, we defined $\czdduffd \define \max(\czdduuuu, \cqdvdd\czdduf)$, $\czddufffdd\define 2 \czdduffd$, 
and used $\abs{\tau-t_j}\le\lhi$.

Assume $\tilde \rr\ge 1$ (the case  $\tilde \rr=0$ will be treated separately 
below) so that $\hrr\ge 1$ and $\tau\in\near{\rr\D\llo}{\setT_k^c}$, 
which implies that we can use~\fref{eq:qzdbndpl} to upper-bound $\abs{\phi'_{0,k}(t)}$:
\begin{equation}
    \abs{\phi_{0,k}'(\tau)}
    \le \sum_{m=1}^{\hrr} \czdu^{\hrr}  
        \frac{\prod_{\substack{1\le l\le \hrr\\l\ne m}} (\vt_{l}-\tau)^{2}\abs{\vt_m-\tau}}{(\rr\llo)^{2\hrr}} 
        +(\rr-1-\hrr) \czdu^{\hrr+1} 
        \frac{ \prod_{l=1}^{\hrr} (\vt_{l}-\tau)^{2}}{(\rr\llo)^{2\hrr+1}}.
    \label{eq:qddub2}
\end{equation}

Multiplying~\fref{eq:phipkdub} and~\fref{eq:qddub2} and simplifying we obtain the 
following upper bound on $E_2$:
\begin{align}
    E_2(\tau)=\abs{\phi'_{0,k}(\tau)}\abs{\phi'_{+,k}(\tau)}
    &\stackrel{(a)}{\le} 
        \rr^{2\rr} \czddufffddd^{\rr}  
        \frac{\prod_{1\le l\le \hrr} (\vt_l-\tau)^2}{(\rr\llo)^{2\hrr}}
        \frac{1}{\llo^2}\nonumber\\
    &\stackrel{(b)}{\le} 
        \rr^{2\rr} \czddufffddd^{\rr}  
        \frac{\prod_{1\le l\le \tilde\rr} (v_l-\tau)^2}{(\rr\llo)^{2\tilde\rr}}
        \frac{1}{\llo^2}.
    \label{eq:E2}
\end{align}
Above, in (a) we  used the bound $\lhi\le \abs{\vt_l-\tau}$, which is true 
for all $l\in\{1,\ldots,\hrr\}$ (follows because the elements of $\setT$ are separated 
by at least $2\lhi$), used $\lhi/\llo<1$, 
and defined $\czddufffddd \define \max(2  \czddufffdd \czdu, 2\pi\czdduf \cqdvd)$; in (b) 
we used the fact that $\abs{\vt_l-\tau}/(\llo \rr)\le \D<1$ for 
all $l\in\{1,\ldots,\hrr\}$. 

For the case $\tilde r=0$, the upper bound~\fref{eq:E2} 
also holds by~\fref{eq:phi0kdub} and~\fref{eq:phi11ddm0}.

\subparagraph{Bounding $E_3(\tau)$.}
By~\fref{eq:phi11ddm},
\begin{equation}
    \label{eq:phipe3}
    \abs{\phi_{+,k}''(\tau)}\le \rr^{2\rr-2}\czdduf^{\rr} \cqdvdd\frac{1}{\llo^2}.
\end{equation}
Assume $\tilde \rr\ge 1$ (the case  $\tilde \rr=0$ will be treated separately below) 
so that $\hrr\ge 1$ and $\tau\in\near{\rr\D\llo}{\setT_k^c}$, 
which implies that we can use~\fref{eq:phikub} to upper-bound $\abs{\phi_{0,k}(\tau)}$:
\begin{align}
    \label{eq:phipke3}
    \abs{\phi_{0,k}(\tau)} 
    &\le \cqu^{\hrr}\frac{ \prod_{l=1}^{\hrr} (\vt_l-\tau)^{2}}{(\rr\llo)^{2\hrr}}.
    \label{eq:qzlunivbndpl}
\end{align}

Multiplying~\fref{eq:phipe3} and~\fref{eq:phipke3} and simplifying we obtain the 
following upper bound on $E_3$:
\begin{align}
    E_3(\tau)
    =\abs{\phi_{0,k}(\tau)}\abs{\phi''_{+,k}(\tau)}
    &\stackrel{(a)}{\le} \rr^{2\rr-2}\czddufffdddd^{\rr} 
        \frac{ \prod_{l=1}^{\hrr} (\vt_l-\tau)^{2}}{(\rr\llo)^{2\hrr}} \frac{1}{\llo^2} \\
    &\stackrel{(b)}{\le} \rr^{2\rr-2}\czddufffdddd^{\rr} 
        \frac{\prod_{1\le l\le \tilde\rr} (v_l-\tau)^2}{(\rr\llo)^{2\tilde\rr}} \frac{1}{\llo^2}.
    \label{eq:E3}
\end{align}
Above, (a) we defined $\czddufffdddd \define \czdduf\cqdvdd\cqu$; in (b) we use the fact 
that $\abs{\vt_l-\tau}/(\llo \rr)\le \D<1$ for 
all $l\in\{1,\ldots,\hrr\}$. 

For the case $\tilde r=0$, the 
upper-bound~\fref{eq:E3} also holds  by~\fref{eq:phipe3} because 
by~\fref{eq:phizk} and \fref{lem:dualcarlos}, \fref{prop:qzeroone}, $\abs{\phi_{0,k}(\tau)}<1$ 
and because $\cqu>1$ and $\cqdvdd>1$.

From~\fref{eq:phiddsumtermsl1}, \fref{eq:E1}, \fref{eq:E2}, 
and~\fref{eq:E3} we conclude that
\begin{equation}
     \abs{\phi''_{k}(\tau)}
     \le \rr^{2\rr+1} \cuos^{\rr+1}  
         \frac{\prod_{1\le l\le \tilde\rr} (v_l-\tau)^2}{(\rr\llo)^{2\tilde\rr}} 
         \frac{1}{\llo^2},
    \label{eq:phiddfinalbnd1}
\end{equation}
where we defined $\cuos \define 4 \max(\czdduffff, \czddufffddd, \czddufffdddd)$.

\subparagraph{Putting pieces together.}
On the one hand, by~\fref{eq:mvt2} [Mean Value theorem],
using~\fref{eq:init1}, \fref{eq:init3}, \fref{eq:phiddfinalbnd1} and 
we can write for all $\tau\in\near{\lhi}{t_j}$:
\begin{align}
    \abs{\phi_l(\tau)-\eta_j}
    &\stackrel{(a)}{\le} 
        \frac{1}{2} \rr^{2\rr+1} \cuos^{\rr+1}  
        \frac{\prod_{1\le l\le \tilde\rr} (v_l-\tau_m)^2}{(\rr\llo)^{2\tilde\rr}}  
        \frac{(\tau-t_j)^2}{\llo^2}\\
    &\stackrel{(b)}{\le} 
        \frac{1}{2} \rr^{2\rr+1} \cuos^{\rr+1}  2^{\tilde\rr}
        \frac{\prod_{1\le l\le \tilde\rr} (v_l-\tau)^2}{(\rr\llo)^{2\tilde\rr}} 
        \frac{(\tau-t_j)^2}{\llo^2}\\
    &\stackrel{(c)}{\le} 
        \rr^{2\rr+3} \cuuuu^{\rr+1} 
        \frac{\prod_{1\le l\le \tilde\rr} (v_l-\tau)^2}{(\rr\llo)^{2\tilde\rr}} 
        \frac{(\tau-t_j)^2}{(\rr\llo)^2}.
\label{eq:philtau1}
\end{align}
Above, in (a) $\tau_m\in (t_j, \tau)$; in (b) we used 
that $\abs{v_l-\tau_m}<\abs{v_l-\tau}+\lhi<2\abs{v_l-\tau}$, which is 
true because $\tau\in\near{\lhi}{t_j}$ and because the elements of $\setT$ are separated 
by at least $2\lhi$; 
in (c) we defined $\cuuuu \define 2\cuos$.

On the other hand, let 
$\{u^\tau_1,\ldots,u^\tau_{\breve \rr}\}\define \near{\rr\D\llo}{\tau}\intersect\setT$. 
Then, by~\fref{eq:qzlp},
\begin{align}
    q_0(\tau)
    \ge \cz^\rr  \frac{ \prod_{l=1}^{\breve  \rr}(u^\tau_l-\tau)^2}{(\rr\llo)^{2\breve \rr}}
    &\stackrel{(a)}{\ge} 
        \cz^\rr  \frac{\prod_{1\le l\le \tilde\rr} (v_l-\tau)^2}{(\rr\llo)^{2\tilde\rr}}  
        \frac{(\tau-t_j)^2}{(\rr\llo)^2} 
        \underbrace{
            \left(\frac{(\rr\D\llo-2\lhi)^2}{(\rr\llo)^2}\right)^{\breve\rr-\tilde\rr-1}
        }_{P_1}\\
    &\stackrel{(b)}{\ge} 
        \czz^\rr \frac{\prod_{1\le l\le \tilde\rr} (v_l-\tau)^2}{(\rr\llo)^{2\tilde\rr}}  
        \frac{(\tau-t_j)^2}{(\rr\llo)^2}.
\label{eq:philtau2}
\end{align}
Above, in (a) we use the fact that 
$\{v_1,\ldots,v_{\tilde r}\}\union \{t_j\}\subset \{u^\tau_1,\ldots,u^\tau_{\breve \rr}\}$ 
and the fact that by construction of the set $\{v_1,\ldots,v_{\tilde r}\}$ 
it follows that if, for some $k$, 
$u^\tau_k\notin \{v_1,\ldots,v_{\tilde r}\}\union \{t_j\}$, 
then\linebreak $\abs{u^\tau_k-\tau}\ge \rr\D\llo-2\lhi$; in (b) 
we used the assumption $\SRF\ge 12$ so that $\lhi\le \llo/12$ and 
therefore $\rr\D\llo-2\lhi\ge \rr(\D-1/6)\llo$, used 
that  $0<\D-1/6<1$, which implies that $P_1\ge (\D-1/6)^{2\rr} $, and 
defined $\czz \define \cz (\D-1/6)^2$ that satisfies $0<\czz<1$.

The bound~\fref{eq:toprove2} follows from~\fref{eq:philtau1} and~\fref{eq:philtau2} 
by defining $\cuuuuu \define \cuuuu/\czz$.

\paragraph{Proof of~\fref{eq:toprove1}: case~$t_j\in \setT^c_k$.}
We only need to consider this case when $\rr>1$. Indeed, when $\rr=1$, the sum in \fref{eq:q1defm} only contains
one element, $\phi_l(\cdot)$, and, necessarily, $t_j\in \setT_l$ because $\setT_l^c$ is empty. 

In this case $t_j$ is one of the elements among 
$\{v_1,\ldots,v_{\tilde\rr}\}\subset \{\vt_1,\ldots,\vt_{\hrr}\}$; in other words,
$t_j=v_{\tilde m}=\vt_{\hat m}$ for some $1\le \tilde m\le\tilde \rr, 1\le \hat m\le\hrr$.
The set $\setT_k\intersect\near{\rr\D\llo-\lhi}{t_j}$ is either 
empty or contains exactly one element. Let 
$b\define\abs{\setT_k\intersect\near{\rr\D\llo-\lhi}{t_j}}$. In 
the case when $b=1$, let $\{\tilde t\}\define \setT_k\intersect\near{\rr\D\llo-\lhi}{t_j}$.

\subparagraph{Bounding $E_1(\tau)$.}
Consider the case $b=1$.
By \fref{eq:mvt2} [Mean Value theorem] we can write
\begin{equation}
    \abs{\phi_{+,k}(\tau)}
    \le \abs{\phi_{+,k}(\tilde t)}+\abs{\phi_{+,k}'(\tilde t)}
        \abs{\tau- \tilde t}+\frac{1}{2}\abs{\phi_{+,k}''(\tau_m)}(\tau- \tilde t)^2
\end{equation}
with $\tau_m\in (\tilde t, \tau)$. Next, we use~\fref{eq:phi11m} 
and~\fref{eq:ph11bdnm0} to upper-bound $\abs{\phi_{+,k}(\tilde t)}$ by the 
right-hand side of~\fref{eq:ph11bdnm0}; use~\fref{eq:phi11m1} and~\fref{eq:bndphikdnear1f} to 
upper-bound $\abs{\phi_{+,k}'(\tilde t)}$ by the right-hand side of~\fref{eq:bndphikdnear1f}; 
use~\fref{eq:phi11ddm} to upper-bound $\abs{\phi_{+,k}''(\tau_m)}$. With these estimates we 
can further upper-bound $\abs{\phi_{+,k}(\tau)}$ as follows:
\begin{align}
    \abs{\phi_{+,k}(\tau)}
    \le\rr^{2\rr-1}\czdduff^\rr \left(\frac{\lhi^2}{\llo^2} 
        +\frac{\lhi}{\llo^2}\abs{\tau-\tilde t}
        +\frac{1}{\llo^2}(\tau-\tilde t)^2\right)
    &\le \rr^{2\rr+1} \czddufff^\rr \frac{(\tilde t-\tau)^2}{(\rr\llo)^2}\\
    &= \rr^{2\rr+1} \czddufff^\rr \left[\frac{(\tilde t-\tau)^2}{(\rr\llo)^2}\right]^\I{b=1},
    \label{eq:phipkubnd1}
\end{align}
where we used that $\lhi\le \abs{\tilde t-\tau}$ because the elements of $\setT$ are separated by at least $2\lhi$ and $\tau\in\near{\lhi}{t_j}$ with $\tilde t\ne t_j$. According to~\fref{eq:phi11ddm00} 
the upper bound~\fref{eq:phipkubnd1} also holds for $b=0$.

Since $t_j\in \setT^c_k$ and $\tau\in\near{\lhi}{t_j}$, it 
follows $\tau\in\near{\rr\D\llo}{\setT_k^c}$ 
so that $\hrr\ge 1$, which implies that we can use~\fref{eq:qzddbndpl} 
to upper-bound $\abs{\phi''_{0,k}(\tau)}$:
\begin{align}
    \abs{\phi_{0,k}''(\tau)}
    &\le 
        \sum_{1\le m\le\hrr} \sum_{\substack{1\le m'\le\hrr\\m\ne m'}} 
        \I{\hrr\ge 2}
        \czdu^{\hrr} 
        \frac{ \prod_{\substack{1\le l\le \hrr\\ l\ne m,\ l\ne m'}} (\vt_l-\tau)^{2}}{(\rr\llo)^{2(\hrr-2)}}
        \frac{\abs{\vt_m-\tau}}{(\rr\llo)^2} 
        \frac{\abs{\vt_{m'}-\tau}}{(\rr\llo)^2}\nonumber\\
    &\dummyrel{\le} 
        +2(\rr-1-\hrr)\sum_{1\le m\le\hrr} 
        \czdu^{\hrr+1} 
        \frac{ \prod_{\substack{1\le l\le \hrr\\ l\ne m}} (\vt_l-\tau)^{2}}{(\rr\llo)^{2(\hrr-1)}}
        \frac{\abs{\vt_m-\tau}}{(\rr\llo)^2}\frac{1}{\rr\llo}\nonumber\\
    &\dummyrel{\le} 
        +(\rr-1-\hrr)^2 \czdu^{\hrr+2} 
            \frac{ \prod_{l=1}^{\hrr} (\vt_l-\tau)^{2}}{(\rr\llo)^{2\hrr}}
            \frac{1}{(\rr\llo)^2}\nonumber\\
    &\dummyrel{\le}
        + \sum_{1\le m\le\hrr} 
            \czdu^{\hrr}
            \frac{ \prod_{\substack{1\le l\le \hrr\\ l\ne m}} (\vt_l-\tau)^{2}}{(\rr\llo)^{2(\hrr-1)}}
            \frac{1}{(\rr\llo)^2} 
        +(\rr-1-\hrr) 
            \czdu^{\hrr+1}
            \frac{ \prod_{l=1}^{\hrr} (\vt_l-\tau)^{2}}{(\rr\llo)^{2\hrr}}
            \frac{1}{(\rr\llo)^2}\\
    &\stackrel{(a)}{\le} 
        \rr^2\cuoss^{\rr+1}  
        \frac{ \prod_{1\le l\le \hrr, l\ne\hat m} (\vt_l-\tau)^{2}}{(\rr\llo)^{2\hrr}}.
\label{eq:qddub3}
\end{align}
Above, in (a) we used (multiple times) the fact that $\abs{\vt_{\hat m}-\tau}\le\abs{\vt_{l}-\tau}$ 
for all $l\in\{1,\ldots,\hrr\}$, the fact that $\abs{\vt_l-\tau}/(\llo \rr)\le \D<1$ 
for all $l\in\{1,\ldots,\hrr\}$, the fact $\hrr\le \rr-1$, and defined $\cuoss \define 6 \czdu$.

Multiplying~\fref{eq:phipkubnd1} and~\fref{eq:qddub3} and simplifying we obtain 
the following upper bound on $E_1$:
\begin{align}
    E_1(\tau)
    =\abs{\phi''_{0,k}(\tau)}\abs{\phi_{+,k}(\tau)}
    &\stackrel{(a)}{\le} 
        \rr^{2\rr+1}\cuosss^{\rr+1} 
        \left[\frac{(\tilde t-\tau)^2}{(\rr\llo)^2}\right]^\I{b=1} 
        \frac{ \prod_{1\le l\le \hrr, l\ne\hat m} (\vt_l-\tau)^{2}}{(\rr\llo)^{2(\hrr-1)}}
        \frac{1}{\llo^2}\nonumber\\
    &\stackrel{(b)}{\le} 
        \rr^{2\rr+1}\cuosss^{\rr+1} 
        \left[\frac{(\tilde t-\tau)^2}{(\rr\llo)^2}\right]^\I{b=1} 
        \frac{\prod_{1\le l\le \tilde\rr, l\ne\tilde m} (v_l-\tau)^{2}}{(\rr\llo)^{2(\tilde\rr-1)}}
        \frac{1}{\llo^2}.
    \label{eq:E12}
\end{align}
Above, in (a) we defined $\cuosss\define\czddufff\cuoss$; in (b) we use the fact 
that $\abs{\vt_l-\tau}/(\llo \rr)\le \D<1$ for all $l\in\{1,\ldots,\hrr\}$.

\subparagraph{Bounding $E_2(\tau)$.}
Consider the case $b=1$.
By \fref{eq:mvt1} [Mean Value theorem] we can write
\begin{equation}
    \abs{\phi'_{+,k}(\tau)}\le 
    \abs{\phi'_{+,k}(\tilde t)}+\abs{\phi_{+,k}''(\tau_m)}\abs{\tau- \tilde t}
\end{equation}
with $\tau_m\in (\tilde t, \tau)$. Next, we use~\fref{eq:phi11m1} and~\fref{eq:bndphikdnear1f} 
to upper-bound $\abs{\phi_{+,k}'(\tilde t)}$ by the right-hand side of~\fref{eq:bndphikdnear1f}; 
use~\fref{eq:phi11ddm} to upper-bound $\abs{\phi_{+,k}''(\tau_m)}$. With these estimates we can 
further upper-bound $\abs{\phi'_{+,k}(\tau)}$ as follows:
\begin{align}
    \abs{\phi'_{+,k}(\tau)}
    \le\rr^{2\rr-1}\czdduffd^\rr \left(\frac{\lhi}{\llo^2}
        +\frac{1}{\llo^2}\abs{\tau-\tilde t}\right)
    &\le \rr^{2\rr}\czddufffdd^\rr 
        \frac{\abs{\tau-\tilde t}}{\rr\llo}\frac{1}{\llo}\\
    &= \rr^{2\rr}\czddufffdd^\rr 
        \left[\frac{\abs{\tau-\tilde t}}{\rr\llo}\right]^{\I{b=1}}
        \frac{1}{\llo},
    \label{eq:phipkubnd2}
\end{align}
where we used that $\lhi\le \abs{\tilde t-\tau}$. According to~\fref{eq:phi11ddm0} 
the upper bound~\fref{eq:phipkubnd2} also holds for $b=0$.

Since $t_j\in \setT^c_k$ and $\tau\in\near{\lhi}{t_j}$, it 
follows $\tau\in\near{\rr\D\llo}{\setT_k^c}$ 
so that $\hrr\ge 1$, which implies that we can use~\fref{eq:qzdbndpl} 
to upper-bound $\abs{\phi'_{0,k}(\tau)}$:
\begin{align}
    \abs{\phi_{0,k}'(\tau)}
    &\le \sum_{m=1}^{\hrr} \czdu^{\hrr} 
        \frac{\prod_{\substack{1\le l\le \hrr\\l\ne m}} (\vt_{l}-\tau)^{2}\abs{\vt_m-\tau}}{(\rr\llo)^{2\hrr}} 
        +(\rr-1-\hrr) \czdu^{\hrr+1}  
        \frac{ \prod_{l=1}^{\hrr} (\vt_{l}-\tau)^{2}}{(\rr\llo)^{2\hrr+1}}\\
    &\stackrel{(a)}{\le}
        \cuossss^\rr\frac{\prod_{1\le l\le \hrr,l\ne \hat m} (\vt_{l}-\tau)^{2}}{(\rr\llo)^{2(\hrr-1)}}
        \left[\frac{\abs{\tilde t-\tau}}{\rr\llo}\right]^{\I{b=1}} \frac{1}{\llo}.
        \label{eq:qddub4}
\end{align}
Above, in (a) we used  the fact that $\abs{\vt_{\hat m}-\tau}\le\abs{\vt_{l}-\tau}$ 
for all $l\in\{1,\ldots,\hrr\}$, the fact that $\abs{\vt_l-\tau}/(\llo \rr)\le \D<1$ 
for all $l\in\{1,\ldots,\hrr\}$, and the fact that 
$\abs{\vt_{\hat m}-\tau}\le \abs{\tilde t-\tau}$, and defined $\cuossss\define 2 \czdu$.

Multiplying~\fref{eq:phipkubnd2} and~\fref{eq:qddub4} and simplifying we obtain 
the following upper bound on $E_2$:
\begin{align}
    E_2(\tau)
    =\abs{\phi'_{0,k}(\tau)}\abs{\phi'_{+,k}(\tau)}
    &\stackrel{(a)}{\le} \rr^{2\rr} \cuosssss^\rr 
        \left[\frac{(\tilde t-\tau)^2}{(\rr\llo)^2}\right]^{\I{b=1}}
        \frac{ \prod_{1\le l\le \hrr,l\ne \hat m} (\vt_{l}-\tau)^{2}}{(\rr\llo)^{2(\hrr-1)}} 
        \frac{1}{\llo^2} \nonumber\\
    &\stackrel{(b)}{\le} 
        \rr^{2\rr} \cuosssss^\rr 
        \left[\frac{(\tilde t-\tau)^2}{(\rr\llo)^2}\right]^{\I{b=1}}
        \frac{\prod_{1\le l\le \tilde\rr,l\ne \tilde m} (v_{l}-\tau)^{2}}{(\rr\llo)^{2(\tilde\rr-1)}} 
        \frac{1}{\llo^2}.
    \label{eq:E22}
\end{align}
Above, in (a)  we defined $\cuosssss \define \czddufffdd \cuossss$; in (b) we used the fact 
that $\abs{\vt_l-\tau}/(\llo \rr)\le \D<1$ for all $l\in\{1,\ldots,\hrr\}$.

\subparagraph{Bounding $E_3(\tau)$.}
By~\fref{eq:phi11ddm},
\begin{equation}
    \label{eq:phipkubnd3}
    \abs{\phi_{+,k}''(\tau)}\le \rr^{2\rr-2}\czdduf^{\rr}\cqdvdd \frac{1}{\llo^2}.
\end{equation}

Since $t_j\in \setT^c_k$ and $\tau\in\near{\lhi}{t_j}$, it 
follows $\tau\in\near{\rr\D\llo}{\setT_k^c}$ 
so that $\hrr\ge 1$, which implies that we can use~\fref{eq:phikub} to 
upper-bound $\abs{\phi_{0,k}(\tau)}$:
\begin{align}
    \abs{\phi_{0,k}(\tau)} 
    \le \cqu^{\hrr}\frac{ \prod_{l=1}^{\hrr} (\vt_l-\tau)^{2}}{(\rr\llo)^{2\hrr}}
    &\stackrel{(a)}{\le} 
        \cqu^{\hrr}
        \left[\frac{(\tilde t-\tau)^2}{(\rr\llo)^2}\right]^{\I{b=1}} 
        \frac{ \prod_{1\le l\le \hrr,l\ne \hat m} (\vt_{l}-\tau)^{2}}{(\rr\llo)^{2(\hrr-1)}}.
        \label{eq:qddub51}
\end{align}
Above, in (a) we used  the fact that $\abs{\vt_{\hat m}-\tau}\le\abs{\vt_{l}-\tau}$ 
for all $l\in\{1,\ldots,\hrr\}$, the fact that $\abs{\vt_l-\tau}/(\llo \rr)\le \D<1$ 
for all $l\in\{1,\ldots,\hrr\}$, and the fact that $\abs{\vt_{\hat m}-\tau}\le \abs{\tilde t-\tau}$.
Multiplying~\fref{eq:phipkubnd3} and~\fref{eq:qddub51} and simplifying we obtain the 
following upper bound on $E_3$:
\begin{align}
    E_3(\tau)
    =\abs{\phi_{0,k}(\tau)}\abs{\phi''_{+,k}(\tau)}
    &\stackrel{(a)}{\le} 
        \rr^{2\rr-2}\cuossssss^{\rr} 
        \left[\frac{(\tilde t-\tau)^2}{(\rr\llo)^2}\right]^{\I{b=1}} 
        \frac{\prod_{1\le l\le \hrr,l\ne \hat m} (\vt_{l}-\tau)^{2}}{(\rr\llo)^{2(\hrr-1)}}
        \frac{1}{\llo^2} \\
    &\stackrel{(b)}{\le} 
        \rr^{2\rr-2}\cuossssss^{\rr}
        \left[\frac{(\tilde t-\tau)^2}{(\rr\llo)^2}\right]^{\I{b=1}}
        \frac{\prod_{1\le l\le \tilde\rr, l\ne \tilde m} (v_l-\tau)^2}{(\rr\llo)^{2(\tilde\rr-1)}} 
        \frac{1}{\llo^2}.
    \label{eq:E23}
\end{align}
Above, in (a) we defined $\cuossssss\define\czdduf\cqdvdd\cqu$; in (b) we used the fact 
that $\abs{\vt_l-\tau}/(\llo \rr)\le \D<1$ for all $l\in\{1,\ldots,\hrr\}$.

From~\fref{eq:phiddsumtermsl1}, \fref{eq:E12}, \fref{eq:E22}, and \fref{eq:E23} we 
conclude that
\begin{equation}
     \abs{\phi''_{k}(\tau)}
     \le \rr^{2\rr+1}\cuuu^{\rr+1}  
     \left[\frac{(\tilde t-\tau)^2}{(\rr\llo)^2}\right]^\I{b=1} 
     \frac{\prod_{1\le l\le \tilde\rr, l\ne \tilde m} (v_l-\tau)^2}{(\rr\llo)^{2(\tilde\rr-1)}} 
     \frac{1}{\llo^2},
    \label{eq:phiddfinalbnd2}
\end{equation}
where we defined $\cuuu \define 4\max(\cuosss, \cuosssss, \cuossssss)$.

\subparagraph{Putting pieces together.}
On the one hand, by \fref{eq:mvt2} [Mean Value theorem], 
using~\fref{eq:init2}, \fref{eq:init3}, \fref{eq:phiddfinalbnd2}, and 
we can write for all $\tau\in\near{\lhi}{t_j}$:
\begin{align}
    \abs{\phi_k(\tau)}
    &\stackrel{(a)}{\le} 
        \frac{1}{2} \rr^{2\rr+1}\cuuu^{\rr+1}  
        \left[\frac{(\tilde t-\tau_m)^2}{(\rr\llo)^2}\right]^\I{b=1} 
        \frac{\prod_{1\le l\le \tilde\rr, l\ne \tilde m} (v_l-\tau_m)^2}{(\rr\llo)^{2(\tilde\rr-1)}} 
        \frac{(\tau-t_j)^2}{\llo^2}\\
    &\stackrel{(b)}{\le} \frac{1}{2} \rr^{2\rr+1}\cuuu^{\rr+1}  2^{\tilde\rr}
        \left[\frac{(\tilde t-\tau)^2}{(\rr\llo)^2}\right]^\I{b=1} 
        \frac{\prod_{1\le l\le \tilde\rr, l\ne \tilde m} (v_l-\tau)^2}{(\rr\llo)^{2(\tilde\rr-1)}} 
        \frac{(\tau-t_j)^2}{\llo^2}\\
    &\stackrel{(c)}{\le} \rr^{2\rr+3} \cuuuuuu^{\rr+1} 
        \left[\frac{(\tilde t-\tau)^2}{(\rr\llo)^2}\right]^\I{b=1}
        \frac{\prod_{1\le l\le \tilde\rr} (v_l-\tau)^2}{(\rr\llo)^{2\tilde\rr}}.
\label{eq:philtau11}
\end{align}
Above, 
in (a) $\tau_m\in (t_j, \tau)$; 
in (b) we used the fact that, for $l\ne\tilde m$, $\abs{v_l-\tau_m}<\abs{v_l-\tau}+\lhi<2\abs{v_l-\tau}$ and 
$\abs{\tilde t-\tau_m}<\abs{\tilde t-\tau}+\lhi<2\abs{\tilde t-\tau}$,
which is true because $\tau\in\near{\lhi}{t_j}$ and because the elements of $\setT$ are separated 
by at least $2\lhi$;
in (c) we defined $\cuuuuuu \define 2\cuuu$ 
and used the fact that $t_j=v_{\tilde m}$.

On the other hand, let 
$\{u^\tau_1,\ldots,u^\tau_{\breve \rr}\}\define \near{\rr\D\llo}{\tau}\intersect\setT$. 
Then by~\fref{eq:qzlp},
\begin{align}
    q_0(\tau)
    &\ge 
        \cz^\rr \frac{ \prod_{l=1}^{\breve  \rr}(u^\tau_l-\tau)^2}{(\rr\llo)^{2\breve \rr}}\\
    &\stackrel{(a)}{\ge} 
        \cz^\rr \left[\frac{(\tilde t-\tau)^2}{(\rr\llo)^2}\right]^\I{b=1} 
        \frac{\prod_{1\le l\le \tilde\rr} (v_l-\tau)^2}{(\rr\llo)^{2\tilde\rr}}
        \underbrace{
            \left(\frac{(\rr\D\llo-2\lhi)^2}{(\rr\llo)^2}\right)^{\breve\rr-\tilde\rr-\I{b=1}}
        }_{P_2}\\
    &\stackrel{(b)}{\ge} \czz^\rr \left[\frac{(\tilde t-\tau)^2}{(\rr\llo)^2}\right]^\I{b=1} 
        \frac{\prod_{1\le l\le \tilde\rr} (v_l-\tau)^2}{(\rr\llo)^{2\tilde\rr}}.
\label{eq:philtau21}
\end{align}
Above, in (a) we used the fact that 
$\{v_1,\ldots,v_{\tilde r}\}\subset \{u^\tau_1,\ldots,u^\tau_{\breve \rr}\}$, 
the fact that if $b=1$, then $\tilde t\in \{u^\tau_1,\ldots,u^\tau_{\breve \rr}\}$, 
and the fact that by construction of the set $\{v_1,\ldots,v_{\tilde r}\}$ it follows 
that if, for some $k$, $u^\tau_k\notin \{v_1,\ldots,v_{\tilde r}\}$ and 
$u^\tau_k\ne \tilde t$, then $\abs{u^\tau_k-\tau}\ge \rr\D\llo-2\lhi$; 
in (b) we used the assumption $\SRF\ge 12$ so that $\lhi\le \llo/12$ 
and therefore $\rr\D\llo-2\lhi\ge \rr(\D-1/6)\llo$, 
used that $0<\D-1/6<1$, which implies that $P_2\ge (\D-1/6)^{2\rr} $.

The bound~\fref{eq:toprove1} follows from~\fref{eq:philtau11} 
and~\fref{eq:philtau21} by defining $\cuf \define \cuuuuuu/\czz$.

\subsubsection{Proof of property 3}
By~\fref{eq:q1defm} and the triangle inequality:
\begin{align}
    \infnorm{q_1(\cdot)}
    \le \rho/2+\rr \max_{1\le k\le \rr}\infnorm{\phi_k(\cdot)}
    &\stackrel{(a)}{\le} \rho/2+\rr \max_{1\le k\le \rr}\infnorm{\phi_{+,k}(\cdot)}\\
    &\stackrel{(b)}{=} \rho/2+\rr^{2\rr+1}\czdduf^{\rr} \max_{1\le k\le \rr}\infnorm{q_{\rr\llo,\setT_k,\{f_j\},\{d_j\}}(\cdot)}\\
    &\stackrel{(c)}{\le} \rho/2+\rr^{2\rr+1}\cqdv\czdduf^{\rr}
    \stackrel{(d)}{\le} \rr^{2\rr+1}\cugm^{\rr}.
\end{align}
Above, in (a) we used~\fref{eq:phikdfn} and the fact that by~\fref{eq:phizk} and 
\fref{lem:dualcarlos}, \fref{prop:qzeroone}, $\infnorm{\phi_{0,k}(\cdot)}\le 1$;
in (b) we used~\fref{eq:phipdef}; in (c) we used \fref{lem:dualcarlos2}, \fref{prop:dqinfbnd};
in (d) we defined $\cugm\define 2\cqdv\czdduf$ and used the fact that $\rho/2<1<\cqdv\czdduf$.

\subsubsection{Proof of property 4}
\label{sec:bndfarV}

Take $\tau\in\far{\lhi}{\setT}$. 
As above, let $\{u^\tau_1,\ldots,u^\tau_{\breve \rr}\}\define \near{\rr\D\llo}{\tau}\intersect\setT$. 
Then by~\fref{eq:qzlp},
\begin{equation}
    \label{eq:qzlbpr}
    q_0(\tau)
    \ge \cz^\rr \frac{ \prod_{l=1}^{\breve  \rr}(u^\tau_l-\tau)^2}{(\rr\llo)^{2\breve \rr}}.
\end{equation}
By~\fref{eq:q0lbxx} this bound is also 
valid when $\breve \rr=0$.

Fix $k$. If $\tau\in\near{\rr\D\llo}{\setT_k^c}$, then we 
can use~\fref{eq:phikub} to upper-bound $\abs{\phi_{0,k}(\tau)}$:
\begin{align}
    \label{eq:phipke3}
    \abs{\phi_{0,k}(\tau)} 
    &\le \cqu^{\hrr}\frac{ \prod_{l=1}^{\hrr} (\vt_l-\tau)^{2}}{(\rr\llo)^{2\hrr}},
    \label{eq:qzkubpr}
\end{align}
where, as before, 
$\{\vt_1,\ldots,\vt_{\hrr}\}\define \near{\rr\D\llo}{\tau}\intersect\setT_k^c$. 
If $\tau\notin\near{\rr\D\llo}{\setT_k^c}$, we will use that by~\fref{eq:phizk} and by \fref{lem:dualcarlos}, \fref{prop:qzeroone},
\begin{equation}
    \abs{\phi_{0,k}(\tau)} \le 1.\label{eq:qzkubpruniv}
\end{equation}
The set $\setT_k\intersect\near{\rr\D\llo}{\tau}$ is either empty or 
contains exactly one element. Let $b\define\abs{ \setT_k\intersect\near{\rr\D\llo}{\tau}}$ 
denote the size of this set; when $b=1$, let $\{\tilde t\}\define \setT_k\intersect\near{\rr\D\llo}{ \tau}$.
Following the steps that lead to~\fref{eq:phipkubnd1}, we obtain
\begin{equation}
    \label{eq:qokubpr}
    \abs{\phi_{+,k}(\tau)}
    \le \rr^{2\rr+1}\czddufff^\rr 
        \left[\frac{(\tilde t-\tau)^2}{(\rr\llo)^2}\right]^\I{b=1}
\end{equation}
and the bound is valid for both cases $b=0$ and $b=1$.

\paragraph*{Case $\hrr\ge 1$:}
Then, $\{u^\tau_1,\ldots,u^\tau_{\breve \rr}\}=\{\vt_1,\ldots,\vt_{\hrr}\}\union \{\tilde t\}$ 
if $b=1$, and $\{u^\tau_1,\ldots,u^\tau_{\breve \rr}\}=\{\vt_1,\ldots,\vt_{\hrr}\}$ if $b=0$. 
Therefore,
\begin{align}
    \abs{\phi_{k}(\tau)}
    =\abs{\phi_{0,k}(\tau)}\abs{\phi_{+,k}(\tau)}
    &\stackrel{(a)}{\le} 
        \rr^{2\rr+1}\czddufff^\rr \cqu^{\hrr} 
        \left[\frac{(\tilde t-\tau)^2}{(\rr\llo)^2}\right]^\I{b=1} 
        \frac{ \prod_{l=1}^{\hrr} (\vt_l-\tau)^{2}}{(\rr\llo)^{2\hrr}}\\
    &= \rr^{2\rr+1} \czddufff^\rr \cqu^{\hrr}  
        \frac{ \prod_{l=1}^{\breve\rr} (u^\tau_l-\tau)^{2}}{(\rr\llo)^{2\breve\rr}}
    \stackrel{(b)}{\le} \rr^{2\rr+1} \cufff^\rr  q_0(\tau).
    \label{eq:phiq0ub1}
\end{align}
Above, (a) follows by~\fref{eq:qzkubpr} and~\fref{eq:qokubpr}; (b) follows by~\fref{eq:qzlbpr} 
with $\cufff \define \czddufff \cqu/\cz$.

\paragraph*{Case $\hrr= 0$:}
Then, $\breve r=1$ and $\{u^\tau_{\breve\rr}\}= \{\tilde t\}$ if $b=1$ 
and $\breve \rr=0$ if $b=0$. Therefore,
\begin{align}
    \abs{\phi_{k}(\tau)}
    =\abs{\phi_{0,k}(\tau)}\abs{\phi_{+,k}(\tau)}
    &\stackrel{(a)}{\le} \czddufff^\rr \rr^{2\rr+1}
        \left[\frac{(\tilde t-\tau)^2}{(\rr\llo)^2}\right]^\I{b=1}\nonumber\\
    &=\rr^{2\rr+1}\czddufff^\rr 
        \frac{ \prod_{l=1}^{\breve  \rr}(u^\tau_l-\tau)^2}{(\rr\llo)^{2\breve \rr}}
    \stackrel{(b)}{\le} \rr^{2\rr+1}\cufff^\rr  q_0(\tau).
    \label{eq:phiq0ub2}
\end{align}
Above, 
(a) follows by~\fref{eq:qzkubpruniv} and \fref{eq:qokubpr}; 
(b) follows by~\fref{eq:qzlbpr} because $\cqu>1$.

By \fref{lem:dualprod}, \fref{prop:qzloinfhi},
\begin{equation}
    \frac{\rho}{2}=\frac{\lhi^{2\rr}}{2\llo^{2\rr}} \le \frac{\rr^{2\rr}}{\cql^\rr} q_0(\tau).
    \label{eq:rho2ub}
\end{equation}
Therefore, by~\fref{eq:q1defm}, \fref{eq:phiq0ub1}, \fref{eq:phiq0ub2}, \fref{eq:rho2ub},%
\begin{equation}
    \abs{q_1(\tau)}
    \le \sum_{k=1}^\rr \abs{\phi_{k}(\tau)} + \rho/2
    \le \rr^{2\rr+2} \cuffff^\rr q_0(\tau),
\end{equation}
where we defined $\cuffff\define\cufff+1/\cql$.
\end{proof}

\subsection{Dual certificate $q_2(\cdot)$}
\label{sec:q2}
Finally, we construct the trigonometric polynomial $q_2(\cdot)$. This trigonometric 
polynomial is conceptually similar to $q_1(\cdot)$.
The difference is that in $q_1(\cdot)$ we control the function values on $\setT$ and the derivatives on
$\setT$ are zero; in $q_2(\cdot)$ we control the derivatives on $\setT$ and the function values on $\setT$
are zero.

Specifically, on the point $t_j\in\setT$, $q_2(\cdot)$ is approximated by a linear function whose 
derivative is controlled by the sign
\begin{equation}
    \label{eq:dsdef2}
    s'_{j}\define
    \sign\lefto(\sum_{m/\N\in\near{\lhi}{t_j}}\left(\frac{m}{\N}-t_j\right) h_m\right),\quad 
    j=1,\ldots,\sparsity,
\end{equation}
as explained in \fref{lem:dualprod1d} below.

\begin{lem}
\label{lem:dualprod1d}
Set $\gamma\define \rho/\lhi=\lhi^{2\rr-1}/\llo^{2\rr}$.
Then, there exists a real-valued trigonometric polynomial $q_2(\cdot)$
that satisfies the following properties.
\begin{enumerate}
    \item \label{prop:q1prop1d} Frequency limitation to 
    $\flo$: $q_2(t)=\sum_{k=-\flo}^{\flo} \hat q_{2,k} e^{-\iu 2\pi kt}$ for some $\hat q_{2,k}\in\complexset$.
    \item \label{prop:q1prop3d} Constrained derivative on $\setT$ and controlled 
    behavior near $\setT$: for all $j=1,\ldots,\sparsity$ and all $\tau\in\near{\lhi}{t_j}$,
    \begin{equation}
        \label{eq:bndneard}
        \abs{q_2(\tau)-\gamma s'_{j}(\tau-t_j)} \le  \rr^{2\rr+4}\cufk^{\rr+1}  q_0(\tau),
    \end{equation}
    where $s'_j$ are 
    defined\footnote{The lemma is valid for arbitrary sign pattern, we formulate it for the sign pattern defined in~\fref{eq:dsdef2} for concreteness.} in~\fref{eq:dsdef2}.
    \item \label{prop:q1prop2d}Uniform confinement: $\infnorm{q_2(\cdot)}\le \rr^{2\rr+1}\cugn^{\rr}$.
    \item \label{prop:q1prop4d} Boundedness far from $\setT$: for 
    all $\tau\in\far{\lhi}{\setT}$,
    \begin{equation}
        \label{eq:dualprod14d}
        \abs{q_2(\tau)}\le \rr^{2\rr+2} \cugh^\rr q_0(\tau).
    \end{equation}
\end{enumerate}
The positive numerical constants $\cufk$, $\cugn$, and $\cugh$ are defined in the proof below.
\end{lem}

The proof of the lemma parallels that of \fref{lem:dualprod1} but contains some 
important differences; it is given in~\fref{app:dualprod1d}.

\section{Stability estimates}
\label{sec:stability}
In this section we use the dual trigonometric polynomials $q_0(\cdot)$, $q_1(\cdot)$, and $q_2(\cdot)$ to prove \fref{thm:mainthm}.

We will use the fact that the high-resolution kernel $\khi(\cdot)$ satisfies 
the following estimates:
\begin{align}
    &\sum_{n=0}^{\N-1} \abs{\khi'\lefto(\frac{n}{\N}\right)}\le \frac{\cabshid}{\lhi},
        \label{eq:cabshid}\\
    &\frac{1}{2}\sum_{n=0}^{\N-1} \sup_{u\in\nearhi{n/\N}} \abs{\khi''(u)} 
        \le \frac{\cabshidd}{\lhi^2},
        \label{eq:cabshidd}
\end{align}
where $\cabshid$, $\cabshidd$ are positive numerical constants. The bounds are 
proven in~\fref{app:propfejer}.

We will use the following shorthand notations:
\begin{align}
    \nearh&\define\union_{t\in\setT} \near{\lhi}{t},\\
    \farh&\define\T\setdiff \nearh.
\end{align}

\subsection{Basic estimates}
We begin by decomposing the error $\onenorm{\khi\conv (\hat\inp-\inp)}$ into a sum of simpler 
terms; each of the terms will then be upper-bounded separately:
\begin{align}
    \onenorm{\khi\conv (\hat\inp-\inp)}
    &=
    \sum_{n=0}^{\N-1} 
    \abs{\sum_{m=0}^{\N-1} \khi\lefto(\frac{n-m}{\N}\right) h_m}\\
    &\le
    \sum_{n=0}^{\N-1} 
    \abs{\sum_{m/\N \in \farh} \khi\lefto(\frac{n-m}{\N}\right) h_m} 
        +\sum_{n=0}^{\N-1} 
             \abs{\sum_{m/\N \in \nearh} \khi\lefto(\frac{n-m}{\N}\right) h_m}.
    \label{eq:dfirstbnd}
\end{align}

The first term in~\fref{eq:dfirstbnd} can be written as follows:
\begin{align}
    \sum_{n=0}^{\N-1} \abs{\sum_{m/\N \in \farh} \khi\lefto(\frac{n-m}{\N}\right) h_m}
    &\stackrel{(a)}{=} 
        \sum_{m/\N \in \farh} \left(\sum_{n=0}^{\N-1} \khi\lefto(\frac{n-m}{\N}\right) \right) h_m\\
    &\stackrel{(b)}{=} 
        \sum_{m/\N \in \farh} \left(\sum_{n=0}^{\N-1} \khi\lefto(\frac{n}{\N}\right) \right) h_m
    \stackrel{(c)}{=} 
        \underbrace{\sum_{m/\N \in \farh} h_m}_{A_0}.
    \label{eq:dtrm1ub}
\end{align}
Above, 
(a) follows because $h_m\ge 0$ for $m/\N\in \farh$ and $\khi(\cdot)\ge 0$; 
(b) follows 
by periodicity of $\khi(\cdot)$; (c) follows  by~\fref{eq:cabshi}.

The second term in~\fref{eq:dfirstbnd} can be upper-bounded as follows:
\begin{align}
    \sum_{n=0}^{\N-1} 
        \abs{\sum_{m/\N \in \nearh} \khi\lefto(\frac{n-m}{\N}\right) h_m}
    &\stackrel{(a)}{=} 
        \sum_{n=0}^{\N-1} 
        \abs{\sum_{j=1}^{\sparsity} \sum_{m/\N \in \nearhi{t_j}} \khi\lefto(\frac{n-m}{\N}\right) h_m}\\
    &\stackrel{(b)}{\le} 
        \underbrace{\sum_{n=0}^{\N-1} \sum_{j=1}^{\sparsity} 
        \abs{\sum_{m/\N \in \nearhi{t_j}} \khi\lefto(\frac{n-m}{\N}\right) h_m}}_{B}.
    \label{eq:dsecterm1}
\end{align}
Above, 
(a) follows because the sets $\nearhi{t_j}$ do not intersect; 
(b) follows by the triangle inequality.
To upper-bound $B$ in~\fref{eq:dsecterm1} we  will use that for 
all $\tau\in\nearhi{t_j}\intersect\T$ and all $t\in\T$,
\begin{align}
    \abs{\khi(t-\tau)-\khi(t- t_j)-\khi'(t-t_j)(t_j-\tau)}
    \le 
    \sup_{u\in\nearhi{t- t_j}} \frac{1}{2} \abs{\khi''(u)} (t_j-\tau)^2.
    \label{eq:dtaylor}
\end{align}
The inequality follows by expanding $\khi(t-\tau)$ in Taylor series in $\tau$ 
around $\tau=t_j$ up to first order and writing the remainder in Lagrange form. We have:

\begin{align}
    &\abs{\sum_{m/\N \in \nearhi{t_j}} \khi\lefto(\frac{n-m}{\N}\right) h_m}\\
    &\qquad\stackrel{(a)}{\le} 
    \abs{\sum_{m/\N \in \nearhi{t_j}} \khi\lefto(\frac{n}{\N}-t_j\right) h_m}
        + \abs{\sum_{m/\N \in \nearhi{t_j}} \khi'\lefto(\frac{n}{\N}-t_j\right)\left(t_j-\frac{m}{\N}\right) h_m} \\
    &\qquad\dummyrel{\le}
        + \sum_{m/\N \in \nearhi{t_j}} 
        \abs{\khi\lefto(\frac{n-m}{\N}\right)
            -\khi\lefto(\frac{n}{\N}-t_j\right)
            -\khi'\lefto(\frac{n}{\N}-t_j\right)
            \left(t_j-\frac{m}{\N}\right)} 
        \abs{h_m} \nonumber\\
    &\qquad\stackrel{(b)}{\le} 
    \khi\lefto(\frac{n}{\N}-t_j\right) 
    \abs{\sum_{m/\N \in \nearhi{t_j}}  h_m}
    + \abs{\khi'\lefto(\frac{n}{\N}-t_j\right)}
        \abs{\sum_{m/\N \in \nearhi{t_j}} \left(t_j-\frac{m}{\N}\right) h_m}\\
    &\qquad\dummyrel{\le}
    + \sum_{m/\N \in \nearhi{t_j}} \sup_{u\in\nearhi{n/\N - t_j}} 
        \frac{1}{2} \abs{\khi''(u)}\lefto(t_j-\frac{m}{\N}\right)^2 \abs{h_m}.
    \label{eq:dtaylorcons}
\end{align}
Above, 
(a) follows by adding and subtracting the corresponding terms and applying 
the triangle inequality;
(b) follows by~\fref{eq:dtaylor} with $t=n/\N$ and $\tau = m/\N$ 
and because $\khi(\cdot)\ge 0$.

Using~\fref{eq:dtaylorcons} we can upper-bound $B$ in~\fref{eq:dsecterm1} as follows
{
\allowdisplaybreaks
\begin{align}
    B&\le 
        \sum_{j=1}^{\sparsity}
        \left(\sum_{n=0}^{\N-1} \khi\lefto(\frac{n}{\N}-t_j\right)\right)
        \abs{\sum_{m/\N \in \nearhi{t_j}} h_m}\\
    &\dummyrel{\le} 
        +\sum_{j=1}^{\sparsity}
         \left(\sum_{n=0}^{\N-1} \abs{\khi'\lefto(\frac{n}{\N}-t_j\right)}\right)
         \abs{\sum_{m/\N \in \nearhi{t_j}} \left(t_j-\frac{m}{\N}\right) h_m} \\
    &\dummyrel{\le}  
        +\sum_{j=1}^{\sparsity}\sum_{n=0}^{\N-1} 
         \sup_{u\in\nearhi{n/\N-t_j}} \frac{1}{2} 
         \abs{\khi''(u)} \sum_{m/\N \in \nearhi{t_j}}
         \lefto(t_j-\frac{m}{\N}\right)^2 \abs{h_m}\\
    &\stackrel{(a)}{=} 
        \left(\sum_{n=0}^{\N-1} \khi
        \lefto(\frac{n}{\N}\right)\right) 
        \sum_{j=1}^{\sparsity} \abs{\sum_{m/\N \in \nearhi{t_j}}  h_m}\\
    &\dummyrel{\le} 
        +\left(\sum_{n=0}^{\N-1} \abs{\khi'\lefto(\frac{n}{\N}\right)}\right)
         \sum_{j=1}^{\sparsity} 
         \abs{\sum_{m/\N \in \nearhi{t_j}} \left(t_j-\frac{m}{\N}\right) h_m}\\
    &\dummyrel{\le} 
        +\left(\sum_{n=0}^{\N-1} \sup_{u\in\nearhi{n/\N}} \frac{1}{2} \abs{\khi''(u)} \right)
         \sum_{j=1}^{\sparsity} \sum_{m/\N \in \nearhi{t_j}}
         \lefto(t_j-\frac{m}{\N}\right)^2 \abs{h_m}\\
    &\stackrel{(b)}{\le} 
        \underbrace{\sum_{j=1}^{\sparsity} 
        \abs{\sum_{m/\N \in \nearhi{t_j}}  h_m}}_{A_1}
        +\,\cabshid\underbrace{\frac{1}{\lhi}\sum_{j=1}^{\sparsity} \abs{\sum_{m/\N \in \nearhi{t_j}} \left(t_j-\frac{m}{\N}\right) h_m}}_{A_2} \\
    &\qquad\dummyrel{\le}\hspace{7.6em}
        +\cabshidd\underbrace{\frac{1}{\lhi^2}\sum_{j=1}^{\sparsity} \sum_{m/\N \in \nearhi{t_j}}\lefto(t_j-\frac{m}{\N}\right)^2 \abs{h_m}}_{A_3}.
    \label{eq:stoodin}
\end{align}
Above, (a) follows by  periodicity of $\khi(\cdot)$; (b) follows 
by~\fref{eq:cabshi},~\fref{eq:cabshid},~\fref{eq:cabshidd}.

To complete the proof of \fref{thm:mainthm}, it remains to upper-bound each 
of the terms $A_0$, $A_1$, $A_2$, and $A_3$ by $\sim\CC{\rr} \SRF^{2\rr} \onenorm{\noise}$.
To do this we will use extended duality arguments that will rely on 
the trigonometric polynomials $q_0(\cdot)$, $q_1(\cdot)$, and $q_2(\cdot)$.

\subsection{Upper bound on $A_0$}
\label{sec:da0}
In this section we use the trigonometric polynomial $q_0(\cdot)$ 
constructed in \fref{lem:dualprod} to upper-bound~$A_0$. Let
\begin{equation}
    \vecq^0 = \tp{[q^0_0, \ldots, q^0_{\N-1}]} \define \tp{[q_{0}(l/\N): l\in[0:\N-1]]}
    \label{eq:vecq0}
\end{equation}
be the vector that consists of the samples of~$q_{0}(\cdot)$.

On the one hand,
\begin{align}
    \inner{\vecq^0}{\vech}
    \stackrel{(a)}{=}\inner{\matQ\vecq^0}{\vech}
    \stackrel{(b)}{=}\inner{\vecq^0}{\matQ\vech}
    \stackrel{(c)}{\le} \infnorm{\vecq^0}\onenorm{\matQ\vech}
    &\stackrel{(d)}{\le} \onenorm{\matQ\hat\inp -\data+\data-\matQ\inp}\\
    &\stackrel{(e)}{\le} \onenorm{\matQ\hat\inp -\data}+\onenorm{\data-\matQ\inp}
    \stackrel{(f)}{\le} 2\onenorm{\noise}.
    \label{eq:dq0bnd1}
\end{align}
Above, 
(a) follows because by \fref{lem:dualprod}, \fref{prop:qlowf}, $q_0(\cdot)$ is frequency-limited to $\flo$, 
and, therefore, the vector of its samples is also frequency limited (in discrete sense) so 
that $\vecq^0=\matQ\vecq^0$; 
(b) follows because $\matQ$ is self-adjoint; (c) follows by Cauchy-Schwartz inequality; 
(d) follows by \fref{lem:dualprod}, \fref{prop:zeroonebnd}; 
(e) follows by the triangle inequality; (f) follows since~\fref{eq:find0} 
implies $\onenorm{\matQ\hat\inp -\data}\le \onenorm{\matQ\inp -\data}$ and, 
by assumption, $\data=\matQ\inp+\noise$.

On the other hand,
\begin{align}
    \inner{\vecq^0}{\vech}
    \stackrel{(a)}{=}
        \sum_{m=0}^{\N-1} q^0_m \abs{h_m}
    \stackrel{(b)}{\ge}
        \sum_{m/\N\in\farh} q^0_m \abs{h_m}
    \stackrel{(c)}{\ge}
        \cql^\rr \frac{\lhi^{2\rr}}{(\rr\llo)^{2\rr}}\sum_{m/\N\in\farh} \abs{h_m}.
\label{eq:dq0bnd2}
\end{align}
Above, 
(a) follows because, by construction, $q_0(t)=0$ for all $t\in\setT$, which means that 
$h_m<0$ implies $q_m^0=q_0(m/N)=0$, so that 
$q^0_m h_m\ge 0$ for $m=0,\ldots, \N-1$; 
(b) follows because all terms in the sum are nonnegative; 
(c) follows from \fref{lem:dualprod}, \fref{prop:qzloinfhi}.
From~\fref{eq:dq0bnd1} and~\fref{eq:dq0bnd2}, we conclude that
\begin{equation}
    \label{eq:dlhabs}
    A_0
    =\sum_{m/\N\in\farh} \abs{h_m}
    \le \frac{\rr^{2\rr}}{\cql^\rr}  \SRF^{2\rr}\onenorm{\noise},
\end{equation}
where the equality follows because $h_m\ge 0$ for $m/N\in\farh$ and we remind the reader that $\cql<1$.

\subsection{Upper bound on $A_3$}
\label{sec:da3}
In this section we use $\vecq^0$ to upper-bound~$A_3$.
We have,
\begin{align}
    2\onenorm{\noise}
    \stackrel{(a)}{\ge}
        \sum_{m=0}^{\N-1} q^0_m \abs{h_m}
    &\stackrel{(b)}{\ge}
        \sum_{j=1}^\sparsity \sum_{m/\N \in \nearhi{t_j}} q^0_m \abs{h_m}\\
    &\stackrel{(c)}{\ge}
        \sum_{j=1}^\sparsity \sum_{m/\N \in \nearhi{t_j}} 
        \cz^{\rr} 
        \frac{ (t_j-m/\N)^2 \lhi^{2(\rr-1)}}{(\rr\llo)^{2\rr}} \abs{h_m}.
\end{align}
Above, 
(a) follows from~\fref{eq:dq0bnd1} because $q^0_m h_m\ge 0$ for $m=0,\ldots, \N-1$; 
(b) follows because the sets $\nearhi{t_j}$ do not intersect since the elements of $\setT$
are separated by at least $2\lhi$; 
(c) follows from~\fref{eq:qzl}.
Hence,
\begin{equation}
    \label{eq:dbnda3final}
    A_3
    =\frac{1}{ \lhi^2 }\sum_{j=1}^\sparsity \sum_{m/\N \in \nearhi{t_j}} 
        \left(t_j-\frac{m}{\N}\right)^2  \abs{h_m}
    \le \frac{\rr^{2\rr}}{\cz^\rr}  \SRF^{2\rr}2\onenorm{\noise}
\end{equation}
and we remind the reader that $\cz<1$.

\subsection{Upper bound on $A_1$}
\label{sec:dq1}
In this section we use trigonometric polynomial $q_1(\cdot)$ constructed 
in \fref{lem:dualprod1} to upper-bound $A_1$.
Set
\begin{align}
    \vecq^1 &= \tp{[q^1_0, \ldots, q^1_{\N-1}]} \define \tp{[q_{1}(l/\N): l\in[0:\N-1]]}.
    \label{eq:vecq0q1}
\end{align}

We now proceed as follows:
\begin{align}
    A_1=
        \sum_{j=1}^\sparsity \abs{\sum_{m/\N\in\nearhi{t_j}} h_m}
    &\stackrel{(a)}{=}
        \frac{2}{\rho}\sum_{j=1}^\sparsity \sum_{m/\N\in\nearhi{t_j}} \frac{\rho s_j}{2} h_m\\
    &\stackrel{(b)}{=}
        \frac{2}{\rho} \abs{\sum_{j=1}^\sparsity \sum_{m/\N\in\nearhi{t_j}} 
            \left(\frac{\rho s_j}{2}-q^1_m\right) h_m+\sum_{j=1}^\sparsity \sum_{m/\N\in\nearhi{t_j}} q^1_m h_m}\\
    &\stackrel{(c)}{\le}
        \frac{2}{\rho} \underbrace{\sum_{j=1}^\sparsity \sum_{m/\N\in\nearhi{t_j}} \abs{\frac{\rho s_j}{2}-q^1_m} \abs{h_m} }_{A_{11}}
        +\frac{2}{\rho} \underbrace{\abs{ \sum_{m/\N\in\nearh} q^1_m h_m}}_{A_{12}}.
\label{eq:da1bnd}
\end{align}
Above, 
(a) follows by~\fref{eq:dsdef}; 
(b) follows by adding and subtracting 
the corresponding term and because the expression in (a) is nonnegative; 
(c) follows by the triangle inequality and because the sets~$\nearhi{t_j}$ do not intersect since the
elements of $\setT$ are separated by at least $2\lhi$. 
Next, we upper-bound the terms $A_{11}$ and $A_{12}$ separately.

The first term in~\fref{eq:da1bnd}, $A_{11}$, can be upper-bounded as follows:
\begin{align}
    A_{11}
    \stackrel{(a)}{\le} 
        \rr^{2\rr+4} \cuff^{\rr+1}  \sum_{m/\N\in\nearh} q^0_m \abs{h_m}
    &\stackrel{(b)}{=} 
        \rr^{2\rr+4} \cuff^{\rr+1}  \sum_{m/\N\in\nearh} q^0_m h_m\\
    &\stackrel{(c)}{\le} 
        \rr^{2\rr+4} \cuff^{\rr+1}  \sum_{m=0}^{\N-1} q^0_m h_m
    \stackrel{(d)}{\le} 
        \rr^{2\rr+4} 2\cuff^{\rr+1} \onenorm{\noise}.
    \label{eq:da11bnd}
\end{align}
Above, 
(a) follows by~\fref{eq:bndnear} and because the sets~$\nearhi{t_j}$ do not intersect;
(b) follows because $h_m<0$ implies $q_m^0=0$; (c) follows because  $q^0_m h_m\ge 0$ for $m=0,\ldots, \N-1$; 
(d) follows by~\fref{eq:dq0bnd1}.

Following exactly the same steps as in~\fref{eq:dq0bnd1}, changing $q_m^0$ to $q_m^1$, and using in step (d) that
by \fref{lem:dualprod1}, \fref{prop:q1prop2}, $\infnorm{\vecq^1}\le \rr^{2\rr+1}\cugm^{\rr}$, we obtain:
\begin{equation}
    \label{eq:dq1intub}
    \abs{\sum_{m=0}^{\N-1} q^1_m h_m}\le 2\rr^{2\rr+1} \cugm^{\rr}\onenorm{\noise}.
\end{equation}
Using this, the second term in~\fref{eq:da1bnd}, $A_{12}$, can be upper-bounded as follows
\begin{align}
    A_{12}
    \stackrel{(a)}{\le} 
        \abs{\sum_{m=0}^{\N-1} q^1_m h_m}+\abs{\sum_{m/\N\in\farh} q^1_m h_m}
    &\stackrel{(b)}{\le} 
        2\rr^{2\rr+1}\cugm^{\rr}\onenorm{\noise}+\sum_{m/\N\in\farh} \abs{q^1_m}\abs{h_m}\\
    &\stackrel{(c)}{\le} 
        2\rr^{2\rr+1}\cugm^{\rr}\onenorm{\noise}+\rr^{2\rr+2} \cuffff^\rr\sum_{m/\N\in\farh} q^0_m \abs{h_m}\\
    &\stackrel{(d)}{=} 
        2\rr^{2\rr+1}\cugm^{\rr}\onenorm{\noise}+\rr^{2\rr+2} \cuffff^\rr\sum_{m/\N\in\farh} q^0_m h_m\\
    &\stackrel{(e)}{\le} 
        4 \rr^{2\rr+2} \cugo^\rr\onenorm{\noise}.
    \label{eq:da12bnd}
\end{align}
Above, 
(a) follow by the triangle inequality and because $\farh$ is complementary to $\nearh$;
(b) follow by~\fref{eq:dq1intub} and by the triangle inequality; 
(c) follow by~\fref{eq:dualprod14}; 
(d) follows because $h_m>0$ for $m/\N\in \farh$; 
(e) follows by~\fref{eq:dq0bnd1} because $q^0_m h_m\ge 0$ for $m=0,\ldots, \N-1$, because $\cuffff>1$,
and by defining~$\cugo\define\max(\cugm, \cuffff)$.

Substituting~\fref{eq:da11bnd} and~\fref{eq:da12bnd} into~\fref{eq:da1bnd}, using that $1/\rho=\SRF^{2\rr}$, we finally obtain
\begin{equation}
    \label{eq:dbnda1final}
    A_1=\sum_{j=1}^{\sparsity} \abs{\sum_{m/\N \in \nearhi{t_j}}  h_m}\le \rr^{2\rr+4}\cugk^{\rr+1}  \SRF^{2\rr} \onenorm{\noise},
\end{equation}
where we defined $\cugk \define 12\max(\cuff, \cugo)$.

\subsection{Upper bound on $A_2$}
\label{sec:dq2}
In this section we use trigonometric polynomial $q_2(\cdot)$ 
to upper-bound $A_2$.
Set
\begin{align}
    \vecq^2 &= \tp{[q^2_0, \ldots, q^2_{\N-1}]} \define \tp{[q_{2}(l/\N): l\in[0:\N-1]]}.
    \label{eq:vecq0q1}
\end{align}

We now proceed as follows:
\begin{align}
    A_2&=
        \frac{1}{\lhi}\sum_{j=1}^\sparsity 
            \abs{\sum_{m/\N\in\nearhi{t_j}} \left(\frac{m}{\N} - t_j\right)h_m}\\
    &\stackrel{(a)}{=}
        \frac{1}{\rho}\sum_{j=1}^\sparsity 
            \sum_{m/\N\in\nearhi{t_j}} \gamma s'_j \left( \frac{m}{\N} - t_j\right) h_m\\
    &\stackrel{(b)}{=}
        \frac{1}{\rho} 
        \abs{
            \sum_{j=1}^\sparsity \sum_{m/\N\in\nearhi{t_j}} 
                \left(\gamma s'_j \left( \frac{m}{\N} -t_j\right)-q^2_m\right) h_m
            +\sum_{j=1}^\sparsity \sum_{m/\N\in\nearhi{t_j}} q^2_m h_m}\\
    &\stackrel{(c)}{\le}
        \frac{1}{\rho} \underbrace{\sum_{j=1}^\sparsity \sum_{m/\N\in\nearhi{t_j}} 
            \abs{\gamma s'_j \left( \frac{m}{\N} - t_j\right)-q^2_m} \abs{h_m} }_{A_{21}}
        +\frac{1}{\rho} \underbrace{\abs{ \sum_{m/\N\in\nearh} q^2_m h_m}}_{A_{22}}.
\label{eq:da2bnd}
\end{align}
Above, 
(a) follows by~\fref{eq:dsdef2} and because $\gamma=\rho/\lhi$;
(b) follows by adding and subtracting the corresponding term and because the 
expression in (a) is nonnegative;
(c) follows by the triangle inequality and because the sets~$\nearhi{t_j}$ do not intersect. Next, we upper-bound the terms $A_{21}$ and $A_{22}$ separately.

The first term in~\fref{eq:da2bnd}, $A_{21}$, can be upper-bounded as follows
\begin{align}
    A_{21}
    \stackrel{(a)}{\le} 
        \rr^{2\rr+4}\cufk^{\rr+1}  \sum_{m/\N\in\nearh} q^0_m \abs{h_m}
    &\stackrel{(b)}{=} 
        \rr^{2\rr+4}\cufk^{\rr+1} \sum_{m/\N\in\nearh} q^0_m h_m\\
    &\stackrel{(c)}{\le} 
        \rr^{2\rr+4}\cufk^{\rr+1} \sum_{m=0}^{\N-1} q^0_m h_m\\
    &\stackrel{(d)}{\le} 
        2\rr^{2\rr+4}\cufk^{\rr+1} \onenorm{\noise}.
    \label{eq:da21bnd}
\end{align}
Above, 
(a) follows by~\fref{eq:bndneard} and because the sets~$\nearhi{t_j}$ do not intersect; 
(b) follows because $h_m<0$ implies $q_m^0=0$; 
(c) follows because $q^0_m h_m\ge 0$ for $m=0,\ldots, \N-1$; 
(d) follows by~\fref{eq:dq0bnd1}.

Following exactly the same steps as in~\fref{eq:dq0bnd1}, changing $q_m^0$ to $q_m^2$,  and using that by
\fref{lem:dualprod1d}, \fref{prop:q1prop2d}, $\infnorm{\vecq_2}\le \rr^{2\rr+1}\cugn^{\rr}$, we obtain:
\begin{equation}
    \label{eq:dq2intub}
    \abs{\sum_{m=0}^{\N-1} q^2_m h_m}\le 2\rr^{2\rr+1}\cugn^{\rr}\onenorm{\noise}.
\end{equation}
Using this, the second term in~\fref{eq:da2bnd}, $A_{22}$, can be upper-bounded as follows
\begin{align}
    A_{22}
    \stackrel{(a)}{\le} 
        \abs{\sum_{m=0}^{\N-1} q^2_m h_m}+\abs{\sum_{m/\N\in\farh} q^2_m h_m}
    &\stackrel{(b)}{\le} 
        2\rr^{2\rr+1}\cugn^{\rr}\onenorm{\noise}+\sum_{m/\N\in\farh} \abs{q^2_m}\abs{h_m}\\
    &\stackrel{(c)}{\le} 
        2\rr^{2\rr+1}\cugn^{\rr}\onenorm{\noise}+\rr^{2\rr+2} \cugh^\rr\sum_{m/\N\in\farh} q^0_m \abs{h_m}\\
    &\stackrel{(d)}{=} 
        2\rr^{2\rr+1}\cugn^{\rr}\onenorm{\noise}+\rr^{2\rr+2} \cugh^\rr\sum_{m/\N\in\farh} q^0_m h_m\\
    &\stackrel{(e)}{\le}
        4\rr^{2\rr+2} \cugp^\rr \onenorm{\noise}.
    \label{eq:da22bnd}
\end{align}
Above,
(a) follow by the triangle inequality and because $\farh$ is complementary to $\nearh$;
(b) follow by~\fref{eq:dq2intub} and by the triangle inequality;
(c) follow by~\fref{eq:dualprod14d};
(d) follows because $h_m>0$ for $m/\N\in \farh$;
(e) follows by~\fref{eq:dq0bnd1} because $q^0_m h_m\ge 0$ for $m=0,\ldots, \N-1$, because $\cugh>1$,
and by defining $\cugp\define\max(\cugn, \cugh)$.

Substituting~\fref{eq:da21bnd} and~\fref{eq:da22bnd} into~\fref{eq:da2bnd}, using that $1/\rho=\SRF^{2\rr}$, we finally obtain:
\begin{equation}
    \label{eq:dbnda2final}
    A_2=\frac{1}{\lhi}\sum_{j=1}^{\sparsity} \abs{\sum_{m/\N \in \nearhi{t_j}}  h_m}
    \le \rr^{2\rr+4}\cugl^{\rr+1} \SRF^{2\rr} \onenorm{\noise},
\end{equation}
where we defined $\cugl\define 6\max(\cufk, \cugp)$.

\subsection{Putting pieces together}
Substituting~\fref{eq:dlhabs} into~\fref{eq:dtrm1ub}; substituting~\fref{eq:dbnda3final}, 
\fref{eq:dbnda1final}, \fref{eq:dbnda2final} into~\fref{eq:stoodin} and 
the result into~\fref{eq:dsecterm1}; then, substituting~\fref{eq:dtrm1ub} 
and~\fref{eq:dsecterm1} into~\fref{eq:dfirstbnd}, and defining
\begin{equation}
    \label{eq:ccdef}
    \cc \define 4 \max(1/\cql, \cabshidd/\cz, \cugk, \cabshid\cugl)
\end{equation}
we obtain the desired bound~\fref{eq:stability} and complete the proof of \fref{thm:mainthm}. \qed

\section{Connection to Bernstein theorem}
\label{sec:bern}
The famous Bernstein theorem states the following~\cite[Ch.~4, eq.~(1.1)]{devore93a}.
\begin{thm}[Bernstein]
    \label{thm:bernstein}
    Consider a trigonometric polynomial frequency-limited to~$\fc$: 
    $q(t)=\sum_{k=-\fc}^{\fc} \hat q_k e^{-\iu 2\pi kt}$. Then,
    \begin{equation}
        \label{eq:bernstein}
        \infnorm{q'(\cdot)}\le 2\pi \fc\infnorm{q(\cdot)}. 
    \end{equation}
\end{thm}
In other words, if a trigonometric polynomial is uniformly bounded, its 
derivative cannot be too large anywhere.

Bernstein theorem helped us construct trigonometric polynomials $q_0(\cdot)$, $q_1(\cdot)$, and $q_2(\cdot)$ 
with the required properties by telling us what may be achievable and what is forbidden. We now describe these
connections to provide more intuition about our constructions.

\paragraph{Independent control.}
Consider $q(\cdot) = q_{\lc, \setV, \{f_j\}, \{d_j\}}(\cdot)$ in \fref{fig:toolsmult2n}. Since we require $\infnorm{q(\cdot)}\le \cqdv$, then, by Bernstein theorem, $\infnorm{q'(\cdot)}\le 2\pi\cqdv\fc$.
Suppose, $q(v_1)$ = 0. How large $q(v_2)$ may possibly be? Since $\infnorm{q'(\cdot)}\le 2\pi\cqdv\fc$, we must have
$q(v_2)\le 2\pi\cqdv (v_2-v_1)\fc$. Now, if the points $v_1$ and $v_2$ are well-separated, i.e., if $v_2-v_1$ is order $\lc$, Bernstein theorem puts \emph{no restrictions} on $q(v_2)$. However, if $v_2-v_1\ll \lc$, then 
$\abs{q(v_1)-q(v_2)}\le 2\pi\cqdv (v_2-v_1)\fc\ll 1$. Generalizing: it may be possible to independently control $q(v_1)$ and $q(v_2)$ \emph{only} if the points $v_1$ and $v_2$ are well-separated. This is the reason why $q_0(\cdot)$, $q_1(\cdot)$, and $q_2(\cdot)$ are constructed in an \emph{interlaced} way. We control the building blocks on sets of interlaced points that are well-separated, then we multiply the resulting trigonometric polynomials. See~\fref{eq:q0asprod} and \fref{fig:proof} for an easy example of interlacing; see \fref{eq:phikdfn}, \fref{eq:phizk}, \fref{fig:sumneed}, and \fref{fig:prodd} for a more sophisticated example of interlacing.

For readers familiar with using~$\ell_1$-minimization for super-resolution of real-valued (spikes may be positive
\emph{and} negative) and complex-valued
signals~\cite{candes13}: Bernstein theorem is responsible for the fact that~$\ell_1$-minimization fails when the
spikes are not well-separated (closer than $\llo$ to one another). The dual certificate in the real-valued case is a 
trigonometric polynomial $q(\cdot)$ with $\infnorm{q(\cdot)}\le 1$ that interpolates the \emph{sign} of the spikes in the signal. If, say 
$q(v_1)=-1$, and $v_2-v_1\ll \lc$, it is not possible that $q(v_2)=+1$ because $\infnorm{q'(\cdot)}\le 2\pi\fc$. The required dual trigonometric polynomial
does not exist and the algorithm fails.

In contrast to the real-valued case, consider our trigonometric polynomial $q_1(\cdot)$, displayed in \fref{fig:toolsneed}. Here, we interpolate the sign of the sequence $s_1, s_2, s_3, s_4$ at a set of points $t_1, t_2, t_3, t_4$ that are \emph{not} well-separated. How is that possible? The difference is that we interpolate the sign sequence at a \emph{low} level $\rho=(\lhi/\llo)^{2\rr}\ll 1$, i.e., we interpolate the points $s_i\rho/2$, and not the points $s_i$. The transitions $q_1(\cdot)$ needs to make between the points, are small; for example $\abs{q_1(t_3)-q_1(t_4)}=\rho\ll 1$ and this is not disallowed by Bernstein theorem.

\paragraph{High curvature.}
As should be clear by now, the curvature of the building block $q_{\lc, \setV}(\cdot)$ in the vicinity of its zeros expressed by \fref{eq:qnearbndeq}
(see also the sections marked in red in \fref{fig:toolsmult2n1}) determines the noise amplification in our bounds. How curvy can $q_{\lc, \setV}(\cdot)$ possibly be? Since $\infnorm{q_{\lc, \setV}(\cdot)}\le 1$, applying Bernstein
theorem twice, we conclude that the 
second derivative must satisfy $\infnorm{q''_{\lc, \setV}(\cdot)}\le 4\pi^2\fc^2$. Therefore, for $v\in\setV$, 
it must hold that $q_{\lc, \setV}(v-\tau)\le 2\pi^2(v-\tau)^2/\lc^2$. We conclude that the curvature of  
$q_{\lc, \setV}(\cdot)$ in \fref{eq:qnearbndeq} depends on $\lc$ in an optimal way (up to a constant). This leads
to the near-optimal stability estimate in \fref{thm:mainthm}.

\section{Conclusion}
When a signal is positive and Rayleigh-regular, then linear programming 
solves the super-resolution problem with near-optimal worst-case performance.
This result holds independently on how fine the discretization grid is, 
approximating the continuum arbitrarily closely.
The proof relies on new trigonometric interpolation constructions; the 
underlying ideas might be useful for other problems.

Finding an efficient algorithm that solves the same problem with a near-optimal 
worst-case performance for complex-valued signals
is still an open problem. Despite recent work that derives stability estimates 
for MUSIC and ESPRIT algorithms in certain cases, 
the question of how far are these algorithms from the optimal performance is 
not yet answered completely.

\section{Acknowledgment}
V. M. was supported by the Simons Foundation when developing the early ideas that led to this work.
The author is grateful to Emmanuel Candès for inspiring and useful discussions.

\appendices
\section{Proof of Lemma \ref{lem:dualcarlos2}}
\label{app:dp}
Let
\begin{equation}
    \fejer{}(t)\define 
    \left[\frac{\sin\lefto(\pi (\fc/2+1) t\right)}{\left(\fc/2+1\right)\sin(\pi t)}\right]^4
\end{equation}
and set
\begin{equation}
    q(t)\define\sum_{j=1}^{\npoints} \alpha_j \fejer{}(t-v_j)+\beta_j \fejer{}'(t-v_j),
    \label{eq:dpolydfn}
\end{equation}
where $\{\alpha_j\}_{j=1}^\npoints$ and $\{\beta_j\}_{j=1}^\npoints$ are free coefficients 
that will be determined in the following. Because $\fejer{}(\cdot)$ is frequency-limited 
to $[-\fc, \fc]$ [cf. \fref{eq:khi}, \fref{eq:fejsum}], $q(\cdot)$ satisfies \fref{prop:lf}.
Note,
\begin{equation}
    q'(t)=\sum_{j=1}^{\npoints} \alpha_j \fejer{}'(t-v_j)+\beta_j \fejer{}''(t-v_j).
\end{equation}
Define matrices $\matD_0,\matD_1,\matD_2\in\reals^{\npoints\times\npoints}$ with 
the elements
\begin{equation}
    [\matD_0]_{jk}=\fejer{}(v_j-v_k),\quad 
    [\matD_1]_{jk}=\fejer{}'(v_j-v_k),\quad 
    [\matD_2]_{jk}=\fejer{}''(v_j-v_k).
\end{equation}
To satisfy the interpolation constraints in \fref{prop:rhoint} we 
define $\vecalpha=\tp{[\alpha_1,\ldots,\alpha_\npoints]}$, 
$\vecbeta=\tp{[\beta_1,\ldots,\beta_\npoints]}$, $\fval=\tp{[f_1,\ldots,f_\npoints]}$, 
$\dval=\tp{[d_1,\ldots,d_\npoints]}$, demand
\begin{equation}
    \underbrace{
    \begin{bmatrix}
        \matD_0 & \matD_1\\
        \matD_1 & \matD_2
    \end{bmatrix}}_\matD
    \begin{bmatrix}
        \vecalpha\\
        \vecbeta
    \end{bmatrix}
    =
    \begin{bmatrix}
        \fval\\
        \dval
    \end{bmatrix}
    \label{eq:dpolydfn2}
\end{equation}
and solve for $\vecalpha$ and $\vecbeta$. It can be verified that $\matD_0$ 
and $\matD_2$ are both invertible; the corresponding Schur complements
\begin{align}
    \matE&=\matD_2-\matD_1 \matD_0^{-1} \matD_1\\
    \matF&=\matD_0-\matD_1 \matD_2^{-1} \matD_1
\end{align}
are well defined and are also both invertible 
(see~\cite[Sec.~2.3.1, pp.~925--926]{candes13},~\cite[App.~B, p.~1249]{candes21-2} 
for the relevant results). Therefore, $\matD$ is invertible and the 
inverse can be written as~\cite[Sec.~9.11.3.(2)]{lutkepohl96}
\begin{equation}
    \matD^{-1}=
    \begin{bmatrix}
        \matF^{-1} & -\matD_0^{-1} \matD_1 \matE^{-1}\\
        -\matD_2^{-1} \matD_1 \matF^{-1} & \matE^{-1}
    \end{bmatrix}.
\end{equation}
We know~(see~\cite[Sec.~2.3.1, pp.~925--926]{candes13},~\cite[App.~B, p.~1249]{candes21-2}):
\begin{align}
    \infnorm{\matD_1}
        &\le \frac{\ckersumd}{\lc},\label{eq:Dbnd1}\\
    \infnorm{\matD_0^{-1}}
        &\le \frac{1}{1-\ckersum}=\cDzinv,\label{eq:Dbnd0}\\
    \infnorm{\matE^{-1}}
        &\le \left(\frac{\pi^2\fc(\fc+4)}{3}-\left(\ckersumdd+\frac{\ckersumd^2}{1-\ckersum}\right)\fc^2\right)^{-1}
         \le \cEinv \lc^2,\label{eq:Ebnd}\\
    \infnorm{\matF^{-1}}
        &\le \cFinv, \label{eq:Fbnd}\\
    \infnorm{\matD_2^{-1}}
        &\le \cDtinv\lc^2. \label{eq:Dbnd2}
\end{align}
Above, $\infnorm{\matA}$ is the infinity norm of a matrix defined as
\begin{equation}
    \infnorm{\matA} = \max_{\infnorm{\vecy}=1} \infnorm{\matA\vecy} = \max_i \sum_j \abs{a_{ij}}
\end{equation}
and
\begin{alignat}{3}
    &\ckersum\le 0.007,\quad &&\ckersumd\le 0.08,\quad &&\ckersumdd\le 1.06,\\
    &\cDzinv \le 1.008,\quad  &&\cEinv\le 0.47,\quad  &&\cDtinv\le 0.43, \quad \cFinv\le 1.009.
\end{alignat}
Now we have
\begin{align}
    \infnorm{\vecalpha}
    &=\infnorm{\matF^{-1}\fval-\matD_0^{-1}\matD_1 \matE^{-1}\dval }
        \nonumber\\
    &\le\infnorm{\matF^{-1}\fval}+\infnorm{\matD_0^{-1}\matD_1 \matE^{-1}\dval }
        \nonumber\\
    &\le\infnorm{\matF^{-1}}\infnorm{\fval}+\infnorm{\matD_0^{-1}\matD_1 \matE^{-1}}\infnorm{\dval }
        \nonumber\\
    &\stackrel{(a)}{\le}
        \infnorm{\matF^{-1}}+\frac{1}{\lc}\infnorm{\matD_0^{-1}}\infnorm{\matD_1}\infnorm{\matE^{-1} }
        \nonumber\\
    &\stackrel{(b)}{\le} 
        \cFinv+\cDzinv\ckersumd\cEinv\define\ca.
\end{align}
Above, (a) follows because $\abs{f_j}\le 1$ and $\abs{d_j}\le 1/\lc$ 
for all $j=1,\ldots\npoints$; (b) follows by~\fref{eq:Dbnd1}, \fref{eq:Dbnd0}, \fref{eq:Ebnd}, and~\fref{eq:Fbnd}; 
and $\ca$ can be upper-bounded as $\ca\le 1.05$.
Similarly,
\begin{align}
    \infnorm{\vecbeta}
    &=\infnorm{-\matD_2^{-1}\matD_1 \matF^{-1} \fval+\matE^{-1}\dval }\nonumber\\
    &\le\infnorm{-\matD_2^{-1}\matD_1 \matF^{-1} \fval}+\infnorm{\matE^{-1}\dval}\nonumber\\
    &\le\infnorm{-\matD_2^{-1}\matD_1\matF^{-1} } 
        \infnorm{\fval}+\infnorm{\matE^{-1}}\infnorm{\dval}\nonumber\\
    &\le\infnorm{\matD_2^{-1}}
        \infnorm{\matD_1}\infnorm{\matF^{-1} } +\frac{1}{\lc} \infnorm{\matE^{-1}}\nonumber\\
    &\le \cb\lc
\end{align}
with $\cb\define\cDtinv\ckersumd \cFinv +\cEinv$ that can be upper-bounded as $\cb\le 0.51$.

The following lemma, proven in the end of this section, records bounds 
on $\sum_{j=1}^{\npoints} \abs{ \fejer{}^{(l)}(t-v_j)}$, $l=0,1$.
\begin{lem}
    \label{lem:sumbnd}
    The following estimates hold:
    \begin{align}
        &\sum_{j=1}^{\npoints} \abs{ \fejer{}(t-v_j)}\le \cgsum,\\
        &\sum_{j=1}^{\npoints} \abs{\fejer{}'(t-v_j)}\le \cgsumd/\lc,
    \end{align}
    where $\cgsum$, $\cgsumd$ are positive numerical constants defined 
    in the proof of the lemma below.
\end{lem}
Using the bounds we obtain the required estimates as follows.
Observe,
\begin{align}
    \abs{q(t)}
    &=\abs{\sum_{j=1}^{\npoints} \alpha_j \fejer{}(t-v_j)+\beta_j \fejer{}'(t-v_j)}\nonumber\\
    &\le\infnorm{\vecalpha}\underbrace{\sum_{j=1}^{\npoints} 
        \abs{ \fejer{}(t-v_j)}}_{\le \cgsum}
        +\infnorm{\vecbeta}\underbrace{\sum_{j=1}^{\npoints} \abs{\fejer{}'(t-v_j)}}_{\le \cgsumd/\lc}
        \nonumber\\
    &\le \ca\cgsum +\cb \cgsumd\define \cqdv.\label{eq:quubnd}
\end{align}
This proves \fref{prop:dqinfbnd}.
\fref{prop:qpbnd1} in the lemma follows from~\fref{eq:quubnd} by \fref{eq:bernstein} [Bernstein 
theorem], using that $q'(t)$ is also a trigonometric polynomial frequency-limited to $\fc$:
\begin{equation}
\abs{q'(t)} \le \cqdvd \fc,\quad \cqdvd\define 2\pi \cqdv.
\label{eq:quubnd2}
\end{equation}
In turn, \fref{prop:qppbnd1} follows from~\fref{eq:quubnd2} by \fref{eq:bernstein} [Bernstein 
theorem], using that $q''(t)$ is also a trigonometric polynomial frequency-limited to $\fc$:
\begin{equation}
\abs{q''(t)} \le \cqdvdd \fc^2,\quad \cqdvdd\define 4\pi^2 \cqdv.
\end{equation}
This completes the proof of \fref{lem:dualcarlos2}. \qed

\begin{proof}[Proof of Lemma~\ref{lem:sumbnd}]
For all $t\in[-1/2,1/2]$, we have the following bounds~\cite[Sec.~2.3.2, p.~928]{candes13}:
\begin{align}
&\abs{\fejer{}(t)}\le 1,\label{eq:kbnd1}\\
&\abs{\fejer{}'(t)}\le \frac{\pi^2}{3}\fc(\fc+4)\abs{t}\label{eq:kbnd2}.
\end{align}

For all $t$ with $\lc/2\le \abs{t}\le 1/2$ and $l=0,1$, by inspection it 
follows from~\cite[Lm.~2.6]{candes13} that the following bound holds:
\begin{equation}
\abs{\fejer{}^{(l)}(t)}\le \frac{\pi^l \cg^l}{(\fc+2)^{4-l}t^4}\label{eq:kbnd5}
\end{equation}
with $\cgz\define 1, \cgd\define 6$.
[To obtain this result from~\cite[Lm.~2.6]{candes13}, observe that, in the terminology of \cite{candes13}, for all $t$ 
with $\lc/2\le \abs{t}\le \sqrt{2}/\pi$, $b(t)<2 a(t)$ and $a(t)<1$.]

Define $u_{k_j} \define t - v_j$,  ordered in such a way that $\abs{u_1}<\ldots<\abs{u_\npoints}$. 
Since $\{v_1, v_1, \ldots, v_\npoints\}\in\rset{\kappa\lc}{1}$, we have
\begin{equation}
    \label{eq:ubnd}
    \abs{u_{2j}} > \lc \kappa (j-1)\ 
    \text{and}\ \abs{u_{2j-1}} > \lc\kappa (j-1), \ \text{for}\ j \ge 2,
\end{equation}
and also,
\begin{equation}
    \label{eq:ubnd1}
    \abs{u_{j}} > \frac{\lc}{2} \kappa j,\ \text{for}\ j > 1.
\end{equation}

First,
\begin{align}
    \sum_{j=1}^{\npoints} \abs{\fejer{}(t-v_j)}
    =\sum_{j=1}^{\npoints} \abs{\fejer{}(u_j)}
    &\stackrel{(a)}{\le} 
        2+\frac{1}{(\fc+2)^4} \sum_{j=3}^{\npoints} \frac{1}{u_j^4}\\
    &\stackrel{(b)}{\le} 
        2+\frac{2}{(\fc+2)^4} \sum_{j=1}^{\lceil\npoints/2\rceil} \frac{1}{(\lc \kappa j)^4}\\
    &\le 
        2+\frac{1}{\fc^4} \frac{2}{(\lc \kappa)^4}\sum_{j=1}^{\infty} \frac{1}{j^4}
    \stackrel{(c)}{\le} 
        2+\frac{2}{\kappa^4} \frac{4}{3}
    \stackrel{(d)}{\le}  \cgsum.
\end{align}
Above, in 
(a) we used~\fref{eq:kbnd1} to bound the terms for $j=1,2$ and used~\fref{eq:kbnd5} 
to bounds the terms for $j>2$ [\fref{eq:kbnd5} is applicable because $u_3>\lc\kappa/2>\lc/2$ since $\kappa=1.87$ ]; 
in (b) we used~\fref{eq:ubnd}; in (c) we used
\begin{equation}
\sum_{j=1}^{\infty} \frac{1}{j^4}=\frac{\pi^2}{90}\le \frac{4}{3};
\end{equation}
and in (d) we defined $\cgsum\define 2.22$.

Second,
\begin{align}
    \sum_{j=1}^{\npoints} \abs{\fejer{}'(t-v_j)}
    &=\abs{\fejer{}'(u_1)} + \sum_{j=2}^{\npoints} \abs{\fejer{}'(u_j)}\\
    &\stackrel{(a)}{\le} 
        \frac{\pi^2}{3}\fc(\fc+4)
        \frac{\lc}{2}\kappa + \frac{1}{(\fc+2)^3}
        \frac{6\pi}{((\lc/2) \kappa)^4} + \sum_{j=2}^{\npoints}
        \abs{\fejer{}'(u_j)} \\
    &\stackrel{(b)}{\le} 
        \frac{\pi^2}{3}\fc(\fc+4)
        \frac{\lc}{2}\kappa + \frac{1}{(\fc+2)^3}
        \frac{6\pi}{((\lc/2) \kappa)^4} + \frac{1}{(\fc+2)^3} \sum_{j=2}^{\npoints}
        \frac{6\pi}{u_j^4}\\
    &\stackrel{(c)}{\le}
        \frac{\pi^2}{3}\fc(\fc+4)\frac{\lc}{2}\kappa 
        + \frac{1}{(\fc+2)^3} \sum_{j=1}^{\npoints} 
            \frac{6\pi}{((\lc/2) \kappa j)^4}\\
    &\le \frac{\pi^2}{6}(\fc+4)\kappa
        +\frac{6\pi}{\fc^3} 
            \frac{2^4}{(\lc \kappa)^4}
            \sum_{j=1}^{\infty} \frac{1}{j^4}\\
    &\stackrel{(d)}{\le} \frac{\pi^2}{6}(\fc+4)\kappa+35 \fc
    \le 12.4+38.1 \fc
    \stackrel{(e)}{\le}  \cgsumd \fc.
\end{align}
Above, in 
(a) we used that if $\abs{u_1} < \lc\kappa/2$, then 
$\abs{\fejer{}'(u_1)}<(\pi^2/3)\fc(\fc+4)\lc\kappa/2$ by $\fref{eq:kbnd2}$ 
and otherwise $\abs{\fejer{}'(u_1)}<6\pi/[(\fc+2)^3((\lc/2) \kappa )^4]$ 
by~\fref{eq:kbnd5}; 
in (b) we used~\fref{eq:kbnd5}, which is applicable because $\abs{u_j}>\lc\kappa>\lc/2$ by~\fref{eq:ubnd1}; 
in (c) we used~\fref{eq:ubnd1}; 
in (d) we used $\sum_{j=1}^{\infty} 1/j^4<4/3$;
in (e) we used that $\fc > 128$ and defined $\cgsumd\define 38.2$.
\end{proof}

\section{Proof of Lemma~\ref{lem:dualprod1d}}
\label{app:dualprod1d}
\subsection{Construction}
\label{sec:const2}
We first describe how the trigonometric polynomial $q_2(\cdot)$ is constructed. 
In Sections \ref{sec:existencephipk}--\ref{sec:proofdualprod14}, we prove that the construction is valid and that it satisfies the required 
Properties \ref{prop:q1prop1d}--\ref{prop:q1prop4d}.

Recall, $\setT=\{t_1,\ldots, t_\sparsity\}$ is defined in~\fref{eq:setT} and, as before, define $\setT_k$, $k=1,\ldots,\rr$, as in~\fref{eq:setTk}; remember that $\setT = \setT_1\cup\ldots \cup\setT_r$ and $\setT_k\in\rset{\kappa\llo\rr}{1}$.

We will construct the trigonometric polynomial $q_2(\cdot)$ as a 
(shifted) sum of $\rr$ trigonometric polynomials $\{\phi_k(\cdot)\}_{k=1}^\rr$:
\begin{equation}
    \label{eq:q1defmdd}
    q_2(t)=\sum_{k=1}^\rr\phi_k(t)-\rho.
\end{equation}
Note that we are overloading the notations here and $\phi_k(\cdot)$ in this 
sections are \emph{different} from $\phi_k(\cdot)$ in \fref{sec:dualprod1m}.
Each of the trigonometric polynomials $\{\phi_k(\cdot)\}_{k=1}^\rr$ 
is frequency-limited to $\flo$,
\begin{equation}
    \label{eq:phifreqconstr1}
    \phi_k(t)
    =\sum_{l=-\flo}^{\flo} \hat \phi_{k,l} e^{-\iu 2\pi l t} \quad 
    \text{for some}\quad \hat \phi_{k, l}\in\complexset
\end{equation}
and is constructed separately to satisfy the following interpolation 
constraints on $\setT$:
\begin{align}
    \phi_k(t_l)&=
    \begin{cases}
    \rho, &\text{ if } t_l\in\setT_k,\\
    0,    &\text{ if } t_l\in \setT_k^c\define \setT\setminus\setT_k,
    \end{cases}\label{eq:req1mdd}\\
    \phi'_k(t_l)&=\begin{cases}
    \gamma s'_l, &\text{ if } t_l\in\setT_k,\\
    0, &\text{ if } t_l\in\setT_k^c.
    \end{cases}\label{eq:req3mdd}
\end{align}

Constraints~\fref{eq:req1mdd}, \fref{eq:req3mdd}, and definition~\fref{eq:q1defmdd} 
guarantee, for all $l=1,\ldots,\sparsity$,
\begin{align}
    q_2(t_l)&=0,\label{eq:q1int1dd}\\
    q'_2(t_l)&=\gamma s'_l\label{eq:q1int2dd}.
\end{align}

To develop intuition about our construction, observe that~\fref{eq:phifreqconstr1} 
and~\fref{eq:q1defmdd} guarantee that \fref{prop:q1prop1d} is satisfied.
Further, observe that interpolation constraints~\fref{eq:q1int1dd} and~\fref{eq:q1int2dd} 
are needed for~\fref{eq:bndneard} to hold because $q_0(t)=q'_0(t)=0$ for all $t\in\setT$.

Next, we explain how to construct the trigonometric polynomials $\phi_k(\cdot)$, $k=1,\ldots,\rr$.
The idea is to construct $\phi_k(\cdot)$ as a product of two trigonometric polynomials:
\begin{equation}
    \label{eq:phikdfndd}
    \phi_k(t)\define\phi_{0,k}(t)\times\phi_{+,k}(t).
\end{equation}
The first term in the product is defined as
\begin{equation}
    \label{eq:phizkdef}
    \phi_{0,k}(t)\define \prod_{l\ne k} q_{\rr\llo,\setT_l}(t),
\end{equation}
where $q_{\rr\llo,\setT_l}(\cdot),\ l=1,\ldots,\rr$ are the trigonometric polynomials 
constructed via \fref{lem:dualcarlos} with $\lc=\rr\llo$ and $\setV=\setT_l\in \rset{\kappa\llo\rr}{1}$.
The second term,
\begin{equation}
    \label{eq:phipdefdd}    
    \phi_{+,k}(t)\define\rr^{2\rr}\cufh^{\rr}q_{\rr\llo,\setT_k,\{f_j\},\{d_j\}}(t)
\end{equation}
is a (rescaled) trigonometric polynomial $q_{\rr\llo,\setT_k,\{f_j\},\{d_j\}}(\cdot)$ constructed 
via \fref{lem:dualcarlos2}  with $\lc=\rr\llo$, and $\setV=\setT_k\in \rset{\kappa\llo\rr}{1}$
and $\cufh$ is a positive numerical constant defined in~\fref{eq:czdduf1d}  below.
Further, the function-values and derivatives of $q_{\rr\llo,\setT_k,\{f_j\},\{d_j\}}(\cdot)$ are constrained 
on $\setT_k$ so that $\phi_{+,k}(\cdot)$ satisfies the following:
\begin{alignat}{2}
    \phi_{+,k}(t)&=\rho\frac{1}{\phi_{0,k}(t)}, &&\text{ for all } t\in \setT_k,\label{eq:phi11mdd}\\
    \phi'_{+,k}(t_l)
    &=-\rho\frac{\phi'_{0,k}(t_l)}{\phi^2_{0,k}(t_l)}
        +\frac{\gamma s'_l}{\phi_{0,k}(t_l)}, &&\text{ for all } t_l\in \setT_k.
        \label{eq:phi11m1dd}
\end{alignat}

Observe that $\phi_{0,k}(\cdot)$ in this section is \emph{identical} to $\phi_{0,k}(\cdot)$ 
defined in \fref{sec:dualprod1m} and, therefore, satisfies all the properties 
derived in \fref{sec:propphizk}; $\phi_{+,k}(\cdot)$ in this section is \emph{different} 
from $\phi_{+,k}(\cdot)$ in \fref{sec:dualprod1m} and the notation is overloaded. 

We will prove in \fref{sec:existencephipk} below, that this specification is valid, 
in the sense that the corresponding function values and derivatives of~$q_{\rr\llo,\setT_k,\{f_j\},\{d_j\}}(\cdot)$ 
on $\setT_k$ satisfy requirements~\fref{eq:condp2} of \fref{lem:dualcarlos2}.

It follows from~\fref{eq:phikdfndd}, \fref{eq:phizkdef}, \fref{eq:phi11mdd}, \fref{lem:dualcarlos}, 
Properties \ref{prop:qzero}, \ref{prop:qnearbnd}, and \ref{prop:qinfbnd},
that the interpolation constraint~\fref{eq:req1mdd} is satisfied:%
\begin{alignat}{2}
    \phi_k(t)&=\phi_{0,k}(t) \phi_{+,k}(t)=\rho,  &&\text{ for all } t\in\setT_k,\\
    \phi_k(t)&=\underbrace{\phi_{0,k}(t)}_0 \phi_{+,k}(t)=0, &&\text{ for all } t\in\setT_k^c.
\end{alignat}
Next, by~\fref{eq:phikdfndd},
\begin{equation}
    \phi'_k(t)=\phi'_{0,k}(t)\phi_{+,k}(t)+\phi_{0,k}(t)\phi'_{+,k}(t).
\end{equation}
Therefore, by~\fref{eq:phi11mdd} and~\fref{eq:phi11m1dd},
\begin{align}
    \phi'_k(t_l)
    &=\phi'_{0,k}(t_l)\phi_{+,k}(t_l)+\phi_{0,k}(t_l)\phi'_{+,k}(t_l)\nonumber\\
    &=\phi'_{0,k}(t_l)\rho\frac{1}{\phi_{0,k}(t_l)}
        +\phi_{0,k}(t_l)\left(-\rho\frac{\phi'_{0,k}(t_l)}{\phi^2_{0,k}(t_l)}
        +\frac{\gamma s'_l}{\phi_{0,k}(t_l)}\right)\nonumber\\
    &=\gamma s'_l,
\end{align}
for all $t_l\in\setT_k$. Further, 
by \fref{eq:phizkdef}, \fref{lem:dualcarlos}, 
Property \ref{prop:qzero},
\begin{equation}
    \phi'_k(t)
    =\underbrace{\phi'_{0,k}(t)}_0 \phi_{+,k}(t)
    +\underbrace{\phi_{0,k}(t)}_0\phi'_{+,k}(t)=0,
\end{equation}
for all $t\in\setT_k^c$.
We conclude that the interpolation constraint~\fref{eq:req3mdd} is satisfied.

Finally,~\fref{eq:phifreqconstr1} follows from~\fref{eq:phikdfndd} 
because $\phi_{0,k}(\cdot)$ in~\fref{eq:phizkdef} is frequency-limited 
to $(\rr-1)/(\llo\rr)$ [\fref{lem:dualcarlos}, \fref{prop:lf}] 
and $\phi_{+,k}(\cdot)$ in~\fref{eq:phipdefdd} 
is frequency-limited to $1/(\llo\rr)$ [\fref{lem:dualcarlos2}, \fref{prop:lf2}] so that $\phi_k(\cdot)$ is 
frequency-limited to 
    $(\rr-1)/(\llo\rr)+1/(\llo\rr)=1/\llo=\flo.$
Therefore, by~\fref{eq:q1defmdd}, $q_2(\cdot)$ is also frequency-limited to $\flo$, 
which proves~\fref{prop:q1prop1d}.

\subsection{Existence of $\phi_{+,k}(\cdot)$}
\label{sec:existencephipk}
In this subsection, we check that trigonometric polynomial $\phi_{+,k}(\cdot)$ that 
satisfies~\fref{eq:phi11mdd} 
and~\fref{eq:phi11m1dd} can indeed be defined according to~\fref{eq:phipdefdd} 
with $q_{\rr\llo,\setT_k,\{f_j\},\{d_j\}}(\cdot)$ constructed via \fref{lem:dualcarlos2}
with $\lc=\rr\llo$ and $\setV=\setT_k\in \rset{\kappa\llo\rr}{1}$.
To this end, we need to show 
that the constraints on the function values $\{f_j\}$ and on the derivatives $\{d_j\}$
that are implied by the constraints~\fref{eq:phi11mdd} and~\fref{eq:phi11m1dd} satisfy 
requirements~\fref{eq:condp2} of \fref{lem:dualcarlos2}.

First consider the case $r=1$. As already discussed, in this case $\phi_{0,k}(t)=1$ for all
$t$, and, therefore, $\phi'_{0,k}(t)=0$ for all $t$. Plugging these values into~\fref{eq:phi11mdd} 
and~\fref{eq:phi11m1dd} we see from~\fref{eq:phipdefdd} that the requirements~\fref{eq:condp2} of \fref{lem:dualcarlos2}
are satisfied.

Next, consider the case $r>1$.

To check that requirements~\fref{eq:condp2} are satisfied for $t\in\setT_k$, 
we need to find upper bounds on $\abs{\phi_{+,k}(t)}$ and on $\abs{\phi_{+,k}'(t)}$.

Take $t\in\setT_k$ and observe:
\begin{align}
    \abs{\phi_{+,k}(t)}
    \stackrel{(a)}{=}
        \abs{\rho\frac{1}{\phi_{0,k}(t)}}
    \stackrel{(b)}{\le}
        \frac{\lhi^{2\rr}}{\llo^{2\rr}} 
        \frac{1}{\frac{\lhi^{2(\rr-1)}}{(\rr\llo)^{2(\rr-1)}}} 
        \frac{1}{\cql^{\rr-1}}
    &\stackrel{(c)}{\le} 
        \rr^{2(\rr-1)}\frac{1}{\cql^{\rr}}\frac{\lhi^{2}}{\llo^{2}} \label{eq:ph11bdnm0d}\\
    &\le \rr^{2\rr}\frac{1}{\cql^{\rr}}.\label{eq:lmtransform2dd}
\end{align}
Above, (a) follows by~\fref{eq:phi11mdd}; (b) follows by~\fref{eq:phizkbndfar},
which is valid because $t\in\setT_k$ implies $t\in\far{\lhi}{\setT_k^c}$; (c) follows because $\cql<1$.

Next, take $t\in\setT_k$ and observe, according to~\fref{eq:phi11m1dd},
\begin{equation}
    \label{eq:phipderivdd}
    \abs{\phi_{+,k}'(t)}
    \le\abs{\rho\frac{\phi'_{0,k}(t)}{\phi^2_{0,k}(t)}}+\abs{\frac{\gamma}{\phi_{0,k}(t)}}.
\end{equation}
The first term above can be upper-bounded following exactly the same steps 
that lead from~\fref{eq:phipderiv} to~\fref{eq:bndphikdnear1f}. This gives:
\begin{equation}
    \label{eq:bbb0}
    \abs{\rho\frac{\phi'_{0,k}(t)}{\phi^2_{0,k}(t)}}
    \le \rr^{2\rr-1}\czdduuuu^\rr\frac{\lhi}{\llo^2}.
\end{equation}
To upper-bound the second term in~\fref{eq:phipderivdd}, consider two cases.

Case 1: $t\in\far{\rr\D\llo}{\setT_k^c}$. Then, 
by~\fref{eq:phi0klb}, $\phi_{0,k}(t) \ge \cqlf^{\rr-1}$ and, therefore,
\begin{equation}
    \label{eq:bbb1}
    \abs{\frac{\gamma}{\phi_{0,k}(t)}}
    \le \frac{\gamma}{\cqlf^{\rr-1}} 
    = \frac{1}{\cqlf^{\rr-1}} \frac{\lhi^{2\rr-1}}{\llo^{2\rr}}
    \le \rr^{2\rr-1} \left(\frac{1}{\cqlf^2}\right)^{\rr} \frac{\lhi}{\llo^2}.
\end{equation}
Above, in the last (crude) inequality we used that $\cqlf<1$ and that $\lhi/\llo<1$.

Case 2: $t\in\near{\rr\D\llo}{\setT_k^c}$. In this case 
set $\{v_1,\ldots,v_{\hrr}\}\define\setT_k^c\intersect\near{\rr\D\llo}{t}$ 
and note that $1\le \hrr\le \rr-1$. Hence,  by~\fref{eq:qzlppl},
\begin{equation}
    \phi_{0,k}(t)
    \ge \cz^{\rr-1}\frac{ \prod_{j=1}^{\hrr}(v_j-t)^2}{(\rr\llo)^{2\hrr}}
    \ge \cz^{\rr-1}\frac{ \lhi^{2\hrr}}{(\rr\llo)^{2\hrr}},
\end{equation}
where we used that $\abs{v_{j}-t}\ge 2\lhi$ for all $j=1,\ldots,\hrr$ because
all elements of $\setT$ are separated by at least $2\lhi$. Therefore,
\begin{equation}
    \label{eq:bbb2}
    \abs{\frac{\gamma}{\phi_{0,k}(t)}}
    \le \frac{\gamma}{\cz^{\rr-1}}\frac{(\rr\llo)^{2\hrr}}{\lhi^{2\hrr}} 
    = 
        \frac{1}{\cz^{\rr-1}}
        \frac{\lhi^{2\rr-1}}{\llo^{2\rr}}
        \frac{(\rr\llo)^{2\hrr}}{\lhi^{2\hrr}}
    \le \rr^{2\rr-1} \left(\frac{1}{\cz^2}\right)^{\rr} 
        \frac{\lhi}{\llo^2}.
\end{equation}
Above, in the last (crude) inequality we used that $\cz<1$, $\lhi/\llo<1$, and $\hrr\le \rr-1$.

Plugging~\fref{eq:bbb0}, \fref{eq:bbb1}, and~\fref{eq:bbb2} into~\fref{eq:phipderivdd} we obtain
\begin{align}
    \abs{\phi_{+,k}'(t)}
    &\le\rr^{2\rr-1} \cufffff^{\rr} \frac{\lhi}{\llo^2}\label{eq:phipderivdd1p}\\
    &\le\rr^{2\rr-1} \cufffff^{\rr} \frac{1}{\llo},\label{eq:phipderivdd1}
\end{align}
where we used that $\lhi/\llo<1$ and defined $\cufffff \define 2\max(\czdduuuu, 1/\cqlf^2, 1/\cz^2)$.

It follows from~\fref{eq:lmtransform2dd} and~\fref{eq:phipderivdd1} that the function values and the derivatives of $q_{\rr\llo,\setT_k}(t)=\phi_{+,k}(t)/(\rr^{2\rr}\cufh^{\rr})$ with 
\begin{equation}
    \cufh\define\max(\cufffff, 1/\cql)
    \label{eq:czdduf1d}
\end{equation} 
satisfy  requirements~\fref{eq:condp2} of \fref{lem:dualcarlos2} on $\setT_k$. 
We conclude that $\phi_{+,k}(\cdot)$ can indeed be defined according 
to~\fref{eq:phipdefdd}. According to Properties~\ref{prop:dqinfbnd}, \ref{prop:qpbnd1}, and \ref{prop:qppbnd1}
of \fref{lem:dualcarlos2}, and~\fref{eq:phipdefdd}, $\phi_{+,k}(\cdot)$ satisfies the following properties:
\begin{align}
    &\infnorm{\phi_{+,k}(\cdot)}\le \rr^{2\rr}\cufh^{\rr}\cqdv,\label{eq:phi11ddm00dd}\\
    &\infnorm{\phi_{+,k}'(\cdot)}\le \rr^{2\rr-1}\cufh^{\rr} \cqdvd\frac{1}{\llo},\label{eq:phi11ddm0dd}\\
    &\infnorm{\phi_{+,k}''(\cdot)}\le \rr^{2\rr-2}\cufh^{\rr} \cqdvdd\frac{1}{\llo^2}.\label{eq:phi11ddmdd}
\end{align}

\subsection{Proof of property 2}
Take $j\in\{1,\ldots, \sparsity\}$ and consider $t_j\in\setT$. 
There exists a unique $l\in\{1,\ldots,\rr\}$ such that $t_j\in\setT_l$. We will 
show that for all $\tau\in\near{\lhi}{t_j}$
\begin{equation}
    \abs{\phi_l(\tau)-\rho-\gamma s'_j(\tau-t_j)}
    \le \rr^{2\rr+3}\cufx^{\rr+1}  q_0(\tau) 
    \label{eq:toprove2d}
\end{equation}
and
\begin{equation}
    \abs{\phi_k(\tau)}
    \le 
    \rr^{2\rr+3}\cugg^{\rr+1}  q_0(\tau),\ 
    \text{for}\ k\in\{1,\ldots, \rr\},\ k\ne l, \label{eq:toprove1d}
\end{equation}
where the positive numerical constants $\cufx$ and $\cugg$ are defined below.

From this we will conclude:
\begin{align}
    \abs{q_2(\tau)-\gamma s'_j(\tau-t_j)}
    &=\abs{\sum_{k=1}^\rr \phi_k(\tau) -\rho-\gamma s'_j(\tau-t_j)}\\
    &\le\sum_{k\ne l}^\rr \abs{\phi_k(\tau)}+ \abs{\phi_l(\tau) - \rho-\gamma s'_j(\tau-t_j)}\\
    &\le\rr^{2\rr+4}\cufk^{\rr+1}  q_0(\tau),
\end{align}
with $\cufk \define 2\max(\cufx, \cugg)$, as desired.

To prove~\fref{eq:toprove2d} and~\fref{eq:toprove1d} recall, by~\fref{eq:req1mdd} and~\fref{eq:req3mdd},
\begin{align}
    &\left.\phi_l(\tau)-\rho -\gamma s'_j(\tau- t_j) \right|_{\tau= t_j}
        =  0 = q_0(t_j),\label{eq:init1d}\\
    &\left.\frac{d(\phi_l(\tau)-\rho-\gamma s'_j(\tau- t_j))}{d \tau}\right|_{\tau= t_j} 
        =\phi'_l(t_j)-\gamma s'_j=0=q'_0(t_j),\label{eq:q1int4dd}\\
    &\phi_k(t_j) =  0 = q_0(t_j)\ \text{for}\ k\in\{1,\ldots, \rr\},\ k\ne l,
        \label{eq:init2d}\\
    &\phi'_k(t_j)=0=q'_0(t_j)\ \text{for}\ k\in\{1,\ldots, \rr\},\ k\ne l.
        \label{eq:init3d}
\end{align}
Hence, in order to prove the bounds in~\fref{eq:toprove1d} and~\fref{eq:toprove2d}, we will derive 
upper bounds on the second derivatives $\abs{\phi''_k(\tau)}$, $k\in\{1,\ldots,\rr\}$, 
valid for all $\tau\in\near{\lhi}{t_j}$, and use the Mean Value 
theorem (see \fref{thm:meanvalue}). 

Taking the second derivative of~\fref{eq:phikdfndd} and applying the triangle inequality we find:
\begin{equation}
    \abs{\phi''_{k}(\tau)}
    \le
        \underbrace{\abs{\phi''_{0,k}(\tau)}\abs{\phi_{+,k}(\tau)}}_{E_{1}(\tau)}
        +2\underbrace{\abs{\phi'_{0,k}(\tau)}\abs{\phi'_{+,k}(\tau)}}_{E_{2}(\tau)}
        +\underbrace{\abs{\phi_{0,k}(\tau)}\abs{\phi''_{+,k}(\tau)}}_{E_3(\tau)}.
    \label{eq:phiddsumtermsl1d}
\end{equation}
In the derivation below we upper-bound the terms separately.

We will need the following notations. 
Set $\{\vt_1,\ldots,\vt_{\hrr}\}\define \near{\rr\D\llo}{\tau}\intersect\setT_k^c$. 
Also, set $\{v_1,\ldots,v_{\tilde\rr}\}\define \near{\rr\D\llo-\lhi}{t_j}\intersect\setT_k^c$. 
Note that the set $\{v_1,\ldots,v_{\tilde\rr}\}$ does not depend on $\tau$ and 
also $\{v_1,\ldots,v_{\tilde\rr}\}\subset \{\vt_1,\ldots,\vt_{\hrr}\}$ so 
that $\tilde\rr\le \hrr$. 

The remainder of the proof of Property 2 is organized as follows. First, consider 
the case $t_j\in \setT_k$ and prove~\fref{eq:toprove2d}, next consider the 
case $t_j\in\setT_k^c$ and prove~\fref{eq:toprove1d}.

\paragraph{Proof of \fref{eq:toprove2d}: case~$t_j\in\setT_k$.}
\subparagraph{Bounding $E_1(\tau)$.} By \fref{eq:mvt2} [Mean Value theorem] and the triangle inequality we can write
\begin{equation}
    \abs{\phi_{+,k}(\tau)}
    \le 
        \abs{\phi_{+,k}(t_j)}
        +\abs{\phi_{+,k}'(t_j)}\abs{\tau- t_j}
        +\frac{1}{2}\abs{\phi_{+,k}''(\tau_m)}(\tau- t_j)^2
\end{equation}
with $\tau_m\in (t_j, \tau)$. Next, we use~\fref{eq:ph11bdnm0d} 
to upper-bound $\abs{\phi_{+,k}(t_j)}$; use~\fref{eq:phipderivdd1p} 
to upper-bound $\abs{\phi_{+,k}'(t_j)}$; use~\fref{eq:phi11ddmdd} 
to upper-bound $\abs{\phi_{+,k}''(\tau_m)}$. With these estimates 
we can further upper-bound $\abs{\phi_{+,k}(\tau)}$ as follows:
\begin{align}
    \abs{\phi_{+,k}(\tau)}
    \le\rr^{2\rr-1}\cufl^\rr \left(\frac{\lhi^2}{\llo^2} +\frac{\lhi}{\llo^2}\abs{\tau-t_j}+\frac{1}{\llo^2}(\tau-t_j)^2\right)
    \le \rr^{2\rr-1}\cufm^\rr \frac{\lhi^2}{\llo^2}.
    \label{eq:phipbnd1d}
\end{align}
Above, we defined $\cufl \define \max(1/\cql, \cufffff, \cufh\cqdvdd)$, $\cufm\define 3 \cufl$, 
and used $\abs{\tau-t_j}\le\lhi$.

Assume $\tilde \rr\ge 1$ (the case  $\tilde \rr=0$ will be treated 
separately below) so that $\hrr\ge 1$ and $\tau\in\near{\rr\D\llo}{\setT_k^c}$, 
which implies that $\abs{\phi_{0,k}''(\tau)}$ can be upper-bounded by~\fref{eq:qddub}.

Multiplying~\fref{eq:phipbnd1d} and~\fref{eq:qddub} and simplifying we obtain 
the following upper bound on $E_1$:
\begin{align}
    E_1(\tau)
    =\abs{\phi''_{0,k}(\tau)}\abs{\phi_{+,k}(\tau)}
    &\stackrel{(a)}{\le}
        \rr^{2\rr+1}\cufp^{\rr+1} 
        \frac{\prod_{1\le l\le \hrr} (\vt_l-\tau)^2}{(\rr\llo)^{2\hrr}}
        \frac{1}{\llo^2}\nonumber\\
    &\stackrel{(b)}{\le}
        \rr^{2\rr+1}\cufp^{\rr+1} 
        \frac{\prod_{1\le l\le \tilde\rr} (v_l-\tau)^2}{(\rr\llo)^{2\tilde\rr}}
        \frac{1}{\llo^2}.
    \label{eq:E1d}
\end{align}
Above, in 
(a) we  used (multiple times) the bound $\lhi\le \abs{\vt_l-\tau}$,
which is true for all $l\in\{1,\ldots,\hrr\}$ (follows because the 
elements of $\setT$ are separated 
by at least $2\lhi$), used $\lhi/\llo<1$, and defined 
$\cufp \define  \max(6 \czdu \cufm, \cufh\cqdv \czdduu)$;
in (b) we used the fact that $\abs{\vt_l-\tau}/(\llo \rr)\le \D<1$ for
all $l\in\{1,\ldots,\hrr\}$. 

For the case $\tilde r=0$, the upper bound~\fref{eq:E1} 
also holds by~\fref{eq:phi0kddub}
and~\fref{eq:phi11ddm00dd}.

\subparagraph{Bounding $E_2(\tau)$.}
By~\fref{eq:mvt1} [Mean Value theorem] we can write
\begin{equation}
    \abs{\phi'_{+,k}(\tau)}
    \le \abs{\phi'_{+,k}(t_j)}+\abs{\phi_{+,k}''(\tau_m)}\abs{\tau- t_j}
\end{equation}
with $\tau_m\in (t_j, \tau)$. Next, we use~\fref{eq:phipderivdd1p} 
to upper-bound $\abs{\phi_{+,k}'(t_j)}$; use~\fref{eq:phi11ddmdd} 
to upper-bound $\abs{\phi_{+,k}''(\tau_m)}$. 
With these estimates we can further upper-bound $\abs{\phi'_{+,k}(\tau)}$ as follows:
\begin{align}
    \abs{\phi'_{+,k}(\tau)}
    \le\rr^{2\rr-1}\cufq^\rr 
        \left(\frac{\lhi}{\llo^2}+\frac{1}{\llo^2}\abs{\tau-t_j}\right)
    \le \rr^{2\rr-1}\cufr^\rr \frac{\lhi}{\llo^2}.
    \label{eq:phipkdubd}
\end{align}
Above, we defined $\cufq \define \max(\cufffff, \cqdvdd\cufh)$, $\cufr\define 2 \cufq$, 
and used $\abs{\tau-t_j}\le\lhi$.

Assume $\tilde \rr\ge 1$ (the case  $\tilde \rr=0$ will be treated 
separately below) so that $\hrr\ge 1$ and $\tau\in\near{\rr\D\llo}{\setT_k^c}$, 
which implies that $\abs{\phi'_{0,k}(t)}$ can be upper-bounded by~\fref{eq:qddub2}.
Multiplying~\fref{eq:phipkdubd} and~\fref{eq:qddub2} and simplifying, we 
obtain the following upper bound on $E_2$:

\begin{align}
    E_2(\tau)
    =\abs{\phi'_{0,k}(\tau)}\abs{\phi'_{+,k}(\tau)}
    &\stackrel{(a)}{\le} 
        \rr^{2\rr}\cufs^{\rr}  
        \frac{\prod_{1\le l\le \hrr} (\vt_l-\tau)^2}{(\rr\llo)^{2\hrr}}
        \frac{1}{\llo^2}\nonumber\\
    &\stackrel{(b)}{\le}
        \rr^{2\rr}\cufs^{\rr}  
        \frac{\prod_{1\le l\le \tilde\rr} (v_l-\tau)^2}{(\rr\llo)^{2\tilde\rr}}
        \frac{1}{\llo^2}.
    \label{eq:E2d}
\end{align}
Above, 
in (a) we  used the bound $\lhi\le \abs{\vt_l-\tau}$, which is true 
for all $l\in\{1,\ldots,\hrr\}$ (follows because the elements of $\setT$ are separated 
by at least $2\lhi$), used $\lhi/\llo<1$, and defined $\cufs \define \max(2 \cufr \czdu, 2\pi\cufh \cqdvd)$; 
in (b) we use the fact that $\abs{\vt_l-\tau}/(\llo \rr)\le \D<1$ for all $l\in\{1,\ldots,\hrr\}$.

For the case $\tilde r=0$, the upper bound~\fref{eq:E2d} also holds 
by~\fref{eq:phi0kdub}
and~\fref{eq:phi11ddm0dd}.

\subparagraph{Bounding $E_3(\tau)$.}
By~\fref{eq:phi11ddmdd},
\begin{equation}
    \label{eq:phipe3d}
    \abs{\phi_{+,k}''(\tau)}\le \rr^{2\rr-2}\cufh^{\rr} \cqdvdd\frac{1}{\llo^2}.
\end{equation}
Assume $\tilde \rr\ge 1$ (the case  $\tilde \rr=0$ will be treated separately below) 
so that $\hrr\ge 1$ and $\tau\in\near{\rr\D\llo}{\setT_k^c}$, 
which implies that $\abs{\phi_{0,k}(\tau)}$ can be upper-bounded by~\fref{eq:phipke3}.
Multiplying~\fref{eq:phipe3d} and~\fref{eq:phipke3} and simplifying, we obtain 
the following upper bound on $E_3$:
\begin{align}
    E_3(\tau)
    =\abs{\phi_{0,k}(\tau)}\abs{\phi''_{+,k}(\tau)}
    &\stackrel{(a)}{\le} 
        \rr^{2\rr-2}\cuft^{\rr} 
        \frac{ \prod_{l=1}^{\hrr} (\vt_l-\tau)^{2}}{(\rr\llo)^{2\hrr}}
        \frac{1}{\llo^2} \\
    &\stackrel{(b)}{\le}
        \rr^{2\rr-2}\cuft^{\rr}
        \frac{\prod_{1\le l\le \tilde\rr} (v_l-\tau)^2}{(\rr\llo)^{2\tilde\rr}}
        \frac{1}{\llo^2}.
    \label{eq:E3d}
\end{align}
Above, (a) we defined $\cuft \define \cufh\cqdvdd\cqu$; in (b) we use the fact 
that $\abs{\vt_l-\tau}/(\llo \rr)\le \D<1$ 
for all $l\in\{1,\ldots,\hrr\}$.

For the case $\tilde r=0$, the upper bound~\fref{eq:E3d} also holds 
by~\fref{eq:phipe3d} because by~\fref{eq:phizkdef} 
and \fref{lem:dualcarlos}, \fref{prop:qzeroone}, $\abs{\phi_{0,k}(\tau)}<1$ 
and because $\cqu>1$ and $\cqdvdd>1$.

From~\fref{eq:phiddsumtermsl1d}, \fref{eq:E1d}, \fref{eq:E2d}, and~\fref{eq:E3d} 
we conclude that
\begin{equation}
     \abs{\phi''_{k}(\tau)}
     \le \rr^{2\rr+1}\cufv^{\rr+1}  
         \frac{\prod_{1\le l\le \tilde\rr} (v_l-\tau)^2}{(\rr\llo)^{2\tilde\rr}}
         \frac{1}{\llo^2},
    \label{eq:phiddfinalbnd1d}
\end{equation}
where we defined $\cufv \define 4 \max(\cufp, \cufs, \cuft)$.

\subparagraph{Putting pieces together.}
Applying~\fref{eq:mvt2} [Mean Value theorem] to the 
function $\fun(\cdot) = \phi_l(\cdot)-\rho-\gamma s'_j(\cdot- t_j)$ with $a=t_j$ 
and $b=\tau$ and using~\fref{eq:init1d}, \fref{eq:q1int4dd}, \fref{eq:phiddfinalbnd1d} 
and  we can write for all $\tau\in\near{\lhi}{t_j}$:
\begin{align}
\abs{\phi_l(\tau)-\rho-\gamma s'_j(\tau- t_j)}
    &\stackrel{(a)}{\le}
        \frac{1}{2} \rr^{2\rr+1}\cufv^{\rr+1} 
        \frac{\prod_{1\le l\le \tilde\rr} (v_l-\tau_m)^2}{(\rr\llo)^{2\tilde\rr}}
        \frac{(\tau-t_j)^2}{\llo^2}\\
    &\stackrel{(b)}{\le} \frac{1}{2} \rr^{2\rr+1}\cufv^{\rr+1}  2^{\tilde\rr}
        \frac{\prod_{1\le l\le \tilde\rr} (v_l-\tau)^2}{(\rr\llo)^{2\tilde\rr}}
        \frac{(\tau-t_j)^2}{\llo^2}\\
    &\stackrel{(c)}{\le} \rr^{2\rr+3}\cufw^{\rr+1} 
        \frac{\prod_{1\le l\le \tilde\rr} (v_l-\tau)^2}{(\rr\llo)^{2\tilde\rr}}
        \frac{(\tau-t_j)^2}{(\rr\llo)^2}.
\label{eq:philtau1d}
\end{align}
Above, 
in (a) $\tau_m\in (t_j, \tau)$; 
in (b) we used that $\abs{v_l-\tau_m}<\abs{v_l-\tau}+\lhi<2\abs{v_l-\tau}$, which 
is true because $\tau\in\near{\lhi}{t_j}$ and because the elements of $\setT$ are separated 
by at least $2\lhi$;  
in (c) we defined $\cufw \define 2\cufv$.

The bound~\fref{eq:toprove2d} follows from~\fref{eq:philtau1d} 
and~\fref{eq:philtau2} by defining $\cufx \define \cufw/\czz$.

\paragraph{Proof of \fref{eq:toprove1d}: case~$t_j\in \setT^c_k$.}
We only need to consider this case when $\rr>1$. Indeed, when $\rr=1$, the sum in \fref{eq:q1defmdd} 
only contains one element, $\phi_l(\cdot)$, and, necessarily, $t_j\in \setT_l$ because $\setT_l^c$ 
is empty.

In this case $t_j$ is one of the elements among 
$\{v_1,\ldots,v_{\tilde\rr}\}\subset \{\vt_1,\ldots,\vt_{\hrr}\}$; 
in other words,
$t_j=v_{\tilde m}=\vt_{\hat m}$ for some $1\le \tilde m\le\tilde \rr$, $1\le \hat m\le\hrr$.
The set $\setT_k\intersect\near{\rr\D\llo-\lhi}{t_j}$ 
is either empty or contains exactly one element. 
Let $b\define\abs{\setT_k\intersect\near{\rr\D\llo-\lhi}{t_j}}$. 
In the case when $b=1$, let $\{\tilde t\}\define \setT_k\intersect\near{\rr\D\llo-\lhi}{t_j}$.

\subparagraph{Bounding $E_1(\tau)$.}
Consider the case $b=1$.
By~\fref{eq:mvt2} [Mean Value theorem] we can write:
\begin{equation}
    \abs{\phi_{+,k}(\tau)}
    \le
        \abs{\phi_{+,k}(\tilde t)}
        +\abs{\phi_{+,k}'(\tilde t)}\abs{\tau- \tilde t}
        +\frac{1}{2}\abs{\phi_{+,k}''(\tau_m)}(\tau- \tilde t)^2
\end{equation}
with $\tau_m\in (\tilde t, \tau)$. Next, we use~\fref{eq:ph11bdnm0d} 
to upper-bound $\abs{\phi_{+,k}(\tilde t)}$; use~\fref{eq:phipderivdd1p} 
to upper-bound $\abs{\phi_{+,k}'(\tilde t)}$; use~\fref{eq:phi11ddmdd} 
to upper-bound $\abs{\phi_{+,k}''(\tau_m)}$. 
With these estimates we can further upper-bound $\abs{\phi_{+,k}(\tau)}$ as follows:
\begin{align}
    \abs{\phi_{+,k}(\tau)}
    \le\rr^{2\rr-1}\cufl^\rr 
        \left(\frac{\lhi^2}{\llo^2} 
            +\frac{\lhi}{\llo^2}\abs{\tau-\tilde t}
            +\frac{1}{\llo^2}(\tau-\tilde t)^2\right)
    &\le \rr^{2\rr+1}\cufm^\rr \frac{(\tilde t-\tau)^2}{(\rr\llo)^2}\\
    &= \rr^{2\rr+1}\cufm^\rr \left[\frac{(\tilde t-\tau)^2}{(\rr\llo)^2}\right]^\I{b=1},
    \label{eq:phipkubnd1d}
\end{align}
where we used that $\lhi\le \abs{\tilde t-\tau}$ because the elements of $\setT$
are separated by at least $2\lhi$ and $\tau\in\near{\lhi}{t_j}$ with $\tilde t\ne t_j$. According 
to~\fref{eq:phi11ddm00dd} the upper bound~\fref{eq:phipkubnd1d} also holds for $b=0$.

Since $t_j\in \setT^c_k$ and $\tau\in\near{\lhi}{t_j}$, 
it follows $\tau\in\near{\rr\D\llo}{\setT_k^c}$ 
so that $\hrr\ge 1$, which implies that $\abs{\phi''_{0,k}(\tau)}$ 
can be upper-bounded by~\fref{eq:qddub3}.

Multiplying~\fref{eq:phipkubnd1d} and~\fref{eq:qddub3} and simplifying, 
we obtain the following upper bound on $E_1$:
\begin{align}
    E_1(\tau)
    =\abs{\phi''_{0,k}(\tau)}\abs{\phi_{+,k}(\tau)}
    &\stackrel{(a)}{\le} \rr^{2\rr+1}\cufy^{\rr+1} 
        \left[\frac{(\tilde t-\tau)^2}{(\rr\llo)^2}\right]^\I{b=1} 
        \frac{ \prod_{1\le l\le \hrr, l\ne\hat m} (\vt_l-\tau)^{2}}{(\rr\llo)^{2(\hrr-1)}}
        \frac{1}{\llo^2}\nonumber\\
    &\stackrel{(b)}{\le} \rr^{2\rr+1}\cufy^{\rr+1}
        \left[\frac{(\tilde t-\tau)^2}{(\rr\llo)^2}\right]^\I{b=1}
        \frac{ \prod_{1\le l\le \tilde\rr, l\ne\tilde m} (v_l-\tau)^{2}}{(\rr\llo)^{2(\tilde\rr-1)}}
        \frac{1}{\llo^2}.
    \label{eq:E12d}
\end{align}
Above, in (a) we defined $\cufy\define\cufm\cuoss$; in (b) we use the fact 
that $\abs{\vt_l-\tau}/(\llo \rr)\le \D<1$ for all $l\in\{1,\ldots,\hrr\}$.

\subparagraph{Bounding $E_2(\tau)$.}
Consider the case $b=1$.
By \fref{eq:mvt1} [Mean Value theorem] we can write:
\begin{equation}
    \abs{\phi'_{+,k}(\tau)}\le \abs{\phi'_{+,k}(\tilde t)}+\abs{\phi_{+,k}''(\tau_m)}\abs{\tau- \tilde t}
\end{equation}
with $\tau_m\in (\tilde t, \tau)$. Next, we use~\fref{eq:phipderivdd1p} 
to upper-bound $\abs{\phi_{+,k}'(\tilde t)}$; use~\fref{eq:phi11ddmdd} 
to upper-bound $\abs{\phi_{+,k}''(\tau_m)}$. 
With these estimates we can further upper-bound $\abs{\phi'_{+,k}(\tau)}$ as follows:
\begin{align}
    \abs{\phi'_{+,k}(\tau)}
    \le\rr^{2\rr-1}\cufq^\rr 
        \left(\frac{\lhi}{\llo^2}+\frac{1}{\llo^2}\abs{\tau-\tilde t}\right)
    &\le \rr^{2\rr}\cufr^\rr 
        \frac{\abs{\tau-\tilde t}}{\rr\llo}\frac{1}{\llo}\\
    &= \rr^{2\rr}\cufr^\rr 
        \left[\frac{\abs{\tau-\tilde t}}{\rr\llo}\right]^{\I{b=1}}\frac{1}{\llo},
    \label{eq:phipkubnd2d}
\end{align}
where we used that $\lhi\le \abs{\tilde t-\tau}$. 
According to~\fref{eq:phi11ddm0dd} the upper bound~\fref{eq:phipkubnd2d} also holds for $b=0$.

Since $t_j\in \setT^c_k$ and $\tau\in\near{\lhi}{t_j}$, 
it follows $\tau\in\near{\rr\D\llo}{\setT_k^c}$ 
so that $\hrr\ge 1$, which implies that $\abs{\phi'_{0,k}(\tau)}$ 
can be upper-bounded by~\fref{eq:qddub4}.

Multiplying~\fref{eq:phipkubnd2d} and~\fref{eq:qddub4} and simplifying
we obtain the following upper bound on $E_2$:
\begin{align}
    E_2(\tau)
    =\abs{\phi'_{0,k}(\tau)}\abs{\phi'_{+,k}(\tau)}
    &\stackrel{(a)}{\le} 
        \rr^{2\rr} \cufz^\rr 
        \left[\frac{(\tilde t-\tau)^2}{(\rr\llo)^2}\right]^{\I{b=1}} 
        \frac{\prod_{1\le l\le \hrr,l\ne \hat m} (\vt_{l}-\tau)^{2}}{(\rr\llo)^{2(\hrr-1)}} 
        \frac{1}{\llo^2} \nonumber\\
    &\stackrel{(b)}{\le}
        \rr^{2\rr} \cufz^\rr
        \left[\frac{(\tilde t-\tau)^2}{(\rr\llo)^2}\right]^{\I{b=1}}
        \frac{ \prod_{1\le l\le \tilde\rr,l\ne \tilde m} (v_{l}-\tau)^{2}}{(\rr\llo)^{2(\tilde\rr-1)}} 
        \frac{1}{\llo^2}.
    \label{eq:E22d}
\end{align}
Above, 
in (a)  we defined $\cufz \define \cufr \cuossss$; 
in (b) we used the fact that $\abs{\vt_l-\tau}/(\llo \rr)\le \D<1$ for all $l\in\{1,\ldots,\hrr\}$.

\subparagraph{Bounding $E_3(\tau)$.}
By~\fref{eq:phi11ddmdd},
\begin{equation}
    \label{eq:phipkubnd3d}
    \abs{\phi_{+,k}''(\tau)}\le \rr^{2\rr-2}\cufh^{\rr} \cqdvdd\frac{1}{\llo^2}.
\end{equation}

Since $t_j\in \setT^c_k$ and $\tau\in\near{\lhi}{t_j}$, 
it follows $\tau\in\near{\rr\D\llo}{\setT_k^c}$ 
so that $\hrr\ge 1$, which implies that $\abs{\phi_{0,k}(\tau)}$ 
can be upper-bounded as in~\fref{eq:qddub51}.
Multiplying~\fref{eq:phipkubnd3d} and~\fref{eq:qddub51} and simplifying,
we obtain the following upper bound on $E_3$:
\begin{align}
    E_3(\tau)
    =\abs{\phi_{0,k}(\tau)}\abs{\phi''_{+,k}(\tau)}
    &\stackrel{(a)}{\le} 
        \rr^{2\rr-2}\cuga^{\rr} 
        \left[\frac{(\tilde t-\tau)^2}{(\rr\llo)^2}\right]^{\I{b=1}} 
        \frac{ \prod_{1\le l\le \hrr,l\ne \hat m} (\vt_{l}-\tau)^{2}}{(\rr\llo)^{2(\hrr-1)}} 
        \frac{1}{\llo^2} \\
    &\stackrel{(b)}{\le} \rr^{2\rr-2}\cuga^{\rr} 
        \left[\frac{(\tilde t-\tau)^2}{(\rr\llo)^2}\right]^{\I{b=1}} 
        \frac{\prod_{1\le l\le \tilde\rr, l\ne \tilde m} 
        (v_l-\tau)^2}{(\rr\llo)^{2(\tilde\rr-1)}} 
        \frac{1}{\llo^2}.
    \label{eq:E23d}
\end{align}
Above, in (a) we defined $\cuga \define \cufh\cqdvdd\cqu$; in (b) we used the fact 
that $\abs{\vt_l-\tau}/(\llo \rr)\le \D<1$ for all $l\in\{1,\ldots,\hrr\}$.

From~\fref{eq:phiddsumtermsl1d}, \fref{eq:E12d}, \fref{eq:E22d}, 
and~\fref{eq:E23d} we conclude that
\begin{equation}
     \abs{\phi''_{k}(\tau)}\le \rr^{2\rr+1}\cugb^{\rr+1}  \left[\frac{(\tilde t-\tau)^2}{(\rr\llo)^2}\right]^\I{b=1} \frac{\prod_{1\le l\le \tilde\rr, l\ne \tilde m} (v_l-\tau)^2}{(\rr\llo)^{2(\tilde\rr-1)}} \frac{1}{\llo^2},
    \label{eq:phiddfinalbnd2d}
\end{equation}
where we defined $\cugb \define 4\max(\cufy, \cufz, \cuga)$.

\subparagraph{Putting pieces together.}
Applying \fref{eq:mvt2} [Mean Value theorem] to the 
function $\fun(\cdot) = \phi_l(\cdot)$ with $a=t_j$ and $b=\tau$ and 
using~\fref{eq:init2d},~\fref{eq:init3d}, \fref{eq:phiddfinalbnd2d}, and 
we can write for all $\tau\in\near{\lhi}{t_j}$:
\begin{align}
    \abs{\phi_k(\tau)}
    &\stackrel{(a)}{\le} 
        \frac{1}{2} \rr^{2\rr+1}\cugb^{\rr+1}  
        \left[\frac{(\tilde t-\tau_m)^2}{(\rr\llo)^2}\right]^\I{b=1} 
        \frac{\prod_{1\le l\le \tilde\rr, l\ne \tilde m} (v_l-\tau_m)^2}{(\rr\llo)^{2(\tilde\rr-1)}} 
        \frac{(\tau-t_j)^2}{\llo^2}\\
    &\stackrel{(b)}{\le} 
        \frac{1}{2} \rr^{2\rr+1}\cugb^{\rr+1}  2^{\tilde\rr}
        \left[\frac{(\tilde t-\tau)^2}{(\rr\llo)^2}\right]^\I{b=1} 
        \frac{\prod_{1\le l\le \tilde\rr, l\ne \tilde m} (v_l-\tau)^2}{(\rr\llo)^{2(\tilde\rr-1)}} 
        \frac{(\tau-t_j)^2}{\llo^2}\\
    &\stackrel{(c)}{\le} \rr^{2\rr+3}\cugc^{\rr+1}  
        \left[\frac{(\tilde t-\tau)^2}{(\rr\llo)^2}\right]^\I{b=1}
        \frac{\prod_{1\le l\le \tilde\rr} (v_l-\tau)^2}{(\rr\llo)^{2\tilde\rr}}.
\label{eq:philtau11d}
\end{align}
Above, 
in (a) $\tau_m\in (t_j, \tau)$; 
in (b) we used the fact that, for $l\ne\tilde m$, $\abs{v_l-\tau_m}<\abs{v_l-\tau}+\lhi<2\abs{v_l-\tau}$ 
and $\abs{\tilde t-\tau_m}<\abs{\tilde t-\tau}+\lhi<2\abs{\tilde t-\tau}$,
which is true because $\tau\in\near{\lhi}{t_j}$ and because the elements of $\setT$ are separated 
by at least $2\lhi$; 
in (c) we defined $\cugc \define 2\cugb$ and used the fact that $t_j=v_{\tilde m}$.

The bound~\fref{eq:toprove1d} follows from~\fref{eq:philtau11d} and~\fref{eq:philtau21} by defining $\cugg \define \cugc/\czz$.

\subsection{Proof of property 3}
By~\fref{eq:q1defmdd} and the triangle inequality:
\begin{align}
    \infnorm{q_2(\cdot)}
    \le \rho+\rr \max_{1\le k\le \rr}\infnorm{\phi_k(\cdot)}
    &\stackrel{(a)}{\le} \rho+\rr \max_{1\le k\le \rr}\infnorm{\phi_{+,k}(\cdot)}\\
    &\stackrel{(b)}{=} \rho+\rr^{2\rr+1}\cufh^{\rr} 
        \max_{1\le k\le \rr}\infnorm{q_{\rr\llo,\setT_k,\{f_j\},\{d_j\}}(\cdot)}\\
    &\stackrel{(c)}{\le} \rho+\rr^{2\rr+1}\cufh^{\rr}\cqdv 
    \stackrel{(d)}{\le} \rr^{2\rr+1}\cugn^{\rr}.
\end{align} 
Above, in (a) we used~\fref{eq:phikdfndd} and the fact that by~\fref{eq:phizkdef} and 
\fref{lem:dualcarlos}, \fref{prop:qzeroone}, $\infnorm{\phi_{0,k}(\cdot)}\le 1$;
in (b) we used~\fref{eq:phipdefdd}; in (c) we used \fref{lem:dualcarlos2}, \fref{prop:dqinfbnd};
in (d) we defined $\cugn\define 2\cqdv\cufh$ and used the fact that $\rho<1<\cqdv\cufh$.

\subsection{Proof of property 4}
\label{sec:proofdualprod14}

Take $\tau\in\far{\lhi}{\setT}$. 
As above, let 
$\{u^\tau_1,\ldots,u^\tau_{\breve \rr}\}\define \near{\rr\D\llo}{\tau}\intersect\setT$. 
Then by~\fref{eq:qzlp},
\begin{equation}
    \label{eq:qzlbprd1}
q_0(\tau)
\ge \cz^\rr  \frac{ \prod_{l=1}^{\breve  \rr}(u^\tau_l-\tau)^2}{(\rr\llo)^{2\breve \rr}}.
\end{equation}
By~\fref{eq:q0lbxx} this bound is also 
valid when $\breve \rr=0$.

Fix $k$. If $\tau\in\near{\rr\D\llo}{\setT_k^c}$, then we can 
use~\fref{eq:phikub} to upper-bound $\abs{\phi_{0,k}(\tau)}$:
\begin{align}
    \abs{\phi_{0,k}(\tau)} 
    &\le \cqu^{\hrr}\frac{ \prod_{l=1}^{\hrr} (\vt_l-\tau)^{2}}{(\rr\llo)^{2\hrr}},
    \label{eq:qzkubprd1}
\end{align}
where, as before, $\{\vt_1,\ldots,\vt_{\hrr}\}\define \near{\rr\D\llo}{\tau}\intersect\setT_k^c$. If $\tau\notin\near{\rr\D\llo}{\setT_k^c}$, we will use that  by~\fref{eq:phizk} and by \fref{lem:dualcarlos}, \fref{prop:qzeroone},
\begin{equation}
    \abs{\phi_{0,k}(\tau)} \le 1.\label{eq:qzkubprunivd1}
\end{equation}
The set $\setT_k\intersect\near{\rr\D\llo}{\tau}$ is either empty or contains 
exactly one element. Let $b\define\abs{ \setT_k\intersect\near{\rr\D\llo}{\tau}}$ 
denote the size of this set; when $b=1$, let $\{\tilde t\}\define \setT_k\intersect\near{\rr\D\llo}{ \tau}$.
Following the steps that lead to~\fref{eq:phipkubnd1d}, we obtain:
\begin{equation}
    \label{eq:qokubprd1}
    \abs{\phi_{+,k}(\tau)}\le \rr^{2\rr+1}\cufm^\rr \left[\frac{(\tilde t-\tau)^2}{(\rr\llo)^2}\right]^\I{b=1}
\end{equation}
and the bound is valid for both cases $b=0$ and $b=1$.

\paragraph*{Case $\hrr\ge 1$:}
Then, $\{u^\tau_1,\ldots,u^\tau_{\breve \rr}\}=\{\vt_1,\ldots,\vt_{\hrr}\}\union \{\tilde t\}$ 
if $b=1$, and $\{u^\tau_1,\ldots,u^\tau_{\breve \rr}\}=\{\vt_1,\ldots,\vt_{\hrr}\}$ if $b=0$. 
Therefore,
\begin{align}
    \abs{\phi_{k}(\tau)}
    =\abs{\phi_{0,k}(\tau)}\abs{\phi_{+,k}(\tau)}
    &\stackrel{(a)}{\le} 
        \rr^{2\rr+1}\cufm^\rr \cqu^{\hrr} 
        \left[\frac{(\tilde t-\tau)^2}{(\rr\llo)^2}\right]^\I{b=1} 
        \frac{ \prod_{l=1}^{\hrr} (\vt_l-\tau)^{2}}{(\rr\llo)^{2\hrr}}\\
    &= \rr^{2\rr+1}\cufm^\rr \cqu^{\hrr}  
        \frac{ \prod_{l=1}^{\breve\rr} (u^\tau_l-\tau)^{2}}{(\rr\llo)^{2\breve\rr}}
    \stackrel{(b)}{\le} \rr^{2\rr+1}\cugf^\rr  q_0(\tau).
    \label{eq:phiq0ub1d1}
\end{align}
Above, 
(a) follows by~\fref{eq:qzkubprd1} and~\fref{eq:qokubprd1}; 
(b) follows by~\fref{eq:qzlbprd1} with $\cugf \define \cufm \cqu/\cz$.

\paragraph*{Case $\hrr= 0$:}
Then, $\breve r=1$ and $\{u^\tau_{\breve\rr}\}= \{\tilde t\}$ if $b=1$ 
and $\breve \rr=0$ if $b=0$. Therefore,
\begin{align}
    \abs{\phi_{k}(\tau)}
    =\abs{\phi_{0,k}(\tau)}\abs{\phi_{+,k}(\tau)}
    &\stackrel{(a)}{\le} 
        \rr^{2\rr+1}\cufm^\rr 
        \left[\frac{(\tilde t-\tau)^2}{(\rr\llo)^2}\right]^\I{b=1}\nonumber\\
    &=
        \rr^{2\rr+1}\cufm^\rr 
        \frac{ \prod_{l=1}^{\breve  \rr}(u^\tau_l-\tau)^2}{(\rr\llo)^{2\breve \rr}}
    \stackrel{(b)}{\le} \rr^{2\rr+1}\cugf^\rr  q_0(\tau).
    \label{eq:phiq0ub2d1}
\end{align}
Above, 
(a) follows by~\fref{eq:qzkubprunivd1} and \fref{eq:qokubprd1}; 
(b) follows by~\fref{eq:qzlbprd1} because $\cqu>1$.

By \fref{lem:dualprod}, \fref{prop:qzloinfhi},
\begin{equation}
    \rho=\frac{\lhi^{2\rr}}{\llo^{2\rr}} \le  \rr^{2\rr} \frac{1}{\cql^\rr}q_0(\tau).
    \label{eq:rho2ubd1}
\end{equation}

Therefore, by~\fref{eq:q1defmdd}, \fref{eq:phiq0ub1d1}, \fref{eq:phiq0ub2d1}, \fref{eq:rho2ubd1},%
\begin{equation}
    \abs{q_2(\tau)}\le \sum_{k=1}^\rr \abs{\phi_{k}(\tau)} + \rho\le \rr^{2\rr+2} \cugh^\rr q_0(\tau),
\end{equation}
where we defined $\cugh\define\cugf+1/\cql$.
\qed

\section{Mean Value theorem}
We repeatedly use the Taylor series approximation with the remained 
expressed via the Mean Value theorem~\cite[p.~880,~25.2.25]{abramowitz64} 
given below for the convenience of the reader.
\begin{thm}
    \label{thm:meanvalue}
    Assume that $\fun(t)$ is twice differentiable on the interval $[a, b]$. Then, there exists $t_1\in (a,b)$ such that
    \begin{equation}
        \fun(b) = \fun(a) + \fun'(t_1) (b-a).
        \label{eq:mvt1}
    \end{equation}
    and there exists $t_2\in(a,b)$ such that
    \begin{equation}
        \fun(b) = \fun(a) + \fun'(a) (b-a) + \frac{\fun''(t_2)}{2} (b-a)^2.
        \label{eq:mvt2}
    \end{equation}
\end{thm}

\section{Properties of Fejér kernel}
\label{app:propfejer}
The results proven in subsections below are analogous to the results 
in~\cite[eq.~(1.11) and eq.~(2.6)]{candes21-2} with the difference that 
here we need bounds on sums and in~\cite{candes21-2} bounds on the 
corresponding integrals are provided. 

Below, we will need uniform upper bounds 
on $\abs{\khi(\cdot)}$, $\abs{\khi'(\cdot)}$, and  $\abs{\khi''(\cdot)}$; these are derived next.

Fejér kernel~\fref{eq:khi} can be written as a Fourier sum as follows:
\begin{equation}
    \label{eq:fejsum}
    \khi(t) = \frac{1}{\N}\sum_{\abs{k}\le \fhi} \left(1-\frac{\abs{k}}{\fhi+1}\right) e^{\iu 2\pi t k}.
\end{equation}
Taking the absolute value of both sides in~\fref{eq:fejsum} and applying 
the triangle inequality we find:
\begin{align}
    \abs{\khi(t)}
    &\le \frac{1}{\N}\sum_{\abs{k}\le \fhi} \left(1-\frac{\abs{k}}{\fhi+1}\right)
    =\frac{1}{\N}\left(1+\fhi\right).\label{eq:ddu33}
\end{align}
Above, the equality follows by summing up the simple series.

Differentiating~\fref{eq:fejsum} we obtain:
\begin{equation}
    \label{eq:fejsumd}
    \khi'(t) 
    = \frac{\iu 2\pi}{\N} 
        \sum_{\abs{k}\le \fhi} \left(1-\frac{\abs{k}}{\fhi+1}\right) k e^{\iu 2\pi t k}.
\end{equation}
Taking the absolute value of both sides in~\fref{eq:fejsumd} and 
applying the triangle inequality we find:
\begin{align}
    \abs{\khi'(t)}
    \le  \frac{2\pi}{\N} \sum_{\abs{k}\le \fhi} \abs{k} \left(1-\frac{\abs{k}}{\fhi+1}\right)
    \stackrel{(a)}{=}\frac{2\pi}{3 \N}\fhi (2 + \fhi)
    \stackrel{(b)}{\le} \frac{2\pi}{\N}\fhi^2.\label{eq:ddub}
\end{align}
Above, (a) follows by summing up the simple power series; (b) follows because $\fhi>1$.

Differentiating~\fref{eq:fejsumd} we obtain:
\begin{equation}
    \label{eq:fejsumdd}
    \khi''(t) = - \frac{(2\pi)^2}{\N} \sum_{\abs{k}\le \fhi} \left(1-\frac{\abs{k}}{\fhi+1}\right) k^2 e^{\iu 2\pi t k}.
\end{equation}
Taking the absolute value of both sides in~\fref{eq:fejsumdd} and applying the triangle inequality we find:
\begin{align}
    \abs{\khi''(t)}
    \le  \frac{(2\pi)^2}{\N} \sum_{\abs{k}\le \fhi} k^2 \left(1-\frac{\abs{k}}{\fhi+1}\right)
    \stackrel{(a)}{=}\frac{2\pi^2}{3 \N} \fhi (2 + 3\fhi + \fhi^2)
    \stackrel{(b)}{\le} \frac{4\pi^2}{\N}\fhi^3. \label{eq:ddubd}
\end{align}
Above, (a) follows by summing up the simple power series; (b) follows because $\fhi>1$.

The bounds derived below in this section are crude in the sense that no attempt 
has been made to obtain the tightest possible constants; for this reason some of 
the steps below may appear unnecessarily wasteful.

\subsection{Proof of~(\ref{eq:cabshid})}
Begin by splitting the summation interval and recombining the terms in the following way:
\begin{align}
    \sum_{n=0}^{\N-1} \abs{\khi'\lefto(\frac{n}{\N}\right)} 
    &=\sum_{n: n/\N\in[0, \lhi)} \abs{\khi'\lefto(\frac{n}{\N}\right)}
        +\sum_{n: n/\N\in[\lhi, 1/2)} \abs{\khi'\lefto(\frac{n}{\N}\right)}\\
    &\dummyrel{=}+\sum_{n: n/\N\in[1/2, 1-\lhi)} \abs{\khi'\lefto(\frac{n}{\N}\right)} 
        +\sum_{n: n/\N\in[1-\lhi, 1)} \abs{\khi'\lefto(\frac{n}{\N}\right)}\\
    &\le \,2\!\!\!\!\!\!\!\sum_{n: n/\N\in[0, \lhi+1/\N)} \abs{\khi'\lefto(\frac{n}{\N}\right)} + 
        2\!\!\!\!\!\!\!\sum_{n: n/\N\in[\lhi, 1/2+1/\N)} \abs{\khi'\lefto(\frac{n}{\N}\right)}.\label{eq:sumkhiubd}
\end{align}
Above, the inequality follows by symmetry of $\khi'(\cdot)$ around $1/2$. 
Next, we upper-bound the two sums separately. 

To upper-bound the first sum in~\fref{eq:sumkhiubd} we proceed as follows:
\begin{equation}
    \sum_{n: n/\N\in[0, \lhi+1/\N)} \abs{\khi'\lefto(\frac{n}{\N}\right)}
    \le (\lhi+1/\N)\N\max_t \abs{\khi'\lefto(t\right)} 
    \stackrel{(a)}{\le} 2\lhi 2\pi \fhi^2 
    \stackrel{(b)}{=} 4\pi/\lhi. \label{eq:sumkhiub1d}
\end{equation} 
Above, 
in (a) we used~\fref{eq:ddub} and the assumption that $1/\N\le\lhi$; 
in (b) we used that $\fhi = 1/\lhi$.

To upper-bound the second term in~\fref{eq:sumkhiubd} we observe 
that $\abs{\khi'(\cdot)}$ can be upper-bounded as follows for $t\in[0,0.55]$:
\begin{align}
    \abs{\khi'(t)} 
    &\stackrel{(a)}{=}
    \frac{\pi \sin(\pi(\fhi+1)t)}{\N}
    \left|
        \left[\frac{2 \cos(\pi(\fhi+1) t)}{\sin^2(\pi t)}-\frac{2\cos(\pi t)\sin(\pi(\fhi+1)t)}
        {(\fhi+1) \sin^3(\pi t)}\right]
    \right|\\
    &\stackrel{(b)}{\le} \frac{1}{\N}\left(\frac{2\pi}{\sin^2(\pi t)}+\frac{2\pi }{(\fhi+1)\sin^3(\pi t)}\right)
    \stackrel{(c)}{\le} \frac{1}{\N}\left(\frac{2}{(\fhi+1) t^3}+\frac{2}{t^2}\right).\label{eq:fejqdbndd}
\end{align}
Above, (a) follows by differentiating $\khi(\cdot)$ in~\fref{eq:khi}; in (b) we used the triangle inequality and the fact
that $\abs{\sin(\cdot)}\le 1$, $\abs{\cos(\cdot)}\le 1$; in (c) we used the inequalities $\sin(\pi t)^2\ge \pi t^2$ and $\sin(\pi t)^3\ge \pi t^3$ for $t\in[0, 0.55]$.
Therefore,
\begin{align}
    \sum_{n: n/\N\in[\lhi, 1/2+1/\N)}\!\!\! \abs{\khi'\lefto(\frac{n}{\N}\right)}
    &\stackrel{(a)}{\le}
        \frac{2}{1+\fhi} 
        \frac{1}{\N}\!\!\!\sum_{n: n/\N\in[\lhi, 1/2+1/\N)}\!\!\! 
            \frac{1}{(n/\N)^3} + 
            \frac{2}{\N}\!\!\!\sum_{n: n/\N\in[\lhi, 1/2+1/\N)}\!\!\! \frac{1}{(n/\N)^2}\\
    &\stackrel{(b)}{\le}
        \frac{2}{1+\fhi}
        \left(\frac{1}{\N}\frac{1}{\lhi^3}+\int_{\lhi}^{1/2+1/\N} \frac{1}{t^3}dt\right) 
        + \frac{2}{\N}\frac{1}{\lhi^2}+2\int_{\lhi}^{1/2+1/\N} \frac{1}{t^2}dt\\
    &\le
        \frac{2}{1+\fhi}
        \left(\frac{1}{\N}\frac{1}{\lhi^3}+\int_{\lhi}^\infty \frac{1}{t^3}dt\right) 
        +\frac{2}{\N}\frac{1}{\lhi^2}+2\int_{\lhi}^\infty \frac{1}{t^2}dt \\
    &=\frac{2}{1+\fhi} 
        \left(\frac{1}{\N}\frac{1}{\lhi^3}+\frac{1}{2\lhi^2}\right) 
        + \frac{2}{\N}\frac{1}{\lhi^2}+\frac{2}{\lhi}\\
    &\stackrel{(c)}{\le}
    \frac{3}{1+\fhi} \frac{1}{\lhi^2}+\frac{4}{\lhi}
    <\frac{7}{\lhi}.\label{eq:sumkhiub2d}
\end{align}
Above, 
(a) follows from~\fref{eq:fejqdbndd} because $1/2+1/\N<0.55$; 
in (b) the bound for the first term follows because the function $1/t^3$ 
is monotonically decreasing and the bound for the second term follows 
because the function $1/t^2$ is monotonically decreasing; 
(c) follows because $1/\N\le\lhi$.

Finally, plugging~\fref{eq:sumkhiub1d} and~\fref{eq:sumkhiub2d} 
into~\fref{eq:sumkhiubd} and defining $\cabshid=8\pi+14$ 
we obtain~\fref{eq:cabshid}.

\subsection{Proof of~(\ref{eq:cabshidd})}
Begin by splitting the summation interval and recombining the terms in the following way:
\begin{align}
    \frac{1}{2}\sum_{n=0}^{\N-1} \sup_{u\in\nearhi{n/\N}} \abs{\khi''(u)}
    &= \frac{1}{2}\sum_{n: n/\N\in[0, 2\lhi)} \sup_{u\in\nearhi{n/\N}} \abs{\khi''(u)}\\
    &\dummyrel{=} +\frac{1}{2}\sum_{n: n/\N\in[2\lhi, 1/2)} \sup_{u\in\nearhi{n/\N}} \abs{\khi''(u)}\\
    &\dummyrel{=} +\frac{1}{2}\sum_{n: n/\N\in[1/2, 1-2\lhi)} \sup_{u\in\nearhi{n/\N}} \abs{\khi''(u)}\\
    &\dummyrel{=} + \frac{1}{2}\sum_{n: n/\N\in[1-2\lhi, 1)} \sup_{u\in\nearhi{n/\N}} \abs{\khi''(u)}\\
    &\le 
    \sum_{n: n/\N\in[0, 2\lhi+1/\N)} \sup_{u\in\nearhi{n/\N}} \abs{\khi''(u)}\\
    &\dummyrel{=} + \sum_{n: n/\N\in[2\lhi, 1/2+1/\N)} \sup_{u\in\nearhi{n/\N}} \abs{\khi''(u)}.
    \label{eq:sumkhiubdd}
\end{align}
Above, the inequality follows by symmetry of $\sup_{u\in\nearhi{\cdot}} \abs{\khi''(u)}$ 
around $1/2$. Next, we upper-bound the two sums separately. 

To upper-bound the first sum in~\fref{eq:sumkhiubdd} we proceed as follows:
\begin{equation}
    \sum_{n: n/\N\in[0, 2\lhi+1/\N)} \sup_{u\in\nearhi{n/\N}} \abs{\khi''(u)}
    \le \left(2\lhi+\frac{1}{\N}\right)\N\max_t \abs{\khi''\lefto(t\right)} \stackrel{(a)}{\le} 3\lhi 4\pi^2\fhi^3 \stackrel{(b)}{=} \frac{12\pi^2}{\lhi^2}.\label{eq:sumkhiub1dd}
\end{equation} 
Above, 
in (a) we used~\fref{eq:ddubd} and the assumption that $1/\N\le\lhi$;
in (b) we used that $\fhi = 1/\lhi$.

To upper-bound the second term in~\fref{eq:sumkhiubdd} we differentiate $\khi(\cdot)$ in~\fref{eq:khi} twice to obtain:
\begin{align}
   \khi''(t) 
   &= 
       -\frac{1}{\N}
       \frac{\pi ^2 (-\sin(2\pi t) \sin(2\pi\fhi t)+\cos(2\pi t) \cos(2\pi\fhi t) + 2\cos(2\pi\fhi t)-\cos (2\pi t)-2)}
            {(\fhi+1) \sin^4(\pi t)}\\
   &\dummyrel{=}
       -\frac{1}{\N}\frac{4 \pi ^2 \fhi \sin (\pi  (2 \fhi+1) t)}{(\fhi+1) \sin ^3(\pi  t)}
       +\frac{1}{\N}\frac{2 \pi ^2 \fhi^2 \cos (2 \pi (\fhi+1) t)}{(\fhi+1) \sin ^2(\pi  t)}.
\end{align}
This leads to the following upper bound on $\abs{\khi''(t)}$ for $t\in[0,0.55]$:
\begin{align}
    \abs{\khi''(t)} 
    &\stackrel{(a)}{\le}\frac{1}{\N}
        \left(\frac{7 \pi^2}{(\fhi+1) \sin^4(\pi  t)} 
            +\frac{4 \pi^2 \fhi}{(\fhi+1) \sin^3(\pi  t)}
            +\frac{2 \pi ^2 \fhi^2}{(\fhi+1) \sin^2(\pi  t)}
        \right)\\
    &\stackrel{(b)}{\le}\frac{1}{\N}
        \left(\frac{7 \pi}{\fhi t^4}
            +\frac{4 \pi}{t^3}
            +\frac{2 \pi \fhi}{t^2}
        \right).\label{eq:fejqdbnddd}
\end{align}
Above, in (a) we used the triangle inequality and the fact
that $\abs{\sin(\cdot)}\le 1$, $\abs{\cos(\cdot)}\le 1$; in (b) we used the inequalities $\sin(\pi t)^2\ge \pi t^2$, $\sin(\pi t)^3\ge \pi t^3$, $\sin(\pi t)^4\ge \pi t^4$ for $t\in[0, 0.55]$.
Next observe that since the right-hand side of~\fref{eq:fejqdbnddd} is monotonically decreasing for $t>0$ we have for $t\in(\lhi, 0.55]$:
\begin{align}
    \sup_{u\in\nearhi{t}} \abs{\khi''(u)}
    \le \frac{1}{\N}
        \left(\frac{7 \pi}{\fhi (t-\lhi)^4} 
            +\frac{4 \pi}{(t-\lhi)^3}+\frac{2 \pi \fhi}{(t-\lhi)^2}
        \right).
        \label{eq:supi}
\end{align}
Therefore,   the second term in~\fref{eq:sumkhiubdd} can be upper-bounded as follows:
\begingroup
\allowdisplaybreaks
\begin{align}
    &\sum_{n: n/\N\in[2\lhi, 1/2+1/\N)}\!\! \sup_{u\in\nearhi{n/\N}} \abs{\khi''(u)}\\
    &\qquad\qquad\stackrel{(a)}{\le}
        \frac{7 \pi}{\fhi}\frac{1}{\N}\!\!\!\!\!\sum_{n: n/\N\in[2\lhi, 1/2+1/\N)}\!\!\frac{1}{ (n/\N-\lhi)^4} \\
    &\qquad\qquad\dummyrel{\le}
        +4 \pi\frac{1}{\N}\!\!\!\!\!\sum_{n: n/\N\in[2\lhi, 1/2+1/\N)}\!\!\frac{1}{(n/\N-\lhi)^3}\\
    &\qquad\qquad\dummyrel{\le}
        +2 \pi \fhi \frac{1}{\N}\!\!\!\!\!\sum_{n: n/\N\in[2\lhi, 1/2+1/\N)}\frac{1}{(n/\N-\lhi)^2}\\
    &\qquad\qquad\stackrel{(b)}{\le}
        \frac{7 \pi}{\fhi}
        \left(\frac{1}{\N}\frac{1}{(2\lhi-\lhi)^4} 
            + \int_{2\lhi}^{1/2+1/\N} \frac{1}{(t-\lhi)^4}d t\right)\\
    &\qquad\qquad\dummyrel{\le}
        +4 \pi\left(\frac{1}{\N}\frac{1}{(2\lhi-\lhi)^3} 
            + \int_{2\lhi}^{1/2+1/\N} \frac{1}{(t-\lhi)^3}d t\right)\\
    &\qquad\qquad\dummyrel{\le}
        +2 \pi \fhi 
        \left(\frac{1}{\N}\frac{1}{(2\lhi-\lhi)^2} 
            + \int_{2\lhi}^{1/2+1/\N} \frac{1}{(t-\lhi)^2}d t\right)\\
    &\qquad\qquad\stackrel{(c)}{\le}
        \frac{7 \pi}{\fhi}
        \left(\frac{1}{\N}\frac{1}{\lhi^4} 
            + \int_{\lhi}^{\infty} \frac{1}{t^4}d t\right)
        +4 \pi\left(\frac{1}{\N}\frac{1}{\lhi^3} 
            + \int_{\lhi}^{\infty} \frac{1}{t^3}d t\right)
        +2 \pi \fhi \left(\frac{1}{\N}\frac{1}{\lhi^2} 
            + \int_{\lhi}^{\infty} \frac{1}{t^2}d t\right)\\
    &\qquad\qquad=
        \frac{7 \pi}{\fhi}
        \left(\frac{1}{\N}\frac{1}{\lhi^4} + \frac{1}{3\lhi^3}\right)
        +4 \pi\left(\frac{1}{\N}\frac{1}{\lhi^3} +  \frac{1}{2\lhi^2}\right)
        +2 \pi \fhi \left(\frac{1}{\N}\frac{1}{\lhi^2} +  \frac{1}{\lhi}\right)\\
    &\qquad\qquad\stackrel{(d)}{\le}
        \left(\frac{28}{3}\pi + 6\pi + 4\pi\right)\frac{1}{\lhi^2} 
    = \frac{58\pi}{3} \frac{1}{\lhi^2}.\label{eq:abra}
\end{align}
\endgroup
Above, (a) follows from \fref{eq:supi} because $1/2+1/\N<0.55$;
(b) follows because the functions $1/(t-\lhi)^4$, $1/(t-\lhi)^3$, 
and $1/(t-\lhi)^2$ are monotonically decreasing; 
(c) follows by changing the integration variable;
(d)  follows because $1/\N\le \lhi$ and because $\fhi=1/\lhi$.

Finally, plugging~\fref{eq:sumkhiub1dd} and~\fref{eq:abra} 
into~\fref{eq:sumkhiubdd} and defining $\cabshidd\define 12\pi^2+58\pi/3$ 
we obtain~\fref{eq:cabshidd}.

\bibliographystyle{IEEEtran}
\bibliography{vm_bib.bib}
\end{document}

%% file: math-macros.tex
\usepackage{amsmath}
\usepackage{amssymb}
\usepackage{amsfonts}
\usepackage{amsthm}
\usepackage{mathrsfs}
\usepackage{xspace}
\usepackage{bm}
\usepackage{fancyref}
\usepackage{textcomp}
\usepackage[Symbol]{upgreek}

\newcommand{\safemath}[2]{\newcommand{#1}{\ensuremath{#2}\xspace}}

\newcommand{\ssa}{\mathsf{a}}
\newcommand{\ssb}{\mathsf{b}}
\newcommand{\ssc}{\mathsf{c}}
\newcommand{\ssd}{\mathsf{d}}
\newcommand{\sse}{\mathsf{e}}
\newcommand{\ssf}{\mathsf{f}}
\newcommand{\ssg}{\mathsf{g}}
\newcommand{\ssh}{\mathsf{h}}
\newcommand{\ssi}{\mathsf{i}}
\newcommand{\ssj}{\mathsf{j}}
\newcommand{\ssk}{\mathsf{k}}
\newcommand{\ssl}{\mathsf{l}}
\newcommand{\ssm}{\mathsf{m}}
\newcommand{\ssn}{\mathsf{n}}
\newcommand{\sso}{\mathsf{o}}
\newcommand{\ssp}{\mathsf{p}}
\newcommand{\ssq}{\mathsf{q}}
\newcommand{\ssr}{\mathsf{r}}
\newcommand{\sss}{\mathsf{s}}
\newcommand{\sst}{\mathsf{t}}
\newcommand{\ssu}{\mathsf{u}}
\newcommand{\ssv}{\mathsf{v}}
\newcommand{\ssw}{\mathsf{w}}
\newcommand{\ssx}{\mathsf{x}}
\newcommand{\ssy}{\mathsf{y}}
\newcommand{\ssz}{\mathsf{z}}

\safemath{\bmsa}{\bm{\ssa}}
\safemath{\bmsb}{\bm{\ssb}}
\safemath{\bmsc}{\bm{\ssc}}
\safemath{\bmsd}{\bm{\ssd}}
\safemath{\bmse}{\bm{\sse}}
\safemath{\bmsf}{\bm{\ssf}}
\safemath{\bmsg}{\bm{\ssg}}
\safemath{\bmsh}{\bm{\ssh}}
\safemath{\bmsi}{\bm{\ssi}}
\safemath{\bmsj}{\bm{\ssj}}
\safemath{\bmsk}{\bm{\ssk}}
\safemath{\bmsl}{\bm{\ssl}}
\safemath{\bmsm}{\bm{\ssm}}
\safemath{\bmsn}{\bm{\ssn}}
\safemath{\bmso}{\bm{\sso}}
\safemath{\bmsp}{\bm{\ssp}}
\safemath{\bmsq}{\bm{\ssq}}
\safemath{\bmsr}{\bm{\ssr}}
\safemath{\bmss}{\bm{\sss}}
\safemath{\bmst}{\bm{\sst}}
\safemath{\bmsu}{\bm{\ssu}}
\safemath{\bmsv}{\bm{\ssv}}
\safemath{\bmsw}{\bm{\ssw}}
\safemath{\bmsx}{\bm{\ssx}}
\safemath{\bmsy}{\bm{\ssy}}
\safemath{\bmsz}{\bm{\ssz}}

\bmdefine{\bmualphad}{\upalpha}
\bmdefine{\bmubetad}{\upbeta}
\bmdefine{\bmuchid}{\upchi}
\bmdefine{\bmudeltad}{\updelta}
\bmdefine{\bmuepsilond}{\upepsilon}
\bmdefine{\bmuvarepsilond}{\upvarepsilon}
\bmdefine{\bmuetad}{\upeta}
\bmdefine{\bmugammad}{\upgamma}
\bmdefine{\bmuiotad}{\upiota}
\bmdefine{\bmukappad}{\upkappa}
\bmdefine{\bmulambdad}{\uplambda}
\bmdefine{\bmumud}{\upmu}
\bmdefine{\bmunud}{\upnu}
\bmdefine{\bmuomegad}{\upomega}
\bmdefine{\bmuphid}{\upphi}
\bmdefine{\bmuvarphid}{\upvarphi}
\bmdefine{\bmupid}{\uppi}
\bmdefine{\bmuvarpid}{\upvarpi}
\bmdefine{\bmupsid}{\uppsi}
\bmdefine{\bmurhod}{\uprho}
\bmdefine{\bmuvarrhod}{\upvarrho}
\bmdefine{\bmusigmad}{\upsigma}
\bmdefine{\bmuvarsigmad}{\upvarsigma}
\bmdefine{\bmutaud}{\uptau}
\bmdefine{\bmuthetad}{\uptheta}
\bmdefine{\bmuvarthetad}{\upvartheta}
\bmdefine{\bmuupsilond}{\upupsilon}
\bmdefine{\bmuxid}{\upxi}
\bmdefine{\bmuzetad}{\upzeta}

\safemath{\bmua}{\mathbf{a}}
\safemath{\bmub}{\mathbf{b}}
\safemath{\bmuc}{\mathbf{c}}
\safemath{\bmud}{\mathbf{d}}
\safemath{\bmue}{\mathbf{e}}
\safemath{\bmuf}{\mathbf{f}}
\safemath{\bmug}{\mathbf{g}}
\safemath{\bmuh}{\mathbf{h}}
\safemath{\bmui}{\mathbf{i}}
\safemath{\bmuj}{\mathbf{j}}
\safemath{\bmuk}{\mathbf{k}}
\safemath{\bmul}{\mathbf{l}}
\safemath{\bmum}{\mathbf{m}}
\safemath{\bmun}{\mathbf{n}}
\safemath{\bmuo}{\mathbf{o}}
\safemath{\bmup}{\mathbf{p}}
\safemath{\bmuq}{\mathbf{q}}
\safemath{\bmur}{\mathbf{r}}
\safemath{\bmus}{\mathbf{s}}
\safemath{\bmut}{\mathbf{t}}
\safemath{\bmuu}{\mathbf{u}}
\safemath{\bmuv}{\mathbf{v}}
\safemath{\bmuw}{\mathbf{w}}
\safemath{\bmux}{\mathbf{x}}
\safemath{\bmuy}{\mathbf{y}}
\safemath{\bmuz}{\mathbf{z}}

\safemath{\bmualpha}{\bmualphad}
\safemath{\bmubeta}{\bmubetad}
\safemath{\bmuchi}{\bumchid}
\safemath{\bmudelta}{\bmudeltad}
\safemath{\bmuepsilon}{\bmuepsilond}
\safemath{\bmuvarepsilon}{\bmuvarepsilond}
\safemath{\bmueta}{\bmuetad}
\safemath{\bmugamma}{\bmugammad}
\safemath{\bmuiota}{\bmuiotad}
\safemath{\bmukappa}{\bmukappad}
\safemath{\bmulambda}{\bmulambdad}
\safemath{\bmumu}{\bmumud}
\safemath{\bmunu}{\bmunud}
\safemath{\bmuomega}{\bmuomegad}
\safemath{\bmuphi}{\bmuphid}
\safemath{\bmuvarphi}{\bmuvarphid}
\safemath{\bmupi}{\bmupid}
\safemath{\bmuvarpi}{\bmuvarpid}
\safemath{\bmupsi}{\bmupsid}
\safemath{\bmurho}{\bmurhod}
\safemath{\bmuvarrho}{\bmuvarrhod}
\safemath{\bmusigma}{\bmusigmad}
\safemath{\bmuvarsigma}{\bmuvarsigmad}
\safemath{\bmutau}{\bmutaud}
\safemath{\bmutheta}{\bmuthetad}
\safemath{\bmuvartheta}{\bmuvarthetad}
\safemath{\bmuupsilon}{\bmuupsilond}
\safemath{\bmuxi}{\bmuxid}
\safemath{\bmuzeta}{\bmuzetad}

\bmdefine{\bmiad}{a}
\bmdefine{\bmibd}{b}
\bmdefine{\bmicd}{c}
\bmdefine{\bmidd}{d}
\bmdefine{\bmied}{e}
\bmdefine{\bmifd}{f}
\bmdefine{\bmigd}{g}
\bmdefine{\bmihd}{h}
\bmdefine{\bmiid}{i}
\bmdefine{\bmijd}{j}
\bmdefine{\bmikd}{k}
\bmdefine{\bmild}{l}
\bmdefine{\bmimd}{m}
\bmdefine{\bmind}{n}
\bmdefine{\bmiod}{o}
\bmdefine{\bmipd}{p}
\bmdefine{\bmiqd}{q}
\bmdefine{\bmird}{r}
\bmdefine{\bmisd}{s}
\bmdefine{\bmitd}{t}
\bmdefine{\bmiud}{u}
\bmdefine{\bmivd}{v}
\bmdefine{\bmiwd}{w}
\bmdefine{\bmixd}{x}
\bmdefine{\bmiyd}{y}
\bmdefine{\bmizd}{z}

\bmdefine{\bmialphad}{\alpha}
\bmdefine{\bmibetad}{\beta}
\bmdefine{\bmichid}{\chi}
\bmdefine{\bmideltad}{\delta}
\bmdefine{\bmiepsilond}{\epsilon}
\bmdefine{\bmivarepsilond}{\varepsilon}
\bmdefine{\bmietad}{\eta}
\bmdefine{\bmigammad}{\gamma}
\bmdefine{\bmiiotad}{\iota}
\bmdefine{\bmikappad}{\kappa}
\bmdefine{\bmivarkappad}{\varkappa}
\bmdefine{\bmilambdad}{\lambda}
\bmdefine{\bmimud}{\mu}
\bmdefine{\bminud}{\nu}
\bmdefine{\bmiomegad}{\omega}
\bmdefine{\bmiphid}{\phi}
\bmdefine{\bmivarphid}{\varphi}
\bmdefine{\bmipid}{\pi}
\bmdefine{\bmivarpid}{\varpi}
\bmdefine{\bmipsid}{\psi}
\bmdefine{\bmirhod}{\rho}
\bmdefine{\bmivarrhod}{\varrho}
\bmdefine{\bmisigmad}{\sigma}
\bmdefine{\bmivarsigmad}{\varsigma}
\bmdefine{\bmitaud}{\tau}
\bmdefine{\bmithetad}{\theta}
\bmdefine{\bmivarthetad}{\vartheta}
\bmdefine{\bmiupsilond}{\upsilon}
\bmdefine{\bmixid}{\xi}
\bmdefine{\bmizetad}{\zeta}

\safemath{\bmia}{\bmiad}
\safemath{\bmib}{\bmibd}
\safemath{\bmic}{\bmicd}
\safemath{\bmid}{\bmidd}
\safemath{\bmie}{\bmied}
\safemath{\bmif}{\bmifd}
\safemath{\bmig}{\bmigd}
\safemath{\bmih}{\bmihd}
\safemath{\bmii}{\bmiid}
\safemath{\bmij}{\bmijd}
\safemath{\bmik}{\bmikd}
\safemath{\bmil}{\bmild}
\safemath{\bmim}{\bmimd}
\safemath{\bmin}{\bmind}
\safemath{\bmio}{\bmiod}
\safemath{\bmip}{\bmipd}
\safemath{\bmiq}{\bmiqd}
\safemath{\bmir}{\bmird}
\safemath{\bmis}{\bmisd}
\safemath{\bmit}{\bmitd}
\safemath{\bmiu}{\bmiud}
\safemath{\bmiv}{\bmivd}
\safemath{\bmiw}{\bmiwd}
\safemath{\bmix}{\bmixd}
\safemath{\bmiy}{\bmiyd}
\safemath{\bmiz}{\bmizd}

\safemath{\bmialpha}{\bmialphad}
\safemath{\bmibeta}{\bmibetad}
\safemath{\bmichi}{\bmichid}
\safemath{\bmidelta}{\bmideltad}
\safemath{\bmiepsilon}{\bmiepsilond}
\safemath{\bmivarepsilon}{\bmivarepsilond}
\safemath{\bmieta}{\bmietad}
\safemath{\bmigamma}{\bmigammad}
\safemath{\bmiiota}{\bmiiotad}
\safemath{\bmikappa}{\bmikappad}
\safemath{\bmivarkappa}{\bmivarkappad}
\safemath{\bmilambda}{\bmilambdad}
\safemath{\bmimu}{\bmimud}
\safemath{\bminu}{\bminud}
\safemath{\bmiomega}{\bmiomegad}
\safemath{\bmiphi}{\bmiphid}
\safemath{\bmivarphi}{\bmivarphid}
\safemath{\bmipi}{\bmipid}
\safemath{\bmivarpi}{\bmivarpid}
\safemath{\bmipsi}{\bmipsid}
\safemath{\bmirho}{\bmirhod}
\safemath{\bmivarrho}{\bmivarrhod}
\safemath{\bmisigma}{\bmisigmad}
\safemath{\bmivarsigma}{\bmivarsigmad}
\safemath{\bmitau}{\bmitaud}
\safemath{\bmitheta}{\bmithetad}
\safemath{\bmivartheta}{\bmivarthetad}
\safemath{\bmiupsilon}{\bmiupsilond}
\safemath{\bmixi}{\bmixid}
\safemath{\bmizeta}{\bmizetad}

\bmdefine{\bmuDeltad}{\Updelta}
\bmdefine{\bmuGammad}{\Upgamma}
\bmdefine{\bmuLambdad}{\Uplambda}
\bmdefine{\bmuOmegad}{\Upomega}
\bmdefine{\bmuPhid}{\Upphi}
\bmdefine{\bmuPid}{\Uppi}
\bmdefine{\bmuPsid}{\Uppsi}
\bmdefine{\bmuSigmad}{\Upsigma}
\bmdefine{\bmuThetad}{\Uptheta}
\bmdefine{\bmuUpsilond}{\Upupsilon}
\bmdefine{\bmuXid}{\Upxi}

\safemath{\bmuA}{\mathbf{A}}
\safemath{\bmuB}{\mathbf{B}}
\safemath{\bmuC}{\mathbf{C}}
\safemath{\bmuD}{\mathbf{D}}
\safemath{\bmuE}{\mathbf{E}}
\safemath{\bmuF}{\mathbf{F}}
\safemath{\bmuG}{\mathbf{G}}
\safemath{\bmuH}{\mathbf{H}}
\safemath{\bmuI}{\mathbf{I}}
\safemath{\bmuJ}{\mathbf{J}}
\safemath{\bmuK}{\mathbf{K}}
\safemath{\bmuL}{\mathbf{L}}
\safemath{\bmuM}{\mathbf{M}}
\safemath{\bmuN}{\mathbf{N}}
\safemath{\bmuO}{\mathbf{O}}
\safemath{\bmuP}{\mathbf{P}}
\safemath{\bmuQ}{\mathbf{Q}}
\safemath{\bmuR}{\mathbf{R}}
\safemath{\bmuS}{\mathbf{S}}
\safemath{\bmuT}{\mathbf{T}}
\safemath{\bmuU}{\mathbf{U}}
\safemath{\bmuV}{\mathbf{V}}
\safemath{\bmuW}{\mathbf{W}}
\safemath{\bmuX}{\mathbf{X}}
\safemath{\bmuY}{\mathbf{Y}}
\safemath{\bmuZ}{\mathbf{Z}}

\safemath{\bmuZero}{\mathbf{0}}
\safemath{\bmuOne}{\mathbf{1}}

\safemath{\bmuDelta}{\bmuDeltad}
\safemath{\bmuGamma}{\bmuGammad}
\safemath{\bmuLambda}{\bmuLambdad}
\safemath{\bmuOmega}{\bmuOmegad}
\safemath{\bmuPhi}{\bmuPhid}
\safemath{\bmuPi}{\bmuPid}
\safemath{\bmuPsi}{\bmuPsid}
\safemath{\bmuSigma}{\bmuSigmad}
\safemath{\bmuTheta}{\bmuThetad}
\safemath{\bmuUpsilon}{\bmuUpsilond}
\safemath{\bmuXi}{\bmuXid}

\bmdefine{\bmiAd}{A}
\bmdefine{\bmiBd}{B}
\bmdefine{\bmiCd}{C}
\bmdefine{\bmiDd}{D}
\bmdefine{\bmiEd}{E}
\bmdefine{\bmiFd}{F}
\bmdefine{\bmiGd}{G}
\bmdefine{\bmiHd}{H}
\bmdefine{\bmiId}{I}
\bmdefine{\bmiJd}{J}
\bmdefine{\bmiKd}{K}
\bmdefine{\bmiLd}{L}
\bmdefine{\bmiMd}{M}
\bmdefine{\bmiOd}{N}
\bmdefine{\bmiPd}{O}
\bmdefine{\bmiQd}{P}
\bmdefine{\bmiRd}{R}
\bmdefine{\bmiSd}{S}
\bmdefine{\bmiTd}{T}
\bmdefine{\bmiUd}{U}
\bmdefine{\bmiVd}{V}
\bmdefine{\bmiWd}{W}
\bmdefine{\bmiXd}{X}
\bmdefine{\bmiYd}{Y}
\bmdefine{\bmiZd}{Z}

\bmdefine{\bmiDeltad}{\Delta}
\bmdefine{\bmiGammad}{\Gamma}
\bmdefine{\bmiLambdad}{\Lambda}
\bmdefine{\bmiOmegad}{\Omega}
\bmdefine{\bmiPhid}{\Phi}
\bmdefine{\bmiPid}{\Pi}
\bmdefine{\bmiPsid}{\Psi}
\bmdefine{\bmiSigmad}{\Sigma}
\bmdefine{\bmiThetad}{\Theta}
\bmdefine{\bmiUpsilond}{\Upsilon}
\bmdefine{\bmiXid}{\Xi}

\safemath{\bmiA}{\bmiAd}
\safemath{\bmiB}{\bmiBd}
\safemath{\bmiC}{\bmiCd}
\safemath{\bmiD}{\bmiDd}
\safemath{\bmiE}{\bmiEd}
\safemath{\bmiF}{\bmiFd}
\safemath{\bmiG}{\bmiGd}
\safemath{\bmiH}{\bmiHd}
\safemath{\bmiI}{\bmiId}
\safemath{\bmiJ}{\bmiJd}
\safemath{\bmiK}{\bmiKd}
\safemath{\bmiL}{\bmiLd}
\safemath{\bmiM}{\bmiMd}
\safemath{\bmiN}{\bmiNd}
\safemath{\bmiO}{\bmiOd}
\safemath{\bmiP}{\bmiPd}
\safemath{\bmiQ}{\bmiQd}
\safemath{\bmiR}{\bmiRd}
\safemath{\bmiS}{\bmiSd}
\safemath{\bmiT}{\bmiTd}
\safemath{\bmiU}{\bmiUd}
\safemath{\bmiV}{\bmiVd}
\safemath{\bmiW}{\bmiWd}
\safemath{\bmiX}{\bmiXd}
\safemath{\bmiY}{\bmiYd}
\safemath{\bmiZ}{\bmiZd}

\safemath{\bmiDelta}{\bmiDeltad}
\safemath{\bmiGamma}{\bmiGammad}
\safemath{\bmiLambda}{\bmiLambdad}
\safemath{\bmiOmega}{\bmiOmegad}
\safemath{\bmiPhi}{\bmiPhid}
\safemath{\bmiPi}{\bmiPid}
\safemath{\bmiPsi}{\bmiPsid}
\safemath{\bmiSigma}{\bmiSigmad}
\safemath{\bmiTheta}{\bmiThetad}
\safemath{\bmiUpsilon}{\bmiUpsilond}
\safemath{\bmiXi}{\bmiXid}

\safemath{\evA}{\mathcal{A}}
\safemath{\evB}{\mathcal{B}}
\safemath{\evC}{\mathcal{C}}
\safemath{\evD}{\mathcal{D}}
\safemath{\evE}{\mathcal{E}}
\safemath{\evF}{\mathcal{F}}
\safemath{\evG}{\mathcal{G}}
\safemath{\evH}{\mathcal{H}}
\safemath{\evI}{\mathcal{I}}
\safemath{\evJ}{\mathcal{J}}
\safemath{\evK}{\mathcal{K}}
\safemath{\evL}{\mathcal{L}}
\safemath{\evM}{\mathcal{M}}
\safemath{\evN}{\mathcal{N}}
\safemath{\evO}{\mathcal{O}}
\safemath{\evP}{\mathcal{P}}
\safemath{\evQ}{\mathcal{Q}}
\safemath{\evR}{\mathcal{R}}
\safemath{\evS}{\mathcal{S}}
\safemath{\evT}{\mathcal{T}}
\safemath{\evU}{\mathcal{U}}
\safemath{\evV}{\mathcal{V}}
\safemath{\evW}{\mathcal{W}}
\safemath{\evX}{\mathcal{X}}
\safemath{\evY}{\mathcal{Y}}
\safemath{\evZ}{\mathcal{Z}}

\safemath{\setA}{\mathcal{A}}
\safemath{\setB}{\mathcal{B}}
\safemath{\setC}{\mathcal{C}}
\safemath{\setD}{\mathcal{D}}
\safemath{\setE}{\mathcal{E}}
\safemath{\setF}{\mathcal{F}}
\safemath{\setG}{\mathcal{G}}
\safemath{\setH}{\mathcal{H}}
\safemath{\setI}{\mathcal{I}}
\safemath{\setJ}{\mathcal{J}}
\safemath{\setK}{\mathcal{K}}
\safemath{\setL}{\mathcal{L}}
\safemath{\setM}{\mathcal{M}}
\safemath{\setN}{\mathcal{N}}
\safemath{\setO}{\mathcal{O}}
\safemath{\setP}{\mathcal{P}}
\safemath{\setQ}{\mathcal{Q}}
\safemath{\setR}{\mathcal{R}}
\safemath{\setS}{\mathcal{S}}
\safemath{\setT}{\mathcal{T}}
\safemath{\setU}{\mathcal{U}}
\safemath{\setV}{\mathcal{V}}
\safemath{\setW}{\mathcal{W}}
\safemath{\setX}{\mathcal{X}}
\safemath{\setY}{\mathcal{Y}}
\safemath{\setZ}{\mathcal{Z}}
\safemath{\emptySet}{\varnothing}

\safemath{\colA}{\mathscr{A}}
\safemath{\colB}{\mathscr{B}}
\safemath{\colC}{\mathscr{C}}
\safemath{\colD}{\mathscr{D}}
\safemath{\colE}{\mathscr{E}}
\safemath{\colF}{\mathscr{F}}
\safemath{\colG}{\mathscr{G}}
\safemath{\colH}{\mathscr{H}}
\safemath{\colI}{\mathscr{I}}
\safemath{\colJ}{\mathscr{J}}
\safemath{\colK}{\mathscr{K}}
\safemath{\colL}{\mathscr{L}}
\safemath{\colM}{\mathscr{M}}
\safemath{\colN}{\mathscr{N}}
\safemath{\colO}{\mathscr{O}}
\safemath{\colP}{\mathscr{P}}
\safemath{\colQ}{\mathscr{Q}}
\safemath{\colR}{\mathscr{R}}
\safemath{\colS}{\mathscr{S}}
\safemath{\colT}{\mathscr{T}}
\safemath{\colU}{\mathscr{U}}
\safemath{\colV}{\mathscr{V}}
\safemath{\colW}{\mathscr{W}}
\safemath{\colX}{\mathscr{X}}
\safemath{\colY}{\mathscr{Y}}
\safemath{\colZ}{\mathscr{Z}}

\safemath{\opA}{\mathbb{A}}
\safemath{\opB}{\mathbb{B}}
\safemath{\opC}{\mathbb{C}}
\safemath{\opD}{\mathbb{D}}
\safemath{\opE}{\mathbb{E}}
\safemath{\opF}{\mathbb{F}}
\safemath{\opG}{\mathbb{G}}
\safemath{\opH}{\mathbb{H}}
\safemath{\opI}{\mathbb{I}}
\safemath{\opJ}{\mathbb{J}}
\safemath{\opK}{\mathbb{K}}
\safemath{\opL}{\mathbb{L}}
\safemath{\opM}{\mathbb{M}}
\safemath{\opN}{\mathbb{N}}
\safemath{\opO}{\mathbb{O}}
\safemath{\opP}{\mathbb{P}}
\safemath{\opQ}{\mathbb{Q}}
\safemath{\opR}{\mathbb{R}}
\safemath{\opS}{\mathbb{S}}
\safemath{\opT}{\mathbb{T}}
\safemath{\opU}{\mathbb{U}}
\safemath{\opV}{\mathbb{V}}
\safemath{\opW}{\mathbb{W}}
\safemath{\opX}{\mathbb{X}}
\safemath{\opY}{\mathbb{Y}}
\safemath{\opZ}{\mathbb{Z}}
\safemath{\opZero}{\mathbb{O}}
\safemath{\identityop}{\opI}

\safemath{\sca}{a}
\safemath{\scb}{b}
\safemath{\scc}{c}
\safemath{\scd}{d}
\safemath{\sce}{e}
\safemath{\scf}{f}
\safemath{\scg}{g}
\safemath{\sch}{h}
\safemath{\sci}{i}
\safemath{\scj}{j}
\safemath{\sck}{k}
\safemath{\scl}{l}
\safemath{\scm}{m}
\safemath{\scn}{n}
\safemath{\sco}{o}
\safemath{\scp}{p}
\safemath{\scq}{q}
\safemath{\scr}{r}
\safemath{\scs}{s}
\safemath{\sct}{t}
\safemath{\scu}{u}
\safemath{\scv}{v}
\safemath{\scw}{w}
\safemath{\scx}{x}
\safemath{\scy}{y}
\safemath{\scz}{z}

\safemath{\scA}{A}
\safemath{\scB}{B}
\safemath{\scC}{C}
\safemath{\scD}{D}
\safemath{\scE}{E}
\safemath{\scF}{F}
\safemath{\scG}{G}
\safemath{\scH}{H}
\safemath{\scI}{I}
\safemath{\scJ}{J}
\safemath{\scK}{K}
\safemath{\scL}{L}
\safemath{\scM}{M}
\safemath{\scN}{N}
\safemath{\scO}{O}
\safemath{\scP}{P}
\safemath{\scQ}{Q}
\safemath{\scR}{R}
\safemath{\scS}{S}
\safemath{\scT}{T}
\safemath{\scU}{U}
\safemath{\scV}{V}
\safemath{\scW}{W}
\safemath{\scX}{X}
\safemath{\scY}{Y}
\safemath{\scZ}{Z}

\safemath{\scalpha}{\alpha}
\safemath{\scbeta}{\beta}
\safemath{\scchi}{\chi}
\safemath{\scdelta}{\delta}
\safemath{\scepsilon}{\epsilon}
\safemath{\scvarepsilon}{\varepsilon}
\safemath{\sceta}{\eta}
\safemath{\scgamma}{\gamma}
\safemath{\sciota}{\iota}
\safemath{\sckappa}{\kappa}
\safemath{\scvarkappa}{\varkappa}
\safemath{\sclambda}{\lambda}
\safemath{\scmu}{\mu}
\safemath{\scnu}{\nu}
\safemath{\scomega}{\omega}
\safemath{\scphi}{\phi}
\safemath{\scvarphi}{\varphi}
\safemath{\scpi}{\pi}
\safemath{\scvarpi}{\varpi}
\safemath{\scpsi}{\psi}
\safemath{\scrho}{\rho}
\safemath{\scvarrho}{\varrho}
\safemath{\scsigma}{\sigma}
\safemath{\scvarsigma}{\varsigma}
\safemath{\sctau}{\tau}
\safemath{\sctheta}{\theta}
\safemath{\scvartheta}{\vartheta}
\safemath{\scupsilon}{\upsilon}
\safemath{\scxi}{\xi}
\safemath{\sczeta}{\zeta}

\safemath{\veca}{\mathbf{a}}
\safemath{\vecb}{\mathbf{b}}
\safemath{\vecc}{\mathbf{c}}
\safemath{\vecd}{\mathbf{d}}
\safemath{\vece}{\mathbf{e}}
\safemath{\vecf}{\mathbf{f}}
\safemath{\vecg}{\mathbf{g}}
\safemath{\vech}{\mathbf{h}}
\safemath{\veci}{\mathbf{i}}
\safemath{\vecj}{\mathbf{j}}
\safemath{\veck}{\mathbf{k}}
\safemath{\vecl}{\mathbf{l}}
\safemath{\vecm}{\mathbf{m}}
\safemath{\vecn}{\mathbf{n}}
\safemath{\veco}{\mathbf{o}}
\safemath{\vecp}{\mathbf{p}}
\safemath{\vecq}{\mathbf{q}}
\safemath{\vecr}{\mathbf{r}}
\safemath{\vecs}{\mathbf{s}}
\safemath{\vect}{\mathbf{t}}
\safemath{\vecu}{\mathbf{u}}
\safemath{\vecv}{\mathbf{v}}
\safemath{\vecw}{\mathbf{w}}
\safemath{\vecx}{\mathbf{x}}
\safemath{\vecy}{\mathbf{y}}
\safemath{\vecz}{\mathbf{z}}
\safemath{\veczero}{\mathbf{0}}
\safemath{\vecone}{\mathbf{1}}

\safemath{\vecalpha}{\upalpha}
\safemath{\vecbeta}{\upbeta}
\safemath{\vecchi}{\upchi}
\safemath{\vecdelta}{\updelta}
\safemath{\vecepsilon}{\upepsilon}
\safemath{\vecvarepsilon}{\upvarepsilon}
\safemath{\veceta}{\upeta}
\safemath{\vecgamma}{\upgamma}
\safemath{\veciota}{\upiota}
\safemath{\veckappa}{\upkappa}
\safemath{\veclambda}{\uplambda}
\safemath{\vecmu}{\text{\textmu}}
\safemath{\vecnu}{\upnu}
\safemath{\vecomega}{\upomega}
\safemath{\vecphi}{\upphi}
\safemath{\vecvarphi}{\upvarphi}
\safemath{\vecpi}{\uppi}
\safemath{\vecvarpi}{\upvarpi}
\safemath{\vecpsi}{\uppsi}
\safemath{\vecrho}{\uprho}
\safemath{\vecvarrho}{\upvarrho}
\safemath{\vecsigma}{\upsigma}
\safemath{\vecvarsigma}{\upvarsigma}
\safemath{\vectau}{\uptau}
\safemath{\vectheta}{\uptheta}
\safemath{\vecvartheta}{\upvartheta}
\safemath{\vecupsilon}{\upupsilon}
\safemath{\vecxi}{\upxi}
\safemath{\veczeta}{\upzeta}

\safemath{\vecac}{a}
\safemath{\vecbc}{b}
\safemath{\veccc}{c}
\safemath{\vecdc}{d}
\safemath{\vecec}{e}
\safemath{\vecfc}{f}
\safemath{\vecgc}{g}
\safemath{\vechc}{h}
\safemath{\vecic}{i}
\safemath{\vecjc}{j}
\safemath{\veckc}{k}
\safemath{\veclc}{l}
\safemath{\vecmc}{m}
\safemath{\vecnc}{n}
\safemath{\vecoc}{o}
\safemath{\vecpc}{p}
\safemath{\vecqc}{q}
\safemath{\vecrc}{r}
\safemath{\vecsc}{s}
\safemath{\vectc}{t}
\safemath{\vecuc}{u}
\safemath{\vecvc}{v}
\safemath{\vecwc}{w}
\safemath{\vecxc}{x}
\safemath{\vecyc}{y}
\safemath{\veczc}{z}

\safemath{\matA}{\mathbf{A}}
\safemath{\matB}{\mathbf{B}}
\safemath{\matC}{\mathbf{C}}
\safemath{\matD}{\mathbf{D}}
\safemath{\matE}{\mathbf{E}}
\safemath{\matF}{\mathbf{F}}
\safemath{\matG}{\mathbf{G}}
\safemath{\matH}{\mathbf{H}}
\safemath{\matI}{\mathbf{I}}
\safemath{\matJ}{\mathbf{J}}
\safemath{\matK}{\mathbf{K}}
\safemath{\matL}{\mathbf{L}}
\safemath{\matM}{\mathbf{M}}
\safemath{\matN}{\mathbf{N}}
\safemath{\matO}{\mathbf{O}}
\safemath{\matP}{\mathbf{P}}
\safemath{\matQ}{\mathbf{Q}}
\safemath{\matR}{\mathbf{R}}
\safemath{\matS}{\mathbf{S}}
\safemath{\matT}{\mathbf{T}}
\safemath{\matU}{\mathbf{U}}
\safemath{\matV}{\mathbf{V}}
\safemath{\matW}{\mathbf{W}}
\safemath{\matX}{\mathbf{X}}
\safemath{\matY}{\mathbf{Y}}
\safemath{\matZ}{\mathbf{Z}}
\safemath{\matzero}{\mathbf{0}}

\safemath{\matDelta}{\Updelta}
\safemath{\matGamma}{\Upgammma}
\safemath{\matLambda}{\Uplambda}
\safemath{\matOmega}{\Upomega}
\safemath{\matPhi}{\Upphi}
\safemath{\matPi}{\Uppi}
\safemath{\matPsi}{\Uppsi}
\safemath{\matSigma}{\Upsigma}
\safemath{\matTheta}{\Uptheta}
\safemath{\matUpsilon}{\Upupsilon}
\safemath{\matXi}{\Upxi}

\safemath{\matidentity}{\matI}
\safemath{\vecunit}{\vece} 
\safemath{\matone}{\matO}

\safemath{\matAc}{a}
\safemath{\matBc}{b}
\safemath{\matCc}{c}
\safemath{\matDc}{d}
\safemath{\matEc}{e}
\safemath{\matFc}{f}
\safemath{\matGc}{g}
\safemath{\matHc}{h}
\safemath{\matIc}{i}
\safemath{\matJc}{j}
\safemath{\matKc}{k}
\safemath{\matLc}{l}
\safemath{\matMc}{m}
\safemath{\matNc}{n}
\safemath{\matOc}{o}
\safemath{\matPc}{p}
\safemath{\matQc}{q}
\safemath{\matRc}{r}
\safemath{\matSc}{s}
\safemath{\matTc}{t}
\safemath{\matUc}{u}
\safemath{\matVc}{v}
\safemath{\matWc}{w}
\safemath{\matXc}{x}
\safemath{\matYc}{y}
\safemath{\matZc}{z}

\safemath{\rnda}{\bmia}
\safemath{\rndb}{\bmib}
\safemath{\rndc}{\bmic}
\safemath{\rndd}{\bmid}
\safemath{\rnde}{\bmie}
\safemath{\rndf}{\bmif}
\safemath{\rndg}{\bmig}
\safemath{\rndh}{\bmih}
\safemath{\rndi}{\bmii}
\safemath{\rndj}{\bmij}
\safemath{\rndk}{\bmik}
\safemath{\rndl}{\bmil}
\safemath{\rndm}{\bmim}
\safemath{\rndn}{\bmin}
\safemath{\rndo}{\bmio}
\safemath{\rndp}{\bmip}
\safemath{\rndq}{\bmiq}
\safemath{\rndr}{\bmir}
\safemath{\rnds}{\bmis}
\safemath{\rndt}{\bmit}
\safemath{\rndu}{\bmiu}
\safemath{\rndv}{\bmiv}
\safemath{\rndw}{\bmiw}
\safemath{\rndx}{\bmix}
\safemath{\rndy}{\bmiy}
\safemath{\rndz}{\bmiz}

\safemath{\rndA}{\bmiA}
\safemath{\rndB}{\bmiB}
\safemath{\rndC}{\bmiC}
\safemath{\rndD}{\bmiD}
\safemath{\rndE}{\bmiE}
\safemath{\rndF}{\bmiF}
\safemath{\rndG}{\bmiG}
\safemath{\rndH}{\bmiH}
\safemath{\rndI}{\bmiI}
\safemath{\rndJ}{\bmiJ}
\safemath{\rndK}{\bmiK}
\safemath{\rndL}{\bmiL}
\safemath{\rndM}{\bmiM}
\safemath{\rndN}{\bmiN}
\safemath{\rndO}{\bmiO}
\safemath{\rndP}{\bmiP}
\safemath{\rndQ}{\bmiQ}
\safemath{\rndR}{\bmiR}
\safemath{\rndS}{\bmiS}
\safemath{\rndT}{\bmiT}
\safemath{\rndU}{\bmiU}
\safemath{\rndV}{\bmiV}
\safemath{\rndW}{\bmiW}
\safemath{\rndX}{\bmiX}
\safemath{\rndY}{\bmiY}
\safemath{\rndZ}{\bmiZ}

\safemath{\rndalpha}{\bmialpha}
\safemath{\rndbeta}{\bmibeta}
\safemath{\rndchi}{\bmichi}
\safemath{\rnddelta}{\bmidelta}
\safemath{\rndepsilon}{\bmiepsilon}
\safemath{\rndvarepsilon}{\bmivarepsilon}
\safemath{\rndeta}{\bmieta}
\safemath{\rndgamma}{\bmigamma}
\safemath{\rndiota}{\bmiiota}
\safemath{\rndkappa}{\bmikappa}
\safemath{\rndlambda}{\bmilambda}
\safemath{\rndmu}{\bmimu}
\safemath{\rndnu}{\bminu}
\safemath{\rndomega}{\bmiomega}
\safemath{\rndphi}{\bmiphi}
\safemath{\rndvarphi}{\bmivarphi}
\safemath{\rndpi}{\bmipi}
\safemath{\rndvarpi}{\bmivarpi}
\safemath{\rndpsi}{\bmipsi}
\safemath{\rndrho}{\bmirho}
\safemath{\rndvarrho}{\bmivarrho}
\safemath{\rndsigma}{\bmisigma}
\safemath{\rndvarsigma}{\bmivarsigma}
\safemath{\rndtau}{\bmitau}
\safemath{\rndtheta}{\bmitheta}
\safemath{\rndvartheta}{\bmivartheta}
\safemath{\rndupsilon}{\bmiupsilon}
\safemath{\rndxi}{\bmixi}
\safemath{\rndzeta}{\bmizeta}

\safemath{\rveca}{\bmua}
\safemath{\rvecb}{\bmub}
\safemath{\rvecc}{\bmuc}
\safemath{\rvecd}{\bmud}
\safemath{\rvece}{\bmue}
\safemath{\rvecf}{\bmuf}
\safemath{\rvecg}{\bmug}
\safemath{\rvech}{\bmuh}
\safemath{\rveci}{\bmui}
\safemath{\rvecj}{\bmuj}
\safemath{\rveck}{\bmuk}
\safemath{\rvecl}{\bmul}
\safemath{\rvecm}{\bmum}
\safemath{\rvecn}{\bmun}
\safemath{\rveco}{\bmuo}
\safemath{\rvecp}{\bmup}
\safemath{\rvecq}{\bmuq}
\safemath{\rvecr}{\bmur}
\safemath{\rvecs}{\bmus}
\safemath{\rvect}{\bmut}
\safemath{\rvecu}{\bmuu}
\safemath{\rvecv}{\bmuv}
\safemath{\rvecw}{\bmuw}
\safemath{\rvecx}{\bmux}
\safemath{\rvecy}{\bmuy}
\safemath{\rvecz}{\bmuz}

\safemath{\rvecalpha}{\bmualpha}
\safemath{\rvecbeta}{\bmubeta}
\safemath{\rvecchi}{\bmuchi}
\safemath{\rvecdelta}{\bmudelta}
\safemath{\rvecepsilon}{\bmuepsilon}
\safemath{\rvecvarepsilon}{\bmuvarepsilon}
\safemath{\rveceta}{\bmueta}
\safemath{\rvecgamma}{\bmugamma}
\safemath{\rveciota}{\bmuiota}
\safemath{\rveckappa}{\bmukappa}
\safemath{\rveclambda}{\bmulambda}
\safemath{\rvecmu}{\bmumu}
\safemath{\rvecnu}{\bmunu}
\safemath{\rvecomega}{\bmuomega}
\safemath{\rvecphi}{\bmuphi}
\safemath{\rvecvarphi}{\bmuvarphi}
\safemath{\rvecpi}{\bmupi}
\safemath{\rvecvarpi}{\bmuvarpi}
\safemath{\rvecpsi}{\bmupsi}
\safemath{\rvecrho}{\bmurho}
\safemath{\rvecvarrho}{\bmuvarrho}
\safemath{\rvecsigma}{\bmusigma}
\safemath{\rvecvarsigma}{\bmuvarsigma}
\safemath{\rvectau}{\bmutau}
\safemath{\rvectheta}{\bmutheta}
\safemath{\rvecvartheta}{\bmuvartheta}
\safemath{\rvecupsilon}{\bmuupsilon}
\safemath{\rvecxi}{\bmuxi}
\safemath{\rveczeta}{\bmuzeta}

\safemath{\rmatA}{\bmuA}
\safemath{\rmatB}{\bmuB}
\safemath{\rmatC}{\bmuC}
\safemath{\rmatD}{\bmuD}
\safemath{\rmatE}{\bmuE}
\safemath{\rmatF}{\bmuF}
\safemath{\rmatG}{\bmuG}
\safemath{\rmatH}{\bmuH}
\safemath{\rmatI}{\bmuI}
\safemath{\rmatJ}{\bmuJ}
\safemath{\rmatK}{\bmuK}
\safemath{\rmatL}{\bmuL}
\safemath{\rmatM}{\bmuM}
\safemath{\rmatN}{\bmuN}
\safemath{\rmatO}{\bmuO}
\safemath{\rmatP}{\bmuP}
\safemath{\rmatQ}{\bmuQ}
\safemath{\rmatR}{\bmuR}
\safemath{\rmatS}{\bmuS}
\safemath{\rmatT}{\bmuT}
\safemath{\rmatU}{\bmuU}
\safemath{\rmatV}{\bmuV}
\safemath{\rmatW}{\bmuW}
\safemath{\rmatX}{\bmuX}
\safemath{\rmatY}{\bmuY}
\safemath{\rmatZ}{\bmuZ}

\safemath{\rmatDelta}{\bmuDelta}
\safemath{\rmatGamma}{\bmuGamma}
\safemath{\rmatLambda}{\bmuLambda}
\safemath{\rmatOmega}{\bmuOmega}
\safemath{\rmatPhi}{\bmuPhi}
\safemath{\rmatPi}{\bmuPi}
\safemath{\rmatPsi}{\bmuPsi}
\safemath{\rmatSigma}{\bmuSigma}
\safemath{\rmatTheta}{\bmuTheta}
\safemath{\rmatUpsilon}{\bmuUpsilon}
\safemath{\rmatXi}{\bmuXi}

\newenvironment{textbmatrix}{    \setlength{\arraycolsep}{2.5pt}%
    \big[\begin{matrix}}{\end{matrix}\big]%
    \raisebox{0.08ex}{\vphantom{M}}}

\def\btm{\begin{textbmatrix}}
\def\etm{\end{textbmatrix}}

\newcommand{\lefto}{\mathopen{}\left}

\newcommand{\dummyrel}[1]{\mathrel{\hphantom{#1}}\strut\mskip-\medmuskip}

\DeclareMathOperator{\diag}{diag}            
\DeclareMathOperator{\adj}{adj}              
\DeclareMathOperator{\sign}{sign}            
\DeclareMathOperator{\sinc}{sinc}            
\DeclareMathOperator*{\argmin}{arg\;min}     
\DeclareMathOperator{\conv}{\star}            

\safemath{\fun}{f}                        
\safemath{\altfun}{\scg}                        
\safemath{\vectr}{\veca}         
\safemath{\vectrc}{\vecac}         

\safemath{\vrbl}{t}                        
\safemath{\altvrbl}{y}                        
\safemath{\aaltvrbl}{z}                        

\safemath{\vvrbl}{\vecx}                        
\safemath{\altvvrbl}{\vecy}                        
\safemath{\aaltvvrbl}{\vecz}                        

\safemath{\aaltfun}{\sch}
\safemath{\bel}{\sce}                    
\safemath{\altbel}{\sce}                    
\safemath{\frel}{g}                    
\safemath{\altfrel}{g}                    
\safemath{\dfrel}{\tilde{g}}                    
\safemath{\altdfrel}{\tilde{g}}                    

\safemath{\mat}{\matA}                        
\safemath{\matc}{\matAc}                        
\safemath{\altmat}{\matB}                        
\safemath{\altmatc}{\matBc}                        

\safemath{\altvectr}{\vecv}                        
\safemath{\altvectrc}{\vecvc}                        

\newcommand{\nullspace}{\setN}                 

\newcommand{\abs}[1]{\mathchoice{{\left\lvert#1\right\rvert}}{{\bigl\lvert#1\bigr\rvert}}{{\left\lvert#1\right\rvert}}{{\left\lvert#1\right\rvert}}}        

\newcommand{\union}{\cup}                    

\newcommand{\intersect}{\cap}                

\newcommand{\twonorm}[1]{\lVert#1\rVert_2}        
\newcommand{\onenorm}[1]{\lVert#1\rVert_1}        
\newcommand{\infnorm}[1]{\lVert#1\rVert_{\infty}}        
\newcommand{\zeronorm}[1]{\lVert#1\rVert_0}        
\newcommand{\vecnorm}[1]{\lVert#1\rVert}        
\newcommand{\conj}[1]{\ensuremath{#1^{*}}}     
\newcommand{\tp}[1]{\ensuremath{#1^{\mathsf{T}}}}         
\newcommand{\herm}[1]{\ensuremath{#1^{\mathsf{H}}}}     
\newcommand{\inv}[1]{\ensuremath{#1^{-1}}}     

\safemath{\dirac}{\delta}                    
\safemath{\diracp}{\dirac(\time)}            
\safemath{\krond}{\dirac}                    
\safemath{\indfun}{I}                        
\safemath{\stepfun}{u}                        

\safemath{\upto}{\uparrow}
\safemath{\downto}{\downarrow}
\safemath{\iu}{\mathrm{i}}                            
\safemath{\maj}{\succ}
\newcommand{\dftmat}[1]{\matF_{#1}}            
\safemath{\mdft}{\dftmat{}}                    
\safemath{\runity}{\beta}                    
\safemath{\eval}{\lambda}                    

\safemath{\veval}{\veclambda}                
\safemath{\littleo}{\sco}                    

\let\im\undefined
\safemath{\re}{\Re}                
\safemath{\im}{\Im}                

\safemath{\euclidspace}{\complexset}            
\safemath{\confunspace}{\setC}                
\newcommand{\banachseqspace}[1]{l^{#1}}        
\safemath{\hilseqspace}{\banachseqspace{2}}    
\newcommand{\banachfunspace}[1]{\setL^{#1}}    
\safemath{\hilfunspace}{\banachfunspace{2}}    
\safemath{\hilfunspacep}{\hilfunspace(\complexset)}    
\safemath{\schwarzspace}{\setS}                

\newcommand{\hadj}[1]{#1^{\star}}            

\safemath{\SNR}{\rho}                 
\safemath{\SINR}{\text{\sc sinr}}                 
\safemath{\No}{N_0}                            
\safemath{\Es}{E_s}                            
\safemath{\Eb}{E_b}                            
\safemath{\EbNo}{\frac{\Eb}{\No}}
\safemath{\EsNo}{\frac{\Es}{\No}}
\safemath{\NoVar}{\variance}                 

\let\time\undefined
\safemath{\time}{\sct}                        
\safemath{\dtime}{\sck}                        
\safemath{\delay}{\sctau}                    
\safemath{\ddelay}{\scl}                        
\safemath{\doppler}{\scnu}                    
\safemath{\ddoppler}{\scm}                    
\safemath{\freq}{\scf}                        
\safemath{\dfreq}{\scn}                        
\safemath{\Dtime}{\Delta\time}
\safemath{\Dfreq}{\Delta\freq}
\safemath{\Ddtime}{\dtime}
\safemath{\Ddfreq}{\dfreq}
\safemath{\bandwidth}{\scB}
\safemath{\maxdoppler}{\doppler_{0}}            
\safemath{\maxdelay}{\delay_{0}}                
\safemath{\spread}{\Delta_{\CHop}}            
\DeclareMathOperator{\CHop}{\ensuremath{\opH}} 
\safemath{\kernel}{\rndk_{\CHop}}            
\safemath{\kernelp}{\kernel(\time,\time')}    
\safemath{\tvir}{\rndh_{\CHop}}                
\safemath{\tvirp}{\tvir(\time,\delay)}        
\safemath{\tvirc}{\conj{\rndh}_{\CHop}}        
\safemath{\tvtf}{\rndl_{\CHop}}                
\safemath{\tvtfp}{\tvtf(\time,\freq)}            
\safemath{\tvtfc}{\conj{\rndl}_{\CHop}}        
\safemath{\spf}{\rnds_{\CHop}}                
\safemath{\spfp}{\spf(\doppler,\delay)}        
\safemath{\spfc}{\conj{\rnds}_{\CHop}}        
\safemath{\bff}{\rndb_{\CHop}}                
\safemath{\bffp}{\bff(\doppler,\freq)}        

\safemath{\irc}{\scr_{\rndh}}                
\safemath{\tfc}{\scr_{\rndl}}                
\safemath{\spc}{\scr_{\rnds}}                
\safemath{\bfc}{\scr_{\rndb}}                

\safemath{\scaf}{\scc_{\rnds}}                
\safemath{\scafp}{\scaf(\doppler,\delay)}        
\safemath{\ccf}{\scc_{\rndl}}                
\safemath{\ccfp}{\ccf(\Dtime,\Dfreq)}            
\safemath{\cic}{\scc_{\rndh}}                
\safemath{\cicp}{\cic(\Dtime,\delay)}            

\safemath{\mi}{\scI}                            
\safemath{\capacity}{\scC}                    

\DeclareMathOperator{\Prob}{\opP}        

\safemath{\normal}{\mathcal{N}}            
\safemath{\jpg}{\mathcal{CN}}            
\safemath{\uniform}{\mathcal{U}}                
\safemath{\mchain}{\leftrightarrow}        

\safemath{\dB}{\,\mathrm{dB}}
\safemath{\dBm}{\,\mathrm{dBm}}
\safemath{\Hz}{\,\mathrm{Hz}}
\safemath{\kHz}{\,\mathrm{kHz}}
\safemath{\MHz}{\,\mathrm{MHz}}
\safemath{\GHz}{\,\mathrm{GHz}}
\safemath{\THz}{\,\mathrm{THz}}
\safemath{\s}{\,\mathrm{s}}
\safemath{\ms}{\,\mathrm{ms}}
\safemath{\mus}{\,\mathrm{\text{\textmu}s}}
\safemath{\ns}{\,\mathrm{ns}}
\safemath{\ps}{\,\mathrm{ps}}
\safemath{\meter}{\,\mathrm{m}}
\safemath{\mm}{\,\mathrm{mm}}
\safemath{\cm}{\,\mathrm{cm}}
\safemath{\m}{\,\mathrm{m}}
\safemath{\W}{\,\mathrm{W}}
\safemath{\mW}{\, \mathrm{mW}}
\safemath{\J}{\,\mathrm{J}}
\safemath{\K}{\,\mathrm{K}}
\safemath{\bit}{\,\mathrm{bit}}
\safemath{\nat}{\,\mathrm{nat}}

\safemath{\define}{\triangleq}                    

\providecommand{\inner}[2]{\ensuremath{\left\langle#1,#2\right\rangle}}
\safemath{\equivalent}{\sim}
\safemath{\distas}{\sim}                    
\safemath{\sdiff}{\Delta}                
\safemath{\setdiff}{\setminus}                

\safemath{\reals}{\mathbb R}
\safemath{\positivereals}{\reals^{+}}
\safemath{\integers}{\mathbb Z}
\safemath{\posint}{\integers^{+}}
\safemath{\naturals}{\mathbb N}
\safemath{\posnaturals}{\naturals^{+}}
\safemath{\complexset}{\mathbb C}
\safemath{\rationals}{\mathbb Q}

\newcommand{\natseg}[2]{[#1 \text{\phantom{\tiny{.}}:\phantom{\tiny{.}}} #2]} 

\DeclarePairedDelimiter\floor{\lfloor}{\rfloor}

\safemath{\iSet}{\setI}
\safemath{\rel}{\bowtie}                    
\safemath{\eqrel}{\sim}                     
\safemath{\rlord}{\leq}                     
\safemath{\slord}{<}                        
\safemath{\rpord}{\preceq}                  
\safemath{\rrpord}{\succeq}                 
\safemath{\spord}{\prec}                    
\safemath{\sig}{\sigma}                     
\safemath{\metric}{d}                       

\safemath{\setfun}{\Phi}                    
\safemath{\measure}{\mu}                    
\safemath{\altmeasure}{\lambda}             
\newcommand{\outerm}[1]{#1^{\star}}         
\newcommand{\innerm}[1]{#1_{\star}}         
\safemath{\omeasure}{\outerm{\measure}}     
\safemath{\imeasure}{\innerm{\measure}}     
\safemath{\aecol}{\colS^{\star}_{\measure}} 
\safemath{\emeasure}{\bar{\measure}_{0}}    
\safemath{\rmeasure}{\tilde{\measure}}      
\safemath{\bmeasure}{\measure_{0}}          
\safemath{\glength}{\measure_{\altfun}}     
\safemath{\lebmea}{\lambda}                 
\safemath{\blebmea}{\lebmea_{0}}            
\safemath{\sfun}{s}                         
\safemath{\absintspace}{\colL^{1}}          
\safemath{\sqintspace}{\colL^{2}}           
\safemath{\abssumspace}{l^{1}}              
\safemath{\sqsumspace}{l^{2}}               

\safemath{\sfield}{\setF}                   
\safemath{\vectors}{\setV}                  
\safemath{\vecspace}{(\vectors,\sfield)}    
\safemath{\linspace}{\setV}                 
\safemath{\clinspace}{(\linspace,\sfield)}  
\safemath{\nspace}{\setU}                   
\safemath{\metspace}{\setM}                 
\safemath{\bspace}{\setB}                   
\safemath{\ipspace}{\setG}                  
\safemath{\hilspace}{\setH}                 
\safemath{\blospace}{\setG}                 
\safemath{\lop}{\opT}                       
\safemath{\altlop}{\opS}                    
\safemath{\nullsp}{\nullspace(\lop)}        
\safemath{\lfun}{l}                         
\safemath{\altlfun}{g}                      
\newcommand{\dual}[1]{#1^{'}}               
\safemath{\dsum}{\oplus}                    
\safemath{\funspace}{\colL}                 
\renewcommand{\adj}[1]{#1^{\times}}         
\safemath{\adjlop}{\adj{\lop}}              
\safemath{\hadjlop}{\hadj{\lop}}            
\safemath{\tow}{\xrightarrow{w}}            
\safemath{\tows}{\xrightarrow{w^{*}}}       
\safemath{\cparam}{\lambda}
\safemath{\lopl}{\lop_{\cparam}}
\safemath{\iop}{\opI}                       
\safemath{\resolop}{\opR}                   
\safemath{\resolvent}{\resolop_{\cparam}(\lop)}    
\safemath{\reset}{\setQ}
\safemath{\spectrum}{\setS}
\safemath{\resolset}{\reset(\lop)}          
\safemath{\lopspec}{\spectrum(\lop)}        
\safemath{\pspec}{\spectrum_{p}(\lop)}      
\safemath{\cspec}{\spectrum_{c}(\lop)}      
\safemath{\rspec}{\spectrum_{r}(\lop)}      
\safemath{\ev}{\cparam}                     
\newcommand{\specrad}[1]{r_{#1}}            
\safemath{\lopsrad}{\specrad{\lop}}         
\safemath{\pop}{\opP}                       

\safemath{\specfam}{\colE}                  
\safemath{\specop}{\opE_{\cparam}}          
\safemath{\altspecop}{\opE_{\mu}}           
\safemath{\vmulti}{\vecone}                 
\safemath{\unitaryop}{\opU}                 
\safemath{\sval}{\sigma}                    
\safemath{\corrcoef}{\rho}                  
\safemath{\sangle}{\theta}                  

\let\time\undefined
\safemath{\iset}{\setI}                     
\safemath{\shift}{\nu}
\safemath{\scale}{\alpha}
\safemath{\time}{t}
\safemath{\specfreq}{\theta}
\newcommand{\transopgen}[1]{\opT_{#1}}      
\safemath{\transop}{\transopgen{\delay}}
\newcommand{\modopgen}[1]{\opM_{#1}}        
\safemath{\modop}{\modopgen{\shift}}
\newcommand{\dilopgen}[1]{\opD_{#1}}        
\safemath{\dilop}{\dilopgen{\scale}}
\safemath{\fram}{\setF}                     
\safemath{\dfram}{\dual{\fram}}             
\safemath{\ufb}{B}                          
\safemath{\lfb}{A}                          
\safemath{\sop}{\hadj{\aop}}                
\safemath{\aop}{\opT}                       
\safemath{\fop}{\opS}                       
\safemath{\daop}{\tilde\opT}                
\safemath{\dsop}{\hadj{\tilde\opT}}         

\safemath{\ifop}{\inv{\fop}}                
\safemath{\rifop}{\fop^{-1/2}}              
\safemath{\transeq}{\setT}                  
\safemath{\nfun}{\Phi}                      
\safemath{\funvec}{\vecf}                   
\safemath{\altfunvec}{\vecg}

\safemath{\samplespace}{\Omega}
\safemath{\probspace}{(\samplespace,\sfield,\Prob)}    
\safemath{\ccoef}{\rho}            

\safemath{\infstate}{\vecpi}                 
\safemath{\typset}{\setA_{\epsilon}^{(N)}}   
\safemath{\expequal}{\doteq}                 
\safemath{\code}{C}                          
\safemath{\dstringset}{\setD^{\star}}        
\safemath{\cwlength}{l}                      
\safemath{\elength}{L}                       
\safemath{\extension}{C^{\star}}             
\safemath{\approaches}{\rightarrow}          
\safemath{\evnt}{\setA}                      
\safemath{\altevnt}{\setB}                   
\safemath{\rv}{\rndx}                        
\safemath{\altrv}{\rndy}                     
\safemath{\complexrv}{\rndu}                 
\safemath{\altcrv}{\rndv}                    
\safemath{\rvec}{\rvecx}                     
\safemath{\altrvec}{\rvecy}                  
\safemath{\crvec}{\rvecu}                    
\safemath{\altcrvec}{\rvecv}                 
\safemath{\variance}{\sigma^{2}}             
\safemath{\map}{T}                           
\safemath{\jacobian}{\matJ}                  
\safemath{\wvec}{\rvecw}                     
\safemath{\wrv}{\rndw}                       
\safemath{\orthmat}{\matQ}                   
\safemath{\evmat}{\matLambda}                
\safemath{\identity}{\matidentity}           
\safemath{\innovec}{\vecv}                   
\safemath{\convas}{\xrightarrow{\text{a.s.}}}
\safemath{\convr}{\xrightarrow{\text{r}}}    
\safemath{\convp}{\xrightarrow{\text{P}}}    
\safemath{\convd}{\xrightarrow{\text{D}}}    
\safemath{\ltis}{\opL}                       
\safemath{\ir}{h}                            
\safemath{\tf}{\MakeUppercase{\ir}}          

\newtheorem{thm}{Theorem}
\newtheorem{lem}{Lemma}

\newtheorem{dfn}{Definition}
\newtheorem{prp}{Proposition}

\newcommand*{\fancyrefparlabelprefix}{par}      
\newcommand*{\fancyrefchalabelprefix}{cha}      
\newcommand*{\fancyrefapplabelprefix}{app}      
\newcommand*{\fancyrefthmlabelprefix}{thm}      
\newcommand*{\fancyreflemlabelprefix}{lem}      
\newcommand*{\fancyrefcorlabelprefix}{cor}      
\newcommand*{\fancyrefdeflabelprefix}{def}      
\newcommand*{\fancyrefproplabelprefix}{prop}    
\newcommand*{\fancyrefprplabelprefix}{prp}      

\frefformat{vario}{\fancyrefparlabelprefix}{Part~#1}
\frefformat{vario}{\fancyrefchalabelprefix}{Chapter~#1}
\frefformat{vario}{\fancyrefseclabelprefix}{Section~#1}
\frefformat{vario}{\fancyrefthmlabelprefix}{Theorem~#1}
\frefformat{vario}{\fancyreflemlabelprefix}{Lemma~#1}
\frefformat{vario}{\fancyrefcorlabelprefix}{Corollary~#1}
\frefformat{vario}{\fancyrefdeflabelprefix}{Definition~#1}
\frefformat{vario}{\fancyreffiglabelprefix}{Figure~#1}
\frefformat{vario}{\fancyrefapplabelprefix}{Appendix~#1}
\frefformat{vario}{\fancyrefeqlabelprefix}{(#1)}
\frefformat{vario}{\fancyrefproplabelprefix}{Property~#1}
\frefformat{vario}{\fancyrefprplabelprefix}{Proposition~#1}


%% file: defs.tex
\newcommand{\inp}{\vecx}                
\newcommand{\data}{\vecs}               
\newcommand{\noise}{\vecz}              

\newcommand{\N}{N}                      
\newcommand{\sparsity}{S}               
\newcommand{\rr}{r}                     
\newcommand{\hrr}{{\hat r}}             
\newcommand{\npoints}{V}                

\newcommand{\T}{\mathbb{T}}         
\newcommand{\support}{\mathrm{supp}}    

\newcommand{\lo}{{\mathrm{lo}}}         
\newcommand{\hi}{{\mathrm{hi}}}         

\newcommand{\near}[2]{\setN(#1,#2)}     
\newcommand{\nearhi}[1]{\setN(\lhi, #1)}  
\newcommand{\nearh}{\setN_\hi}          
\newcommand{\far}[2]{\setF(#1,#2)}      
\newcommand{\farh}{\setF_\hi}           

\newcommand{\khi}{k_\hi}                
\newcommand{\klo}{k_\lo}                

\newcommand{\flo}{f_\lo}                
\newcommand{\fhi}{f_\hi}                
\newcommand{\llo}{\lambda_\lo}          
\newcommand{\lhi}{\lambda_\hi}          
\newcommand{\fc}{f_c}                   
\newcommand{\lc}{\lambda_c}             

\newcommand{\rset}[2]{\setR(#1,#2)}     

\newcommand{\fval}{\vecf}               
\newcommand{\dval}{\vecd}               
\newcommand{\fvalc}{f}                  
\newcommand{\dvalc}{d}                  
\newcommand{\I}[1]{{I[#1]}}             

\newcommand{\SRF}{\mathrm{SRF}}         
\newcommand{\DSRF}{\mathrm{DSRF}}       

\newcommand{\D}{\Delta}                 

\newcommand{\fejer}[1]{g_{#1}}          

\newcommand{\vt}{v^\tau}
\newcommand{\vtn}{v^\tau}

\newcommand{\cc}{c}                     
\newcommand{\tcc}{\tilde c}             
\newcommand{\CC}[1]{C(#1)}              
\newcommand{\tCC}[1]{\tilde C(#1)}      

\newcommand{\cgsum}{c_{s0}}             
\newcommand{\cgsumd}{c_{s1}}            

\newcommand{\cabshid}{c_k'}             
\newcommand{\cabshidd}{c_k''}           
\newcommand{\ckersum}{{\tilde c}_0}     
\newcommand{\ckersumd}{{\tilde c}_1}    
\newcommand{\ckersumdd}{{\tilde c}_2}   
\newcommand{\cDzinv}{{\bar c}_0}        
\newcommand{\cDtinv}{{\bar c}_2}        
\newcommand{\cEinv}{{\bar c}_E}         

\newcommand{\cFinv}{{\bar c}_F}         
\newcommand{\ca}{c_\alpha}              
\newcommand{\cb}{c_\beta}               
\newcommand{\cg}{c_g}                   
\newcommand{\cgz}{c_g^0}                
\newcommand{\cgd}{c_g^1}                

\newcommand{\cql}{c_l}
\newcommand{\cqlf}{c_{l1}}
\newcommand{\cz}{c_{l2}}
\newcommand{\czz}{c_{l3}}
\newcommand{\cqu}{c_u}
\newcommand{\cqdvd}{c_{u1}}
\newcommand{\cqdvdd}{c_{u2}}
\newcommand{\cqdv}{c_{u0}}
\newcommand{\czdu}{c_{u3}}

\newcommand{\czdduu}{c_{u5}}
\newcommand{\czdduuu}{c_{u6}}
\newcommand{\czdduuuu}{c_{u7}}
\newcommand{\czdduf}{c_{u8}}
\newcommand{\czdduff}{c_{u9}}
\newcommand{\czddufff}{c_{u10}}
\newcommand{\czdduffff}{c_{u11}}
\newcommand{\czdduffd}{c_{u12}}
\newcommand{\czddufffdd}{c_{u13}}
\newcommand{\czddufffddd}{c_{u14}}
\newcommand{\czddufffdddd}{c_{u15}}
\newcommand{\cuos}{c_{u16}}
\newcommand{\cuoss}{c_{u17}}
\newcommand{\cuosss}{c_{u18}}
\newcommand{\cuossss}{c_{u19}}
\newcommand{\cuosssss}{c_{u20}}
\newcommand{\cuossssss}{c_{u21}}
\newcommand{\cuuu}{c_{u22}}
\newcommand{\cuuuu}{c_{u23}}
\newcommand{\cuuuuu}{c_{u24}}
\newcommand{\cuuuuuu}{c_{u25}}
\newcommand{\cuf}{c_{u26}}
\newcommand{\cuff}{c_{u27}}
\newcommand{\cufff}{c_{u28}}
\newcommand{\cuffff}{c_{u29}}
\newcommand{\cufffff}{c_{u30}}
\newcommand{\cufh}{c_{u31}}

\newcommand{\cufk}{c_{u34}}
\newcommand{\cufl}{c_{u35}}
\newcommand{\cufm}{c_{u36}}
\newcommand{\cufp}{c_{u37}}
\newcommand{\cufq}{c_{u38}}
\newcommand{\cufr}{c_{u39}}
\newcommand{\cufs}{c_{u40}}
\newcommand{\cuft}{c_{u41}}
\newcommand{\cufv}{c_{u42}}
\newcommand{\cufw}{c_{u43}}
\newcommand{\cufx}{c_{u44}}
\newcommand{\cufy}{c_{u45}}
\newcommand{\cufz}{c_{u46}}
\newcommand{\cuga}{c_{u47}}
\newcommand{\cugb}{c_{u48}}
\newcommand{\cugc}{c_{u49}}
\newcommand{\cugg}{c_{u50}}
\newcommand{\cugf}{c_{u51}}
\newcommand{\cugh}{c_{u52}}
\newcommand{\cugk}{c_{u53}}
\newcommand{\cugl}{c_{u54}}
\newcommand{\cugm}{c_{u55}}
\newcommand{\cugn}{c_{u56}}
\newcommand{\cugo}{c_{u57}}
\newcommand{\cugp}{c_{u58}}